\documentclass[a4paper,12pt]{article}

\usepackage[explicit]{titlesec}
\usepackage{hyphenat}
\usepackage{ragged2e}

\usepackage[T1]{fontenc}
\usepackage{amsmath, amsthm, amssymb}
\usepackage{graphicx}
\usepackage{graphics}

\usepackage{appendix}
\usepackage{float}
\usepackage{caption}
\usepackage{booktabs}
\usepackage{subcaption}
\usepackage{verbatim}

\usepackage{soul}
\setul{.6pt}{.4pt}

\usepackage{geometry}
\usepackage{fancyhdr}
\usepackage{pdflscape}

\usepackage{color}

\newcommand{\blue}[1]{{\color{blue} #1}}

\numberwithin{equation}{section}
\newtheorem{theorem}{\bf \blue{Theorem}}[section]

\usepackage[labelfont={footnotesize,bf} , textfont=footnotesize]{caption}
\titleformat{\section}
  {\normalfont}{\thesection}{1em}{\MakeUppercase{\textbf{#1}}}
\titlespacing\section{0pt}{0pt}{-10pt}
\titleformat{\subsection}
  {\normalfont}{\thesubsection}{1em}{\textit{#1}}
\titlespacing\subsection{0pt}{0pt}{-8pt}

\makeatletter
\newcommand\sixteen{\@setfontsize\sixteen{17pt}{6}}
\renewcommand{\maketitle}{\bgroup\setlength{\parindent}{0pt}
\begin{flushleft}
\sixteen\bfseries \@title
\medskip
\end{flushleft}
\textit{\@author}
\egroup}
\makeatother

{\title{\begin{center}Crucial Inflammatory Mediators and Efficacy of Drug Interventions in Pneumonia Inflated COVID-19:\vspace{.15cm} \\ An Invivo Mathematical Modelling Study\end{center}}

\author{
Bishal Chhetri$^{a}$, D. K. K. Vamsi$^{*a}$, Vijay M. Bhagat$^{b}$, Ananth V. S.$^{a}$, Bhanu Prakash$^{a}$, Roshan Mandale$^{c}$, Swapna Muthusamy$^{b}$, Carani B Sanjeevi$^{d, e}$\\\\ \medskip 
$^{a}$Department of Mathematics and Computer Science, Sri Sathya Sai Institute of Higher Learning - \\ SSSIHL, India  \\ 
$^{b}$Central Leprosy Teaching and Research Institute - CLTRI, Chennai, India\\ 
$^{c}$ Sri Sathya Sai Higher Secondary School - SSSHSS, Puttaparthi, India\\ \medskip 
$^{d}$ Vice Chancellor, Sri Sathya Sai Institute of Higher Learning, India.\\
$^{e}$ Department of Medicine, Karolinska Institute, Stockholm, Sweden\\
$*_{ {\textit{Corresponding  \ Author}}}$\\ \medskip 
bishalchhetri@sssihl.edu.in,  dkkvamsi@sssihl.edu.in*, vijaydr100@gmail.com, ananthvs@sssihl.edu.in, \\ prakashdmacs@gmail.com, roshanmandle1996@gmail.com,    swapnamuthuswamy@gmail.com, sanjeevi.carani@sssihl.edu.in, sanjeevi.carani@ki.se \\
}

\begin{document}
\numberwithin{equation}{section}
\vspace*{.01 in}
\maketitle
\vspace{.12 in}

\section*{abstract} \vspace{.25cm}
The virus {\emph{SARS-COV-2}} caused disease {\emph{COVID-19}} has been declared a pandemic by WHO. Currently, over 210 countries and territories have been affected. Careful, well-designed drugs and vaccine for the total elimination of this virus seem to be the need of the hour. In this context, the invivo mathematical modelling studies can be extremely helpful in understanding the efficacy of the drug interventions. These studies can also help understand the role of the crucial inflammatory mediators and the behaviour of immune response towards this novel coronavirus. Motivated by these facts, in this paper, we study the invivo dynamics of Covid-19. We initially model and study the natural history, the course of the infection and its dynamics. We then validate the model by generating two-parameter heat plots that represent the characteristics of Covid-19. We also do the sensitivity analysis to identify the sensitive parameters of the system. Lastly, we study the efficacy of drug interventions for Covid-19 by formulating an Optimal Control Problem. The outcomes of these studies are multi-fold. The system admits two steady states:  the disease-free equilibrium and the infected equilibrium. The dynamics of the system show that the disease takes its course to one of these steady states based on the reproduction number $R_0$. The system undergoes a transcritical bifurcation at $R_0 = 1$. From the sensitivity analysis, it is seen that the burst rate of the virus particles and the natural death rate of the virus are the sensitive parameters of the system. Results from the optimal control studies suggest that the antiviral drugs that target viral replication and the drugs that enhance the immune system response both reduce the infected cells and viral load when taken individually. However, it is observed that these drugs yield the best possible results when administered together. Hence, it is concluded that the optimal control strategy would be to use the combination of both these drugs which not only help in patient's recovery  but also reduce the side effects caused to the patient because of the  minimal/optimal dosage administered. The results obtained here are inline with some of the clinical findings for Covid-19.  This invivo modelling study involving the crucial biomarkers of Covid-19 is the first of its kind and the results obtained from this can be helpful to researchers, epidemiologists, clinicians and doctors who are working in this field.

\section*{keywords} \vspace{.25cm}
Covid-19; Epithelial Cells; Pneumocytes; Inflammatory Mediators; Cytokines; Chemokines;  Drug Interventions; Optimal Control Problems;
\vspace{.12 in}


\section{introduction} \vspace{.25cm}

\qquad The Coronavirus caused disease COVID-19 has been declared a pandemic by WHO. Currently, over 210 countries and territories have been affected. As on 02nd May 2020, 2, 37, 996 people have lost their lives and more than 33 lakh people have been affected due to Covid-19 all over the world \cite{1}.  This pandemic has catalysed the development of novel coronavirus vaccines across pharmaceutical companies and research organizations. Drugs such as remdesivir, favipiravir, ivermectin, lopinavir/ritonavir, mRNA-1273, phase I trial (NCT04280224) and AVT technology are being used as therapeutic agents by different countries for treating Covid-19  \cite{10,8,9,tu2020review}. \vspace{.25cm}

 \qquad In this context, the invivo mathematical modelling studies can be extremely helpful in understanding the efficacy of the drug/medicine administered. These studies can also help in understanding the behaviour of cytokines, chemokines and immune system response of the body towards this virus  \cite{2,3}. The results of these studies can suggest the optimal drug regime for treating Covid-19. 

   \vspace{.25cm}


 \qquad Motivated by the above observations, in this paper, we propose to model and study the invivo dynamics of Covid-19 with and without drug interventions to understand the efficacy of drugs administered. We study these problems as Optimal Control Problems. In the recent work  \cite{ea2020host}, an in-host modelling study deals with the qualitative characteristics and estimation of standard parameters of corona viral infections. Some of the mathematical models that deal with transmission and spread of COVID-19 at the population level can be found in  \cite{5, 7, 4, 6}.  Modelling the  Invivo dynamics of Covid-19 involving the crucial biomarkers, which is being attempted here is the first of its kind for Covid-19. \vspace{.25cm}

  \qquad The section-wise split-up of the paper is as follows: In the next section, we discuss the pathogenesis of Covid- 19 and formulate the mathematical model dealing with natural history. Further, we discuss the antiviral drug interventions that are being tried across the nations and then we develop the invivo model incorporating these interventions. In the later sections, we discuss the disease dynamics, we frame the optimal control problem and do the optimal control studies. Finally, we present the discussions and conclusions followed by a few pointers to future research in this direction. \vspace{.25cm} 

\section{mathematical models formulation} \vspace{.25cm}

{\flushleft{  \textbf{PATHOGENSIS OF COVID-19} }}\vspace{.25cm}

\qquad On Feb. 11, 2020 World Health Organization named novel corona viral pneumonia induced as Coronavirus disease(COVID-19), which is caused by Severe Acute Respiratory Syndrome Coronavirus-2 (SARS-CoV-2). \vspace{.25cm}

\qquad Out of four coronavirus genera($\alpha,\beta,\gamma,\delta$), $\beta$ CoV strain showed 88\% identity of genetic sequence with two bat derived SARS corona viruses (bat-SL-CoVZXC45,bat-SL-CoVZXC21) and 50\% with virus causing Middle East Respiratory Syndrome(MERSCoV) \cite{p1}. Therefore the functional mechanism in pathogenesis is also quite resembling with SARS CoV and MERSCoV. \vspace{.25cm}

\qquad Human to human transmission of SARS-CoV-2 occurs either through droplet infection or direct contact from an infected person. Transmission from asymptomatics and through faeco-oral route are also reported. After the entry of virus through droplet infection or through contact transmission from hands to the mucous membrane of mouth, nose or eyes. It is reported that Angiotensin Converting Enzyme-2 (ACE-2) plays crucial role in providing binding site to the viral structural protein(Spike protein-S) on the cell surface and subsequent entry into the host cell \cite{p6}. The SARS-CoV-2 have much higher affinity towards ACE-2 receptors as compared with its previous conterparts(SARS-CoV and MERS-CoV). The ACE-2 expresses in lung alveolar Type-2 cells(AT2), liver cholangiocytes, colon colonocytes, esophagus, keratinocytes, endothelial cells of ileum, rectum and stomach, and proximal tubules of kidneys. AT2 secret surfactant, which reduces surface tension preventing collapsing of the alveoli and playing crucial role in oxygen diffusion across lungs and blood vessels.Viral antigens are presented by antigen presenting cells(APC) which are Human Leucocyte Antigen(HLA) cytotoxic T cell(also known as T-killer cell, cytotoxic T-lymphocyte, CD8+) \cite{p2}. The virus enters the cell with fusion of its membranes with host cell and begins transcription with ssRNA acting as template. Synthesis of the viral proteins takes place in the cytoplasm of the pneumocytes, new virus is released by budding and ready to infect new cell which is confirmed by presence of abundent viral antigens in the cytoplasm of the pneumocytes(AT2) in case of SARS-CoV \cite{p5}. Viremia (viral particle in the blood/serum) was also noticed by some authors along with very high levels of IL-6 especially among severe cases of Covid-19 illness leading to increase vascular permeability and impairment of organs \cite{p4}. Acute Respiratory Distress Syndrome(ARDS) is the common immunopathological event for all the aforesaid corona viral diseases.One of the main mechanisms in causation of ARDS is cytokine storm,deadly uncontrolled systemic inflammatory response due to the release of large quantity of pro-inflammatory cytokines and chemokines by immune effector cells. These cytokines are identified as  IFN-$\alpha$, IFN-$\gamma$, IL-1b, IL-6, IL-2, IL-18, IL-33, TNF-$\alpha$, TGF-$\beta$ and chemokines as CCL-2, CCL-3, CCL-5, CXCL-8, CXCL-9, CXCL-10. Severe infections correlated high levels of IL-6, IFN $\alpha$, CCL-5, CXCL-8 and CXCL-10 \cite{p4,p2}. Also the viral load was noted to be crucial in determining severity of the disease and strongly correlated with the lung injury Murray score \cite{p3}. The cytokine storm is a violent attack by the immune system causing ARDS, multi-organ failure and eventually death \cite{p2}.  \vspace{.25cm}

\qquad From the above discussed pathogenesis  it can be understood that the study of both the epithelial (pneumocytes) cells (including both the healthy and infected) and virus population levels and their changes over the time due to inflammatory mediators play a crucial in understanding the dynamics of pneumonia inflated  COVID-19. \vspace{.25cm}

\qquad Motivated by these  we first consider the following mathematical  model.
   \vspace{.25cm}
   {\flushleft \bf{Model 1 : Model without Interventions/Medication }}
   
    \begin{eqnarray}
   	\frac{dS}{dt}& =&  \omega \ - \beta SV  - \mu S  \label{sec2equ1} \\
   	\frac{dI}{dt} &=& \beta SV \ -  { \bigg(d_{1}  + d_{2}  + d_{3} +  d_{4} + d_{5}+ d_{6}\bigg)I }  \ - \mu I   \label{sec2equ2}\\ 
   	\frac{dV}{dt} &=&  \alpha I   \ -  \bigg( b_{1}  + b_{2}  + b_{3}  +  b_{4}+ b_{5} + b_{6}\bigg)V    \ -  \mu_{1} V \label{sec2equ3}
   \end{eqnarray} 
   \newpage
  
     \begin{table}[ht!]
     	
     	\centering 
     	\begin{tabular}{|l|l|} 
     		\hline\hline
     		
     		\textbf{Parameters} &  \textbf{Biological Meaning} \\  
     		\hline\hline 
     		$S$ & Healthy Type II Pneumocytes  \\
     		\hline\hline
     		
     		$I$ & Infected Type II Pneumocytes  \\
     		\hline\hline
     		$\omega$ & Natural birth rate of Type II Pneumocytes \\
     		\hline\hline
     		$V$ & Viral load  \\
     		\hline\hline
     		$\beta$ & Rate at which healthy Pneumocytes are infected  \\
     		
     		\hline\hline
     		$\alpha$ & Burst rate of virus particles \\
     		\hline\hline
     		$\mu$ & Natural death rate of Type II Pneumocytes \\
     		\hline\hline
     		$\mu_{1}$ & Natural death rate of virus \\
     		\hline\hline
     		
     		$d_{1}, \hspace{.25cm} d_{2}, \hspace{.25cm} d_{3}, \hspace{.25cm} d_{4}, \hspace{.25cm} d_{5}, \hspace{.25cm} d_{6}$ & Rates at which Infected Pneumocytes are removed because\\
     	& the release of cytokines and chemokines  IL-6\\
     	&  TNF-$\alpha$, \hspace{.2cm}INF-$\alpha$,  \hspace{.2cm}CCL5, \hspace{.2cm}CXCL8 , \hspace{.2cm}CXCL10   \hspace{.2cm} respectively   \\ 
     	\hline\hline
     		
     		$b_{1}, \hspace{.25cm} b_{2}, \hspace{.25cm} b_{3}, \hspace{.25cm} b_{4}, \hspace{.25cm} b_{5}, \hspace{.25cm} b_{6}$ & Rates at which Virus is removed because of\\
     		& the release of cytokines and chemokines  IL-6\\
     		&  TNF-$\alpha$, \hspace{.2cm}INF-$\alpha$,  \hspace{.2cm}CCL5, \hspace{.2cm}CXCL8 , \hspace{.2cm}CXCL10   \hspace{.2cm} respectively   \\ 
     		\hline\hline
     	$	u_2(t) $ & Rate at which virus replication/birth is decreased due to medication\\
     		\hline\hline

     	\end{tabular}
     \end{table} \vspace{.25cm}
   
   \qquad  The purpose of drug interventions can be two fold, the first to target the virus replication cycle and the second based on immunotherapy approaches either aimed to boost innate antiviral immune responses or alleviate damage induced by dysregulated inflammatory responses.  Based on this the therapeutic agents for virus infections can be divided into two categories each serving the designated purpose \cite{tu2020review}.
   
     \qquad Drugs such as {\emph{remdesivir, favipiravir}} inhibit RNA-dependent RNA polymerase and drugs  {\emph{ivermectin, lopinavir/ritonavir}}  inhibit the viral protease there by reducing the viral replication. On the other hand clinical trials such as phase I trial (NCT04280224) in China aim to enhance the innate immune system by increasing  the production of cytokines and chemokines with the end goal of increasing NK cells \cite{tu2020review}.   \vspace{.25cm}
   
     \qquad Motivated by the above clinical findings we consider an  control problem with the following drug interventions as controls.
      \begin{eqnarray*}
      	u_{11}(t) &=& d_1(t) + d_1(t) + d_2(t) + d_3(t) + d_4(t) + d_5(t) +  d_6(t) \\*
      	u_{12}(t) &=& b_1(t) + b_1(t) + b_2(t) + b_3(t) +  b_4(t) + b_5(t) +  b_6(t) \\*
      	&and&  u_2(t)
      \end{eqnarray*}
       \qquad Here the controls $u_{11}$ and $u_{12}$ incorporate the effect of drug interventions which enhance the innate immune response which in turn leads to decrease in the infected population and viral load. The control $u_2$ incorporates the effect of drug interventions which prevent viral replication thereby reducing the virus birth rate.   
   \newpage
   {\flushleft \bf{Model 2 :  Model with Interventions/Medication as Controls }}
   
   \begin{eqnarray}
   	\frac{dS}{dt} &=&   \omega \  - \beta SV \ - \mu S \label{sec2eqn4} \\
   	\frac{dI}{dt} &=& \beta SV \ - \bigg(\bigg(d_{1}(t)  + d_{2}(t)  + d_{3}(t) +  d_{4}(t) + d_{5}(t) + d_{6}(t)\bigg) = u_{11}(t)\bigg)I  \  - \mu I \label{sec2eqn5}\\ 
   	 	\frac{dV}{dt} &=& (\alpha - u_2(t))I \ - \bigg(\bigg( b_{1}(t)  + b_{2}(t)  + b_{3}(t)  +  b_{4}(t)+ b_{5}(t) + b_{6}(t) \bigg) = u_{12}(t) \bigg)V 
   	 	- \mu_{1} V  \label{sec2eqn6}
  \end{eqnarray}

\section{natural history, the course of the infection and its dynamics} \vspace{.25cm}
In this section we consider model 1 which deals with the natural history and course of infection. We study its dynamics.

{\flushleft{  \textbf{POSITIVITY AND BOUNDEDNESS} }}\vspace{.25cm}

The positivity and boundedness of the solutions of the model  is the fundamental thing which needs to be established before doing any other analysis. \vspace{.25cm}

\underline{\textbf{Positivity}}\textbf{:}
 We now show that if the initial conditions of the system (\ref{sec2equ1}-\ref{sec2equ3}) are positive, then the solution remain positive for any future time. Using the  equations (\ref{sec2equ1}-\ref{sec2equ3}), we get,
\begin{align*}
\frac{dS}{dt} \bigg|_{S=0} &= \omega \geq 0 ,  &  
\frac{dI}{dt} \bigg|_{I=0} &= \beta S V  \geq 0,
\\ \\
\frac{dV}{dt} \bigg|_{V=0} &= \alpha I \geq 0.  &
\end{align*}

\vspace{1.5mm}
\noindent
\\ 
Thus all the above rates are non-negative on the bounding planes (given by $S=0$, $I=0$, and $V=0$) of the non-negative region of the real space. So, if a solution begins in the interior of this region, it will remain inside it throughout time $t$. This  happens because the direction of the vector field is always in the inward direction on the bounding planes as indicated by the above inequalities. Hence, we conclude that all the solutions of the the system (\ref{sec2equ1}-\ref{sec2equ3}) remain positive for any time $t>0$  provided that the initial conditions are positive. This establishes the positivity of the solutions of the system (\ref{sec2equ1}-\ref{sec2equ3}). Next we will show that the solution is bounded.  \vspace{.25cm}

\underline{\textbf{Boundedness}}\textbf{:}
Let  $N(t) = S(t)+I(t)+V(t) $ \\

Let $x =d_{1}+d_{2}+d_{3}+d_{4}+d_{5}+d_{6} $
and $y=b_{1}+b_{2}+b_{3}+b_{4}+b_{5}+b_{6}$\\

Now,  
\begin{equation*}
\begin{split}
\frac{dN}{dt} & = \frac{dS}{dt} +  \frac{dI}{dt}+  \frac{dV}{dt}  \\[4pt]
& = \omega -\mu(S+I)-\mu_{1}V-(x-\alpha)I-yV \\
& \le \omega -\mu(S+I+V)  \\
\end{split}
\end{equation*}
with the assumption that $x > \alpha$ and $\mu=\mu_{1}.$ \\ 
\noindent
Here the integrating factor is $e^{\mu t}.$ Therefore after integration we get, \\

$N(t)\le \frac{\omega}{\mu} + ce^{-\mu t}.$ Now  as $t \rightarrow \infty$ we get, \\

$$N(t)\le \frac{\omega}{\mu}$$

Thus we have shown that the system (\ref{sec2equ1}-\ref{sec2equ3}) is positive and bounded. Therefore the biologically feasible region is given by the following set, 
\begin{equation*}
\Omega = \bigg\{\bigg(S(t), I(t), B(t)\bigg) \in \mathbb{R}^{3}_{+} : S(t)  + I(t) + V(t) \leq \frac{\omega}{\mu}, \ t \geq 0 \bigg\}
\end{equation*}
{\flushleft{  \textbf{EQUILIBRIUM POINTS AND REPRODUCTION NUMBER ($R_{0}$}) }}\vspace{.25cm}	

Model 1 has two equilibrium points namely, the infection free equilibrium $E_{0}=\bigg(\frac{\omega}{\mu},0,0 \bigg)$ and the infected equilibrium $E_{1}=(S^*, I^*, V^*),$  where,\\ 

$$\hspace{-1cm}S^*=\frac{(y+\mu_{1})(x+\mu)}{\alpha \beta}$$\\

$$I^*=\frac{\alpha \beta \omega-\mu (y+\mu_{1})(x+\mu)}{\alpha \beta (x+\mu)}$$\\

$$V^*=\frac{\alpha \beta \omega-\mu(y+\mu_{1})(x+\mu)}{ \beta (x+\mu)(y+\mu_{1})}$$\\

and $$x =d_{1}+d_{2}+d_{3}+d_{4}+d_{5}+d_{6}$$ 
$$y=b_{1}+b_{2}+b_{3}+b_{4}+b_{5}+b_{6}$$\\

The basic reproduction number is calculated using the next generation matrix method \cite{diekmann2010construction} and the expression for $R_{0}$ is given by \vspace{.5cm}\\
\begin{equation}
\mathbf{ R_{0}}= \mathbf{\frac{\beta \alpha \omega}{\mu (x+\mu) (y+\mu_{1})}} \label{sec3equ1}\\
\end{equation}

{\flushleft{  \textbf{LOCAL DYNAMICS AND GLOBAL DYNAMICS OF THE MODEL 1} }}\vspace{.25cm}

In the following we discuss the local stability analysis of the infection free equilibrium $E_{0}$ and infected equilibrium $E_{1}$. \vspace{.25cm}

{\flushleft{  \textbf{STABILITY ANALYSIS OF  $E_{0}$} }}\vspace{.25cm}

The jacobian matrix of the system (\ref{sec2equ1}-\ref{sec2equ3}) at the infection free equilibrium $E_0$ is given by, \\

\begin{equation*}
J_{E_{0}} = 
\begin{pmatrix}
-\mu & 0 & \frac{-\beta\omega}{\mu} \\
0 & -(x+\mu) & \frac{\beta\omega}{\mu} \\
0 & \alpha & -(y+\mu_{1})
\end{pmatrix}
\end{equation*} \\

The characteristics equation is given by,

\begin{equation}
\bigg(-\mu_{1}-\lambda\bigg)\bigg(\lambda^2+\bigg(x+y+\mu+\mu_{1}\bigg)\lambda+(x+\mu)(y+\mu_{1})-\frac{\beta \alpha \omega}{\mu}\bigg)=0 \label{sec3equ2}  
\end{equation}

The first root of equation (\ref{sec3equ2}) is $\lambda_{1}=-\mu_{1}$ \\

Now from the definition of $\mathbf{R_{0}}$ (\ref{sec3equ1}) we get, \vspace{.5cm}\\
\begin{eqnarray}
\beta \alpha \omega -\mu(x+\mu)(y+\mu_{1}) y=\mu(R_{0}-1)(x+\mu)(y+\mu_{1}) \label{sec3equ3}
\end{eqnarray}

Using the relation (\ref{sec3equ3}) the roots of the  quadratic part of  equation (\ref{sec3equ2}), are given by, \\

$$\lambda_{2,3}=A \pm \sqrt{A^2+ 4(R_{0}-1)D}  $$    where, \\

$$A= (x+y+\mu+\mu_{1})$$ and $$ D=(x+\mu)(\mu_{1}+ y)$$ \\

There are two cases that we need to consider here.\\

\textbf{case I: When $R_{0} < 1$}\\
\vspace{.25cm}
When $R_{0} < 1$ there are further two subcases, \\

\hspace{2cm}\textbf{(a)}: $A^2+ 4(R_{0}-1)D > 0$ \vspace{.3cm}\\

\hspace{2cm}\textbf{(b)}: $A^2+ 4(R_{0}-1)D < 0$ \\

\textbf{Sub-case (a)}: When $A^2+ 4(R_{0}-1)D > 0$, both  the eigenvalues $\lambda_{2,3}$ are negative, \\

$$\lambda_{2,3}=A \pm \sqrt{A^2+ 4(R_{0}-1)D} < 0 $$ \\

Hence all the eigen values of the characteristics equation (\ref{sec3equ2}) are negative. Therefore the infection free equilibrium point $E_{0}$ is asymptotically stable. \\

\textbf{Sub-case (b)}: When $A^2+ 4(R_{0}-1)D < 0$ the eigenvalues of the quadratic part of equation (\ref{sec3equ2}) are complex conjugates with the negative real parts. Therefore we again have $E_{0}$ to be locally asymptotically stable.\\

Hence we conclude that $E_{0}$ is locally asymptotically stable provided $R_{0} < 1$ \vspace{.25cm}

\textbf{case II: When $R_{0} > 1$}\\
For the case $R_{0}>1$, the characteristics equation (\ref{sec3equ2}) has two negative eigenvalues and one positive eigenvalue. Therefore whenever $R_{0}>1$ the infection free equilibrium $E_{0}$ is unstable.\\

{\flushleft{  \textbf{STABILITY ANALYSIS OF  $E_{1}$} }}\vspace{.25cm}

With the definition of $R_{0}$ (\ref{sec3equ1}) the infected equilibrium is given by, \\

$$E_{1}=(S^*, I^*, B^*)$$ where,\\ 

$$\hspace{.3cm}S^*=\frac{(x+\mu)(y+\mu_{1})}{\alpha \beta}$$\\

$$\hspace{.8cm}I^*=\frac{\mu(\mu_{1}+ y) \bigg(R_{0}-1\bigg)}{\alpha \beta}$$\\

$$\hspace{-.4cm}V^*=\frac{\mu \bigg(R_{0}-1\bigg)}{\beta}$$\\

Therefore the infected equilibrium exists only if $R_{0}>1$, otherwise $E_{1}$ will become negative which does not make sense.\\

The jacobian matrix of the system (\ref{sec2equ1}-\ref{sec2equ3}) is given by, \\

\begin{equation*}
J = 
\begin{pmatrix}
-\beta V-\mu & 0 & -\beta S \\
\beta V & -(x+\mu) & \beta S \\
0 & \alpha & -(y+\mu_{1})
\end{pmatrix}
\end{equation*} \\

The characteristic equation of the jacobian $J$ evaluated at $E_{1}$ is given by, \\
\begin{equation}
\lambda^3 + \bigg(p+\mu R_{0}\bigg)\lambda^2 + \bigg(p\mu R_{0}\bigg)\lambda + q\mu \bigg(R_{0}-1\bigg) = 0 \label{sec3equ4}  \\
\end{equation}
where $p=x+y+\mu+\mu_{1}$ and $q=(x+\mu)(\mu_{1}+y).$ \\

Since $R_{0} > 1$,  $(p+\mu R_{0})> 0, (p\mu R_{0}) > 1$ and $ q\mu (R_{0}-1) > 0 $ (because $R_{0}>1$ for $E_{1}$ to exists). Therefore if we substitute $\lambda = -k $ in  equation (\ref{sec3equ4}) using Descartes rule of sign change we get all the roots of (\ref{sec3equ4}) to be negative. Hence we conclude that the infected equilibrium point $E_{1}$ exists and remains asymptotically stable provided $R_{0} > 1.$\\

{\flushleft{  \textbf{GLOBAL DYNAMICS OF THE MODEL 1} }}\vspace{.25cm}

{\flushleft{  \textbf{GLOBAL STABILITY OF  $E_{0}$} }}\vspace{.25cm}

To establish the global stability of the infection free equilibrium $E_{0}$ we make use of the method discussed in Castillo-Chavez \textit{et al}  {\cite{GLB}}.\\ 

\begin{theorem}

Consider the following general system,\\

\begin{equation}
\begin{split}
\frac{dX}{dt} & = F(X,Y) \\[6pt] \label{sec3equ5}
\frac{dY}{dt} & = G(X,Y) \qquad   
\end{split}
\end{equation}
where $ X $ denotes the uninfected population compartments and $Y$ denotes the infected population compartments including latent, infectious etc. Here the function $G$ is such that it satisfies $G(X,0)=0$. Let $U_{0} = (X_{0}, \bar{0}) $ denote the equilibrium point of the above general system.

If the following  two conditions are satisfied then the infection free equilibrium point $U_{0}$ is globally asymptotically stable for the above general system provided $R_{0}<1$

$A_{1}$: For the subsystem $\frac{dX}{dt} = F(X,0)$, $X_{0}$ is globally asymptotically stable.
\vspace{1cm}

$A_{2}$: The function $G=G(X,Y)$ can be written as $G(X,Y)=AY - \widehat{G}(X,Y)$, where $\widehat{G_{j}}(X,Y) \geq 0$ $\forall \; (X,Y)$ in the biologically feasible region $\Omega$ for j=1,2 and $A = D_{Y}G(X,Y)$ at $(X_{0},\bar{0})$ is a M-matrix(matrix with non-negative off diagonal element).
\end{theorem}

\vspace{2cm}
Now we will prove the global stability of $E_{0}=(\frac{\omega}{\mu},0,0)$ of system (\ref{sec2equ1}-\ref{sec2equ3}) by showing that system (\ref{sec2equ1}-\ref{sec2equ3}) can be written as the above general form and both the conditions $A_{1}$ and $A_{2}$ are satisfied .\\
\vspace{.5cm}
	


Comparing the above general system (\ref{sec3equ5}) to the system (\ref{sec2equ1} - \ref{sec2equ3}) the functions F and G are given by
\\
$$   F(X,Y) = \omega - \beta S V - \mu S$$
$$\hspace{1.3cm}G(X,Y) = \bigg( \beta V S-(x+\mu) I, \; \alpha I-(y+\mu_{1}) V\bigg)$$ 

\noindent
\\ 
where $X=S$ and $Y=(I, V)$ \\
The disease free equilibrium point is $U_{0} = (X_{0}, \bar{0}), $ where,
\\
$$X_{0} = \frac{\omega}{\mu} \quad \text{and} \quad \bar{0} = (0, 0)$$

\noindent
\\
From the stability analysis of $E_{0}$, we know that  $U_{0}$ is locally asymptotically stable iff $R_{0} < 1$. Clearly, we see that $G(X,\bar{0}) = (0,\bar{0})$. Now, we show that $ X_{0} = (\frac{\omega}{\mu})$ is globally asymptotically stable for the subsystem 

\begin{equation}
    \frac{dS}{dt}=F(S,\bar{0})=\omega - \mu S  \label{sec3equ6}
\end{equation}
\noindent
The integrating factor is $e^{\mu t}$ and therefore after performing integration on the above equation (\ref{sec3equ6}) we get, \\

$$S(t)e^{\mu t}=\frac{\omega e^{\mu t}}{\mu }  + c$$
\\\noindent

As $t \rightarrow \infty$ we get, \\

$$S(t)= \frac{\omega}{\mu}$$
\noindent
which is independent of c. This independency implies that $X_{0}=\frac{\omega}{\mu_{1}}$ is globally asymptotically stable for the subsystem $\frac{dS}{dt}=\omega - \mu S$. So, the assumption  $A_{1}$ is satisfied.\\

\noindent
\\
Now, we will show that assumption $A_{2}$  holds. First, we will find the matrix $A$. As per the theorem, $A = D_{Y}G(X,Y)$ at $X=X_{0}$ and $Y=\bar{0}$. Now\\ 
\vspace{2mm}
\[ D_{Y}G(X,Y) = 
\begin{bmatrix}
-(x+\mu) & \beta S \\[6pt]
\alpha & -(y+\mu_{1})
\end{bmatrix}\]

\noindent
\\ 
At $X=X_{0}$ and $Y=\bar{0}$, we obtain,
\\ \vspace{1mm}
\[A = 
\begin{bmatrix}
-(x+\mu) & \frac{\beta \omega}{\mu} \\[6pt]
\alpha & -(y+\mu_{1}) 
\end{bmatrix}\]
\vspace{2mm}
\noindent
\\ 
Clearly, matrix $A$ has non-negative off-diagonal elements. Hence, $A$ is a M-matrix. Using 
$\widehat{G}(X,Y) = AY - G(X,Y)$, we get,
\vspace{2mm}
\[\widehat{G}(X,Y) \;\; = \;\;
\begin{bmatrix}
\widehat{G_{1}}(X,Y) \\[6pt]    
\widehat{G_{2}}(X,Y)    
\end{bmatrix} 
\;\; = \;\;
\begin{bmatrix}
\beta V(\frac{\omega}{\mu}-S) \\[6pt]
0
\end{bmatrix} \] \\

Hence $\widehat{G_{1}}(X,Y) = \beta V(\frac{\omega}{\mu}-S)\geq 0 $ because $S(t)\le\frac{\omega}{\mu}$ since $S(t)+I(t)+V(t)\le\frac{\omega}{\mu}$ and  $\widehat{G_{2}}(X,Y) =0$ \\

Thus both the assumptions $A_{1}$ and $A_{2}$ are satisfied and therefore infection free equilibrium point $E_{0}$ is globally asymptotically stable provided $R_{0} < 1.$

{\flushleft{  \textbf{GLOBAL STABILITY OF  $E_{1}$} }}\vspace{.25cm}

From the definition of infected equilibrium $E_{1}$ we get the following relations

\begin{enumerate}
	
	\item $\mu S^*= \omega-(x+\mu)I^*$ \\
	
	\item $\beta S^* V^*=(x+\mu)I^*$ \\
	
	\item $\alpha I^* = (y+\mu_{1})V^*$ \\ 
\end{enumerate}

We define the Lyaponov function as,

$L(S,I,V)=\frac{\bigg((S-S^*)+(I-I^*)\bigg)^2}{2}+ \frac{(I-I^*)^2}{2}+ \frac{(V-V^*)^2}{2}.$

Differentiating $L$ with respect to time we get, \\

$\frac{dL}{dt}=\bigg((S-S^*)+(I-I^*)\bigg)\frac{d(S+I)}{dt}+ (I-I^*)\frac{dI}{dt}+(V-V^*)\frac{dV}{dt}$

$$ = \bigg((S-S^*)+(I-I^*)\bigg)\bigg(\omega-\mu S -(x+\mu)I\bigg) +(I-I^*)\bigg(\beta SV-(x+\mu)I\bigg) \\
+ (V-V^*)\bigg(\alpha I-(y+\mu_{1})\bigg)$$

Using the above relations 1, 2 and 3 we get the following,

$\frac{dL}{dt}=\bigg((S-S^*)+ (I-I^*)\bigg)\bigg( \mu S^* +I^*(x+\mu)-\mu S -(x+\mu)I\bigg)\\
+(I-I^*)\bigg(\beta S(V-V^*) + \beta V^*(S-S^*)-(x+\mu)(I-I^*)\bigg)$
$ + (V-V^*)\bigg(\alpha (I-I^*)-(y+\mu_{1})(V-V^*)\bigg)$
$$\hspace{1cm}=\bigg((S-S^*)+ (I-I^*)\bigg)\bigg(-\mu(S-S^*)-(x+\mu)(I-I^*)\bigg) + \beta S(V-V^*)(I-I^*)+\beta V^*(S-S^*)(I-I^*)$$ 
$$-(x+\mu)(I-I^*)^2+\alpha(V-V^*)(I-I^*)-(y+\mu_{1})(V-V^*)^2$$

Further simplifying we get,\\
$$\frac{dL}{dt}=-\bigg(x+2\mu+\beta V^*\bigg)(S-S^*)(I-I^*)+(\beta S+\alpha)(V-V^*)(I-I^*)-\mu(S-S^*)^2-(x+\mu)(I-I^*)^2-(y+\mu_{1}(V-V^*)^2$$

Now we assume $\bigg(x+2\mu+\beta V^*\bigg)(S-S^*)(I-I^*) > 0$
and $\bigg(x+2\mu+\beta V^*\bigg)(S-S^*)(I-I^*) > (\beta S+\alpha)(V-V^*)(I-I^*)$
whenever $(\beta S+\alpha)(V-V^*)(I-I^*) > 0$ \\

Thus $\frac{dL}{dt}\le 0$  and  $\frac{dL}{dt}=0$ iff $S=S^*$, $I=I^*$, and $V=V^*$ \\

Hence by Lyapunov Lasalle theorem \cite{misra2011modeling}, the infected equilibrium $E_{1}$ is globally asymptotically stable. \\

{\flushleft{  \textbf{BIFURCATION ANALYSIS} }}\vspace{.25cm}

In this section we use the method given by Chavez and Song in  \cite{BIF} to do the bifurcation  analysis.

\vspace{.25cm}

\begin{theorem}

	Consider a system, 
	$$\frac{dX}{dt}=f(X,\phi)$$
	where $X \in \mathbb{R}^{n}$, $\phi \in \mathbb{R}$ is the bifurcation parameter and $f : \mathbb{R}^{n} \times \mathbb{R} \rightarrow \mathbb{R}^{n} $ where $f \in \mathbb{C}^2 (\mathbb{R}^n, \mathbb{R})$. Let $\bar{0}$ be the equilibrium point of the system such that $f(\bar{0},\phi) = \bar{0}, \forall \; \phi \in \mathbb{R}$. 
	Let the following conditions hold :
	
	\begin{enumerate}
		\item For the matrix $A = D_{X}f(\bar{0},0)$, zero is the simple eigenvalue and all other eigenvalues have negative real parts.
		
		\item Corresponding to zero eigenvalue, matrix A has non-negative right eigenvector, denoted as $u$ and non-negative left eigenvectors, denoted as $v$.
	\end{enumerate}
	
	\noindent
	Let $f_{k}$ be the $k^{th}$ component of $f$. Let $a$ and $b$ be defined as follows -
	
	$$ a = \sum_{k,i,j=1}^{n} \Bigg[ v_{k} w_{i} w_{j} \bigg(\frac{\partial^{2} f_{k}}{\partial x_{i} \partial x{j}} (\bar{0},0)\bigg) \Bigg]$$
	
	$$ \hspace*{-0.65cm} b = \sum_{k,i=1}^{n} \Bigg[ v_{k} w_{i} \bigg(\frac{\partial^{2} f_{k}}{\partial x_{i} \partial \phi}(\bar{0},0)\bigg) \Bigg]$$

	\noindent
	Then local dynamics of the system near the equilibrium point $\bar{0}$ is totally determined by the signs of $a$ and $b$. Here are the following conclusions :
	
	\begin{enumerate}
		\item If $a > 0$ and $b > 0$, then whenever $\phi < 0$ with $\mid \phi \mid \ll 1$, the equilibrium $\bar{0}$ is locally asymptotically stable, and moreover there exists a positive unstable equilibrium. However when $0 < \phi \ll 1$, $\bar{0}$ is an unstable equilibrium and there exists a negative and locally asymptotically stable equilibrium.
		
		\item If $ a < 0$, $b < 0$, then whenever $\phi < 0$ with $\mid \phi \ll 1 $, $\bar{0}$ is an unstable equilibrium whereas if $0 < \phi \ll 1$, $\bar{0}$ is locally asymptotically stable equilibrium and there exists a positive unstable equilibrium.
		
		\item If $a > 0$, $b < 0$, then whenever $\phi < 0$ with $\mid \phi \mid \ll 1$, $\bar{0}$ is an unstable equilibrium, and there exists a locally asymptotically stable negative equilibrium. However if $0 < \phi \ll 1$, $\bar{0}$ is stable, and a there appears a positive unstable equilibrium.
		
		\item If $a < 0$, $b > 0$, then whenever $\phi$ changes its value from negative to positive, the equilibrium $\bar{0}$ changes its stability from stable to unstable. Correspondingly a negative  equilibrium, unstable in nature, becomes positive and locally asymptotically stable.
	\end{enumerate}

\end{theorem}
\underline{\textbf{Applying the Theorem 3.2 to our Model 1 }}\textbf{: }
\noindent

In our case, we have $x = (S, I, B) \in \mathbb{R}^3$ where $x_{1} = S$, $x_{2} = I$ and $x_{3} = V$. Let us consider $\beta$ (transmission rate of the infection) to be the bifurcation parameter.\\

We know that $R_{0}=\frac{\beta \alpha \omega}{\mu(p+\mu)(q+\mu_{1})}$ where $p=d_{1}+d_{2}+d_{3}+d_{4}+d_{5}+d_{6}$ and $q=b_{1}+b_{2}+b_{3}+b_{4}+b_{5}+b_{6}.$\\

Therefore we have, \\

\vspace{2mm}
\begin{equation*}
\beta = \frac{ R_{0}\mu(p+\mu)(q+\mu_{1})}{\alpha \omega}
\end{equation*}
\noindent
\\
Let $\beta = \beta^{*}$ at $R_{0}=1$. So, we have,
\vspace{2mm}
\begin{equation}
\beta^* = \frac{\mu(p+\mu)(q+\mu_{1})}{\alpha \omega} \label{sec3equ7}
\end{equation}

\noindent
With $x=(x_{1},x_{2},x_3)=(S,I,V)$
system (\ref{sec2equ1} -\ref{sec2equ3}) can be written as follows :

\begin{align*}
\frac{dx_{1}}{dt} &= \omega-\beta x_{1}x_{3}-\mu x_{1} = f_{1}\\[6pt] 
\frac{dx_{2}}{dt} &=  \beta x_{1}x_{3}-(P+\mu)x_{2} = f_{2}\\[6pt]
\frac{dx_{3}}{dt} &= \alpha x_{2}-(q+\mu_{1}) x_{3} \quad= f_{3}
\end{align*}

\noindent
The disease free equilibrium point $E_{0}$ is given by,

$$x^* = \bigg(\frac{\omega}{\mu}, 0, 0\bigg) = (x_{1}^*, x_{2}^*, x_{3}^*)$$
\vspace{2mm}
\noindent
\\ 
Clearly, $f(x^*,\beta) = 0, \; \forall \; \beta \in \mathbb{R}$, where $ f = (f_{1},f_{2}, f_{3})$. Let $D_{x}f(x^*,\beta^*)$ denote the Jacobian matrix of the above system at the equilibrium point $x^*$ and $R_{0} = 1$. Now we see that,
\vspace{2mm}

\[
D_{x}f(x^*,\beta^*) =
\begin{bmatrix}
-\mu & 0 & \frac{-\beta^*\omega}{\mu} \\[6pt]
0 & -(p+\mu) & \frac{\beta^*\omega}{\mu} \\[6pt]
0 & \alpha & -(q+\mu_{1})
\end{bmatrix}
\]
\vspace{2mm}
\noindent
\\
The characteristic polynomial of the above matrix is obtained as 

\begin{equation}
    (-\mu-\lambda)\bigg[(-(p+\mu)-\lambda)(-(q+\mu_{1})-\lambda) - \bigg(\frac{\alpha \beta^*\omega}{\mu}\bigg)\bigg] = 0  \label{sec3equ8}
 \end{equation}

\vspace{2mm}
\noindent
\\
Hence, we obtain the first eigenvalue  of (\ref{sec3equ8})
as $$\boldsymbol{\lambda_{1}= -\mu < 0} $$

\noindent
\\
The other eigenvalues $\lambda_{2,3}$ of (\ref{sec3equ8}) are the solution the following equation,\\ 
\begin{equation}
    \lambda^2 +\bigg(p+q+\mu+\mu_{1}\bigg)\lambda+(p+\mu)(\mu_{1}+q)-\frac{\beta^* \alpha\omega}{\mu}=0  \label{sec3equ9}
\end{equation}

substituting the expression for  $\beta^*$ from (\ref{sec3equ7}) in (\ref{sec3equ9}) we get, \\

\begin{equation}
    \lambda^2 +\bigg(p+\mu +q+\mu_{1}\bigg)\lambda = 0 \label{sec3equ10}
\end{equation}

The  eigen values of (\ref{sec3equ10}) are $\lambda_{2}=0$ and $\lambda_{3} = -(p+q+\mu+\mu_{1})$

\noindent
Hence, the matrix $D_{x}f(x^*,\beta^*)$ has zero as its simple eigenvalue and all other eigenvalues with negative real parts. Thus, the condition 1 of the Theorem 3.2 is satisfied.
\vspace{1cm}\\
Next, for proving condition 2, we need to find the right and left eigenvectors of the zero eigenvalue ($\lambda_{2}$). Let us denote the right and left eigenvectors by $\boldsymbol{u}$ and $\boldsymbol{v}$ respectively. To find $u$, we use $(D_{x}f(x^*,\beta^*) - \lambda_{2} I_{d})u = 0 $, which implies that 
\\

\[\
\begin{bmatrix}
-\mu & 0 & \frac{-\beta^*\omega}{\mu} \\[6pt]
0 & -(p+\mu) & \frac{\beta^*\omega}{\mu} \\[6pt]
0 & \alpha & -(q+\mu_{1})

\end{bmatrix}
\begin{bmatrix}
u_{1} \\[6pt]
u_{2} \\[6pt]
u_{3} \\[6pt]
\end{bmatrix}
\;\; = \;\;
\begin{bmatrix}
0 \\[6pt]
0 \\[6pt]
0 
\end{bmatrix}
\]

\vspace{2mm}
\noindent
\\ 
where $ u = (u_{1}, u_{2}, u_{3})^T$. As a result, we obtain the system of simultaneous equations as follows :

\begin{equation}
    -\mu u_{1} - \frac{\beta^*\omega}{\mu}u_{3} = 0 \label{sec3equ11}
\end{equation}
 
 \begin{equation}
 -(p+\mu)u_{2} + \frac{\beta^*\omega}{\mu}u_{3} = 0 \label{sec3equ12}
\end{equation}

 \begin{equation}
 \alpha u_{2} - (q+\mu_{1}) u_{3} = 0 \label{sec3equ13}
\end{equation}
\noindent
\\

By choosing $u_{3} = \mu$ in the above simultaneous equation  (\ref{sec3equ11}-\ref{sec3equ13}) we obtain 
$$ u_{2} = \frac{\beta^*\omega}{(p+\mu)} \;\;\; \text{and} \;\;\;{u_{1} = -\frac{\beta^*\omega}{\mu}} $$

\noindent
Therefore, the right eigen vector of zero eigenvalue is given by

\begin{equation*}
\boldsymbol{u = \bigg(-\frac{\beta^*\omega}{\mu}, \; \frac{\beta^*\omega}{(p+\mu)}, \; \mu \bigg)} 
\end{equation*}

\noindent

Similarly, to find the left eigenvector $v$, we use $v(D_{x}f(x^*,\beta^*) - \lambda_{2} I_{d}) = 0 $, which implies that

\[
\begin{bmatrix}
v_{1} & v_{2} & v_{3}
\end{bmatrix}
\begin{bmatrix}
-\mu & 0 & \frac{-\beta^*\omega}{\mu} \\[6pt]
0 & -(p+\mu) & \frac{\beta^*\omega}{\mu} \\[6pt]
0 & \alpha & -(q+\mu_{1})\\[6pt]

\end{bmatrix}
\;\; = \;\;
\begin{bmatrix}
0 & 0 & 0
\end{bmatrix}
\]

\noindent

where $v = (v_{1}, v_{2}, v_{3})$. The simultaneous equations obtained thereby are as follows :

\begin{equation}
     -\mu_{1}v_{1} = 0  \label{sec3equ14}
\end{equation}

\begin{equation}
 -\mu_{1}v_{2} + \alpha v_{3} = 0    \label{sec3equ15}
\end{equation}

\begin{equation}
 \frac{-\beta^*\omega}{\mu_{1}}v_{1}+ \frac{\beta^*\omega}{\mu_{1}}v_{2} -qv_{3} = 0   \label{sec3equ16} 
\end{equation}
\vspace{2mm}
\noindent
\\
Therefore solving the above simultaneous equation (\ref{sec3equ14} - \ref{sec3equ16}) we obtain $v_{1} = 0$.\\

By choosing $v_{2}=1$ we get $$ v_{3}= \frac{\beta^*\omega}{(p+\mu_{1})} $$

\noindent
\\
Hence, the left eigen vector is given by 

\begin{equation*}
\boldsymbol{v = \bigg(0, \; 1, \; \frac{\beta^*\omega}{(q+\mu_{1})}\;\bigg)} 
\end{equation*}

\noindent
Now, we need to find $a$ and $b$. As per the Theorem 3.2, $a$ and $b$ are given by
\\ \vspace{2mm}
$$ \hspace*{-10mm} \boldsymbol{a = \sum_{k,i,j=1}^{3} \Bigg[v_{k}u_{i}u_{j} \bigg(\frac{\partial^2 f_{k}}{\partial x_{i} \partial x_{j}}(x^*, \beta^*)\bigg) \Bigg]}$$
\vspace{2mm}
$$ \hspace*{-15mm} \boldsymbol{b =\sum_{k,i =1}^{3} \Bigg[v_{k}u_{i}\bigg(\frac{\partial^2 f_{k}}{\partial x_{i} \partial \beta}(x^*, \beta^*)\bigg)\Bigg]}$$

\noindent
\\
Expanding the summation in the expression for $a$, it reduces to

$$ a = u_{1}u_{3} \frac{\partial^2 f_{2}}{\partial x_{1} \partial x_{3}} \thinspace + \thinspace \thinspace  u_{3}\bigg(u_{1}\frac{\partial^2 f_{2}}{\partial x_{3} \partial x_{1}} + \thinspace  u_{3}\frac{\partial^2 f_{2}}{\partial x_{3} \partial x_{3}}\bigg) \thinspace \thinspace  + \thinspace \thinspace  v_{3}u_{3}u_{3}\frac{\partial^2 f_{3}}{\partial x_{3} \partial x_{3}} $$

\noindent
\\ where partial derivatives are found at $(x^*, \beta^*)$. Since we know $u$ and $v$ we only need to find the partial derivatives in above expression. They are found to be
\\
$$ \frac{\partial^2 f_{2}}{\partial x_{1} \partial x_{3}}(x^*, \beta^*) = \beta^* \qquad \frac{\partial^2 f_{2}}{\partial x_{3} \partial x_{1}}(x^*, \beta^*) = \beta^*  $$

$$ \frac{\partial^2 f_{2}}{\partial x_{3} \partial x_{3}}(x^*, \beta^*) = 0 \qquad \frac{\partial^2 f_{3}}{\partial x_{3} \partial x_{3}}(x^*, \beta^*) = 0 $$

\noindent
Substituting these partial derivatives along with $u$ and $v$ in the expression of $a$, we get,

$$ \boldsymbol{a = -2\beta^{*2} \omega < 0}$$

\noindent
\\
Next, expanding the summation in the expression for $b$, we get,

$$b = v_{2}u_{3}\bigg(\frac{\partial^2 f_{2}}{\partial x_{3} \partial \beta}(x^*, \beta^*)\bigg) $$

\noindent
\\
Now
$$ \frac{\partial^2 f_{2}}{\partial x_{3} \partial \beta}(x^*, \beta^*) = \frac{\omega}{\mu} $$

\noindent
which implies,

$$ \boldsymbol{b = \omega > 0}$$

\noindent
We notice that condition (iv) of the theorem is satisfied. Hence, we conclude that the system undergoes bifurcation at $\beta = \beta^*$ implying $R_{0} = 1$.

\noindent
\\ 
Thus, we conclude that when $R_{0} < 1$, there exists a unique disease free equilibrium  which is globally asymptotically stable  and negative infected equilibrium which is unstable . Since negative values of population is not practical, therefore we ignore it in this case. Further, as $R_{0}$ crosses unity from below, the disease free equilibrium point loses its stable nature and become unstable, the bifurcation point being at $\beta = \beta^*$ implying $R_{0}=1$ and there appears a positive locally asymptotically stable infected equilibrium point. There is an exchange of stability between disease free equilibrium and infected equilibrium at $R_{0} = 1$. Hence, a \textbf{trans-critical bifurcation} takes place at the break point $\beta = \beta^*$. \\

{\flushleft{  \textbf{PARAMETER ESTIMATION} }}\vspace{.25cm}

We now numerically depict and verify the results obtained in Local and  Global dynamics and Bifurcation  Analysis sections.  The theoretical results obtained  are validated for a  set of model parameters using MATLAB software. The values of $\omega,$ $\mu$ and $\mu_{1},$ $\alpha$ are approximated and chosen to be  from \cite{2} and \cite{ea2020host} respectively.  The rest of the  parameter values of the model are estimated minimizing the root mean square difference between the model predictive output and the experimental data from \cite{ehmannvirological, qin2020dysregulation}. All the parameter values chosen for the model 1  are summarized in the following table. \vspace{1cm}

\begin{table}[ht!]
	\caption{Values of the Model 1 parameters} 
	
	\begin{center}
		\begin{tabular}{|l|l|l|}
			\hline
			\textbf{S.No.} & \textbf{Parameters} & \textbf{Value}\\
			\hline 
			1 & $\omega$ & 10 \\
			\hline
			2 & $\beta$ & 0.005 \\
			\hline
			3 & $\mu$ & .05 \\
			\hline
			4 & $\mu_{1}$ & 1.1 \\
			\hline
			5 & $\alpha$ & 8.2 \\
			\hline
			6 & $d_{1}$ & 0.027 \\
			\hline
			7 & $d_{2}$ & 0.22 \\
			\hline
			8 & $d_{3}$ & 0.1 \\
			\hline
			9 & $d_{4}$ & 0.428 \\
			\hline
			
			10 & $d_{5}$ & 0.01 \\
			\hline
			11 & $d_{6}$ & .01 \\
			\hline
			12 & $b_{1}$ & 0.1\\
			\hline
			13 & $b_{2}$ & 0.1\\
			\hline
			14 & $b_{3}$ & 0.08\\
			\hline
			15 & $b_{4}$ & 0.11\\
			\hline
			16 & $b_{5}$ & 0.1\\
			\hline
			17 & $b_{6}$ & 0.07\\
			
			\hline
		\end{tabular}
	\end{center}
\end{table}

\newpage

{\flushleft{  \textbf{NUMERICAL SIMULATIONS} }}\vspace{.25cm}

{\flushleft{  \textbf{DISEASE FREE EQUILIBRIUM $E_0$} }}\vspace{.25cm}

The following table shows different values of parameters chosen so as to obtain the phase portraits of the system for the case $R_{0} < 1$.

\begin{table}[ht!]
	\caption{Values of the Model parameters : $R_{0} < 1$}
	\begin{center}
		\begin{tabular}{|l|l|l|}
			\hline
			\textbf{S.No.} & \textbf{Parameters} & \textbf{Value}\\
			\hline 
			1 & $\omega$ & 2 \\
			\hline
			2 & $\beta$ & 0.05 \\
			\hline
			3 & $\mu$ & .1 \\
			\hline
			4 & $\mu_{1}$ & .1 \\
			\hline
			5 & $\alpha$ & 0.5 \\
			\hline
			6 & $d_{1}$ & 0.027 \\
			\hline
			7 & $d_{2}$ & 0.22 \\
			\hline
			8 & $d_{3}$ & 0.1 \\
			\hline
			9 & $d_{4}$ & 0.428 \\
			\hline
			10 & $d_{5}$ & 0.01 \\
			\hline
			11 & $d_{6}$ & 0.01 \\
			\hline
			12 & $b_{1}$ & 0.1\\
			\hline
			13 & $b_{2}$ & 0.1\\
			\hline
			14 & $b_{3}$ & 0.08\\
			\hline
			15 & $b_{4}$ & 0.11\\
			\hline
			16 & $b_{5}$ & 0.1\\
			\hline
			17 & $b_{6}$ & 0.07\\
			
			\hline
		\end{tabular}
	\end{center}
\end{table}

For the above set of parameter values, $\boldsymbol{R_{0} = 0.77 < 1}.$  Since $R_{0}<1$, we can say that the disease free equilibrium point, $\boldsymbol{E_{0} = (20,0,0)}$  is globally asymptotically stable. The following figures(1-2) shows the global stability of the infection free equilibrium $E_{0}$ .\\

Figure 2 shows the global asymptotic stability of $\boldsymbol{E_{0} = (20, 0, 0)}$. It is obtained using different initial conditions for the same set of the model parameters given in Table 2
\newpage
\begin{center}
	\begin{figure}[hbt!]
		\includegraphics[height = 6cm, width = 15.5cm]{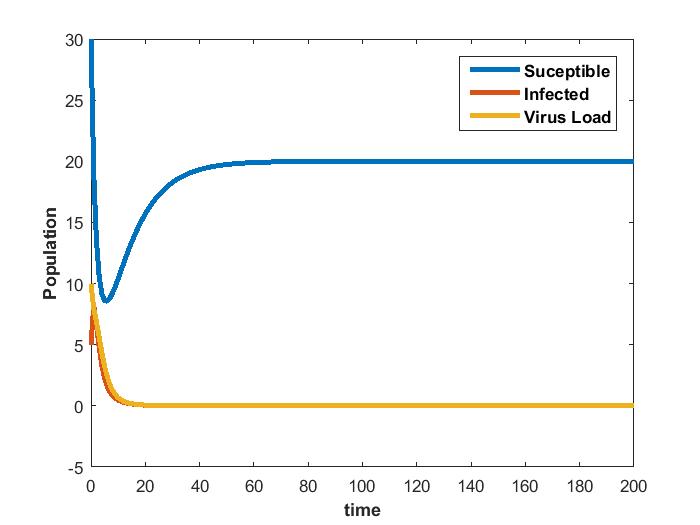} 
		\caption{\label{first}Initial condition:  $(S_{0}, E_{0}, I_{0}) = (30, 5, 10)$ }
	\end{figure} 
\end{center}

\begin{figure}[hbt!]
	\includegraphics[height = 8cm, width = 15.5cm]{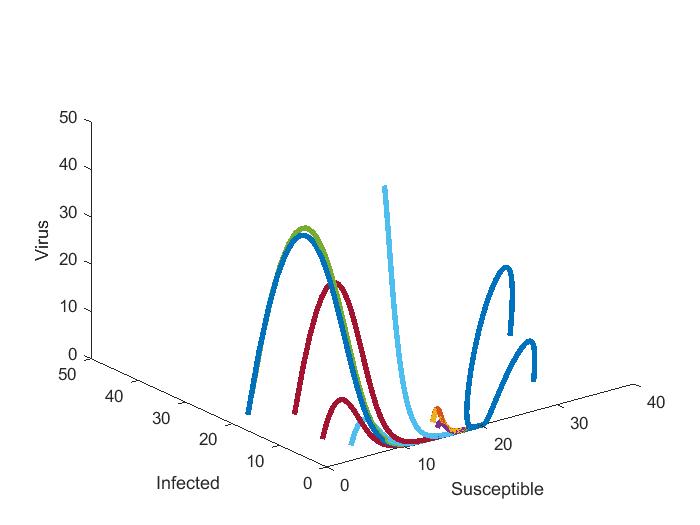} 
	\vspace{3mm}
	\caption{\label{second}3-d plot of the system for $R_{0} < 1$.}
	
\end{figure}   

\newpage
{\flushleft{  \textbf{INFECTED/ENDEMIC  EQUILIBRIUM $E_1$} }}\vspace{.25cm}

We know that the infected equilibrium $E_{1}$ exists only if $R_{0}>1$   and it  is also locally asymptotically stable when $R_{0}>1$. For the parameter values in the following table 3 the value of $R_{0}$ was calculated to be 1.929 and  $E_{1}=(25.9,2.44,1.85)$. Figure 3 demonstrates that $E_{1}$ is locally asymptotically stable whenever $R_{0} > 1$.

\vspace{1cm}
\begin{table}[ht!]
	\caption{Values of the Model parameters : $R_{0} > 1$}
	\begin{center}
		\begin{tabular}{|l|l|l|}
			\hline
			\textbf{S.No.} & \textbf{Parameters} & \textbf{Value}\\
			\hline 
			1 & $\omega$ & 5 \\
			\hline
			2 & $\beta$ & 0.05 \\
			\hline
			3 & $\mu$ & .1 \\
			\hline
			4 & $\mu_{1}$ & .1 \\
			\hline
			5 & $\alpha$ & 0.5 \\
			\hline
			6 & $d_{1}$ & 0.027 \\
			\hline
			7 & $d_{2}$ & 0.22 \\
			\hline
			8 & $d_{3}$ & 0.1 \\
			\hline
			9 & $d_{4}$ & 0.428 \\
			\hline
			10 & $d_{5}$ & 0.01 \\
			
			\hline
			11 & $d_{6}$ & 0.01 \\
			\hline
			12 & $b_{1}$ & 0.1\\
			\hline
			13 & $b_{2}$ & 0.1\\
			\hline
			14 & $b_{3}$ & 0.08\\
			\hline
			15 & $b_{4}$ & 0.11\\
			\hline
			16 & $b_{5}$ & 0.1\\
			\hline
			17 & $b_{6}$ & 0.07\\
			
			\hline
		\end{tabular}
	\end{center}
\end{table}

\begin{figure}[hbt!]
	\includegraphics[height = 6cm, width = 15.5cm]{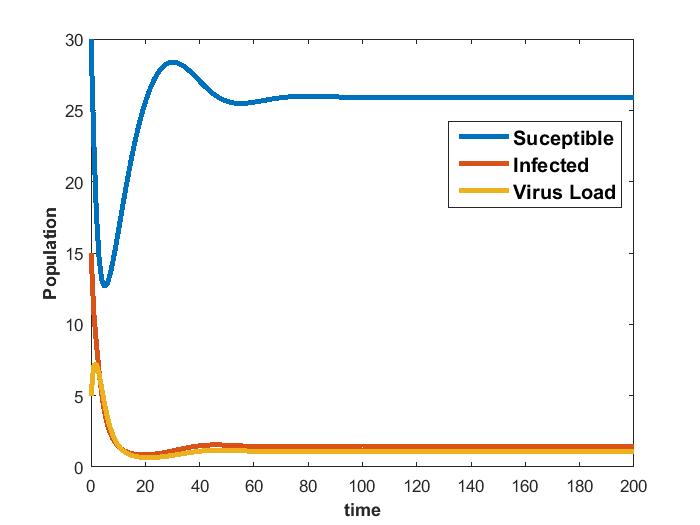}
	\caption{\label{first}Initial condition:  $(S_{0}, E_{0}, I_{0}) = (30, 15, 5)$ }
	\vspace{6mm}
\end{figure}

{\flushleft{  \textbf{TRANSCRITICAL BIFURCATION} }}\vspace{.25cm}

In this section, we numerically show that as $R_{0}$ crosses unity from below, there is an exchange of stability between the disease free equilibrium and infected equilibrium. Here the parameter $\omega$ is varied in the interval $(0, 5)$ in steps of 0.1. The other parameters chosen are given in the following table 5. As a result $R_{0}$ varies in the interval $(0,3)$. Figure 4 depicts the trans-critical bifurcation at $R_{0}=1$. 
\vspace{1cm}
\begin{table}[ht!]
	\caption{Values of the parameters}
	\begin{center}
		\begin{tabular}{|l|l|l|}
			\hline
			\textbf{S.No.} & \textbf{Parameters} & \textbf{Value}\\
			\hline 
			
			1 & $\beta$ & 0.05 \\
			\hline
			2 & $\mu$ & .1 \\
			\hline
			3 & $\mu_{1}$ & .1 \\
			\hline
			4 & $\alpha$ & 0.5 \\
			\hline
			5 & $d_{1}$ & 0.027 \\
			\hline
			6 & $d_{2}$ & 0.22 \\
			\hline
			7 & $d_{3}$ & 0.1 \\
			\hline
			8 & $d_{4}$ & 0.428 \\
			\hline
			9 & $d_{5}$ & 0.01 \\
			\hline
			10 & $d_{6}$ & 0.01 \\
			\hline
			11 & $b_{1}$ & 0.1\\
			\hline
			12 & $b_{2}$ & 0.1\\
			\hline
			13 & $b_{3}$ & 0.08\\
			\hline
			14 & $b_{4}$ & 0.11\\
			\hline
			15 & $b_{5}$ & 0.1\\
			\hline
			16 & $b_{6}$ & 0.07\\
			\hline
		\end{tabular}
	\end{center}
\end{table}

\begin{figure}[hbt!]
	\includegraphics[height = 6cm, width = 15.5cm]{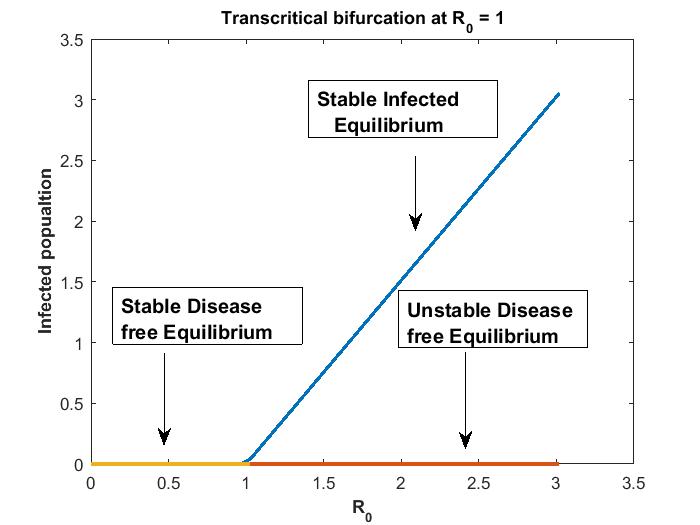}
	\caption{\label{first}Trans-critical Bifurcation }
	
\end{figure}

\newpage

\vspace{4cm}
{\flushleft{  \textbf{CHARACTERISTICS OF COVID-19 AND VALIDATION OF THE PROPOSED MODEL} }}\vspace{.25cm}

The characteristic features of the disease Covid-19 are as follows:\vspace{.25cm}

(a) Target cells(monocytes) deplete to approximately 66 \% to the original level (from $6$x$10^8$ to $4$x$10^8$ in severe cases) \cite{qin2020dysregulation}.\\

(b) The median incubation period of SARS-Cov-2 is approximately 4-5 days and peak viremia  occurs within 5-6 days of  disease onset \cite{tay2020trinity}.

\vspace{.5cm}
In this section we vary the parameters $x=(d_{1}+d_{2}+d_{3}+d_{4}+d_{5}+d_{6})$ and $y=(b_{1}+b_{2}+b_{3}+b_{4}+b_{5}+b_{6})$ and do a two parameter heat plot to validate our model 1. The parameters $x$ and  $y$ were varied in the interval $(0,5)$ and  model 1 was able to reproduce both the above characteristic  of the Covid-19. Other fixed parameter values  were taken from the following table 5.

\vspace{1cm}

\begin{table}[ht!]
	\caption{Parameter values } 
	
	  \begin{center}
		\begin{tabular}{|l|l|l|}
			\hline
			\textbf{S.No.} & \textbf{Parameters} & \textbf{Value}\\
			\hline 
			1 & $\omega$ & 10 \\
			\hline
			2 & $\beta$ & 0.005 \\
			\hline
			3 & $\mu$ & .05 \\
			\hline
			4 & $\mu_{1}$ & 1.1 \\
			\hline
			5 & $\alpha$ & 8.2 \\
			
			\hline
		\end{tabular}
	\end{center}
\end{table}

We choose the parameters $x$  and $y$ in the $x-$ axis and $y-$ axis respectively and vary them to check for reproduction  of  the characteristics (a) and (b).

\vspace{.25cm}

In figure 5 model 1 is able to recover the first characteristic (a) exactly for the parameter values in the region pointed by the arrow. In this  region with yellow colour the final fraction of uninfected cells after infection lies between $(60-70)\%$ and for the parameter values in this region the model 1 is able to reproduce the first characteristic (a). \vspace{.25cm}

The region with yellow colour in figure 6 pointed by the arrow is the region where Model 1 is able to recover the second characteristics (b). From \cite {tay2020trinity} the peak viremia occurs approximately during the second week of disease onset. In figure 6 time to peak viremia lies between 8-16 days and hence for the parameter values in this region Model 1 is able to reproduce the second characteristics (b).
\vspace{.25cm}

\begin{figure}[hbt!]
	\includegraphics[height = 8cm, width = 15.5cm]{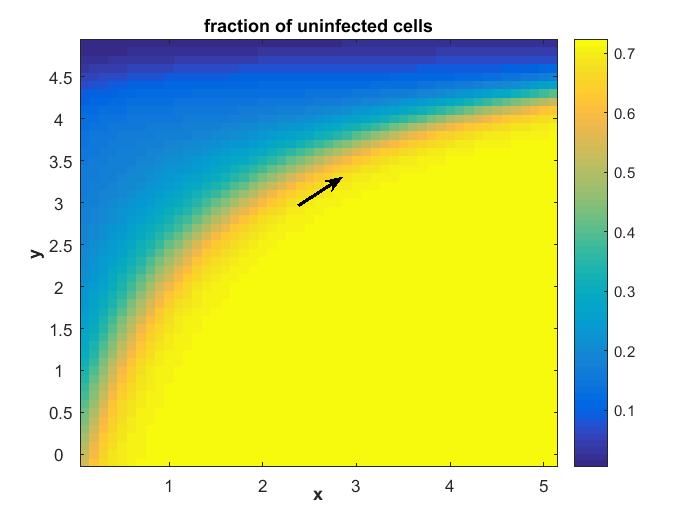} 
	\caption{\label{first}Initial condition: $(S_{0}, E_{0}, I_{0}) $= ($6$x$10^8$, 0, 1) }
\end{figure} 

\vspace{4cm}

\begin{figure}[hbt!]
	\includegraphics[height = 8cm, width = 15.5cm]{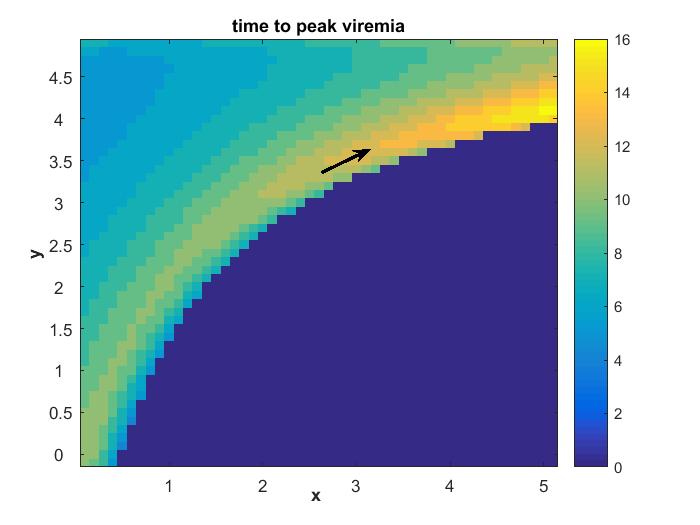} 
	\caption{\label{first}Initial condition:  $(S_{0}, E_{0}, I_{0}) =$ ($6$x$10^8$, $0, 1)$ }
\end{figure}

\newpage

{\flushleft{  \textbf{SENSITIVITY ANALYSIS} }}\vspace{.25cm}

	From the previous sections, it is clear that the infected cell population and the virus load die out  when $R_{0} < 1$. Therefore it is important to control the model parameters in a manner which will make $R_{0}$ less than one. Thus, determining the intervals in which the model parameters are sensitive becomes vital. As each parameter is varied in different intervals, the infected cell population, mean infected cell population and the mean square error are plotted with respect to time. These plots are used to determine the sensitivity of the parameter. The different intervals chosen are given in the following Table \ref{sen_anl}. For all the plots in this section, the time scale is the following: $x-$axis: 10 units = 1 day, $y-$ axis: 1 unit = 1 cell.  \vspace{.75cm}
	
	\begin{table}[hbt!]
		\caption{Sensitivity Analysis}
		\centering
		\label{sen_anl}
		{
			\begin{tabular}{|l|l|l|l|l|}
				\hline
				\textbf{Parameter} & \textbf{Interval} & \textbf{Step Size} & \textbf{Other Parameters} \\
				\hline
				
				$u_{11}$ & 0 to 0.4 & 0.01 & $\omega$ = 10, $\beta$ = 0.005, $\mu$ = 0.05
				\\ \cline{2-2}
				& 0.4 to 0.8 & &  $\mu_{1}$ = 1.1, $u_{12}$ = 0.6240, $\alpha$ = 0.7 
				\\ \cline{2-3}
				& 0.8 to 0.1 & 0.001 & \\ 
				\hline 
				
				$u_{12}$ & 0 to 1 & 0.05 & $\omega$ = 10, $\beta$ = 0.005, $\mu$ = 0.05
				\\ \cline{2-2} 
				& 1 to 2 & & $\mu_{1}$ = 1.1, $u_{11}$ = 1.0238, $\alpha$ = 0.7 
				\\ \hline
				
				$\beta$ & 0 to 0.01 & 0.001 & $\omega$ = 10, $u_{11}$ = 1.0238, $\mu$ = 0.05
				\\ \cline{2-2}
				& 0.01 to 0.02 & &  $\mu_{1}$ = 1.1, $u_{12}$ = 0.6240, $\alpha$ = 0.7 
				\\ \cline{2-2}
				& 0.02 to 0.03 & & \\ 
				\hline 
				
				$\alpha$ & 0.5 to 1.2 & 0.01 & $\omega$ = 10, $\beta$ = 0.005, $\mu$ = 0.05
				\\ \cline{2-2}
				& 1.2 to 3 & &  $\mu_{1}$ = 1.1, $u_{12}$ = 0.6240, $u_{11}$ = 1.0238 
				\\ \cline{2-3}
				& 3 to 4 & 0.05 & \\ 
				\hline 
				
				$\omega$ & 5 to 15 & 0.1 & $\beta$ = 0.005, $u_{11}$ = 1.0238, $\mu$ = 0.05
				\\ \cline{2-2}
				& 20 to 40 & &  $\mu_{1}$ = 1.1, $u_{12}$ = 0.6240, $\alpha$ = 0.7 
				\\ \cline{2-2}
				& 40 to 60 & & \\ 
				\hline 
				
				$\mu$ & 0.001 to 0.015 & 0.0001 & $\omega$ = 10, $\beta$ = 0.005, $u_{11}$ = 1.0238
				\\ \cline{2-2}
				& 0.01to 0.03 & &  $\mu_{1}$ = 1.1, $u_{12}$ = 0.6240, $\alpha$ = 0.7 
				\\ \cline{2-2}
				& 0.03 to 0.08 & & \\ 
				\hline 
				
				$\mu_{1}$ & 0.01 to 0.1 & 0.001 & $\omega$ = 10, $\beta$ = 0.005, $\mu$ = 0.05
				\\ \cline{2-3}
				& 0.1 to 0.4 & 0.01 &  $u_{11}$ = 1.0238, $u_{12}$ = 0.6240, $\alpha$ = 0.7 
				\\ \cline{2-2}
				& 0.4 to 1.8 & & \\ 
				\hline
				
			\end{tabular}
		}
	\end{table}
	
	\subsubsection{Parameter $\boldsymbol{\alpha}$}
	
	\begin{enumerate}
		\item \textbf{Interval I : 0.5 to 1.2 :}
		The results related to sensitivity of $\alpha$, varied from 0.5 to 1.2 in step size of 0.01  are shown in the Figure (\ref{sen_alpha_1}). The parameter $\alpha$ is not sensitive to the system in this interval.
		
		\begin{figure}[hbt!]
			\begin{center}
				\includegraphics[width=3in, height=1.8in, angle=0]{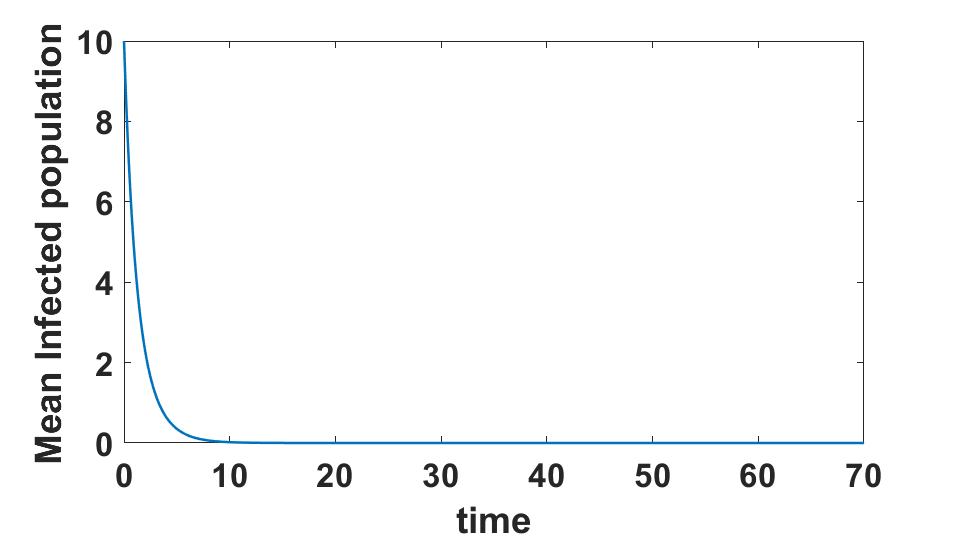}
				\includegraphics[width=3in, height=1.8in, angle=0]{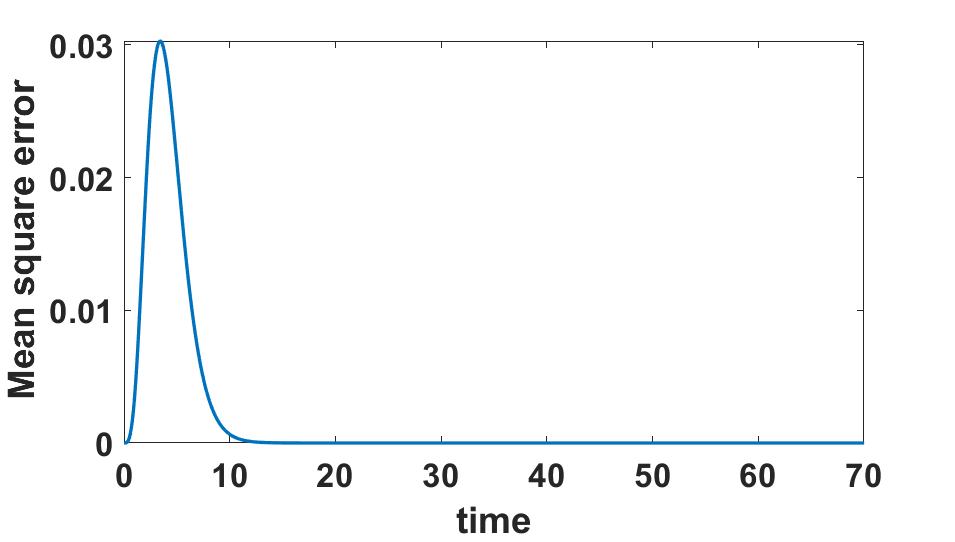}
				\caption{Sensitivity Analysis of $\alpha$ in Interval I.}
				\label{sen_alpha_1}
			\end{center}
		\end{figure}
	\end{enumerate}
	
	\begin{enumerate}
		\setcounter{enumi}{1}
		\item \textbf{Interval II : 1.2 to 3 :} The results related to sensitivity of $\alpha$, varied from 1.2 to 3 in step size of 0.01, are shown in the Figure (\ref{sen_alpha_2}). The parameter $\alpha$ is sensitive to the system in this interval.  
		
		\begin{figure}[hbt!]
			\begin{center}
				\includegraphics[width=3in, height=1.8in, angle=0]{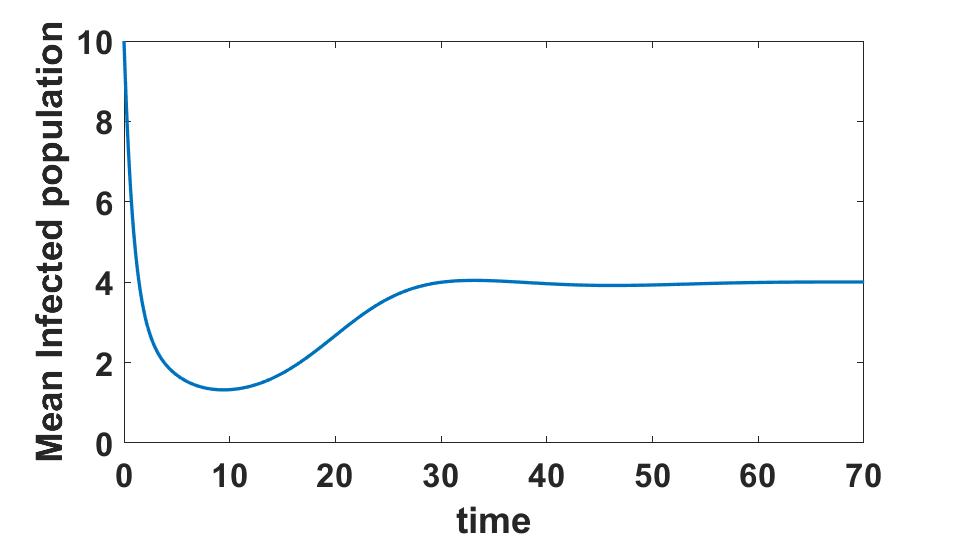}
				\includegraphics[width=3in, height=1.8in, angle=0]{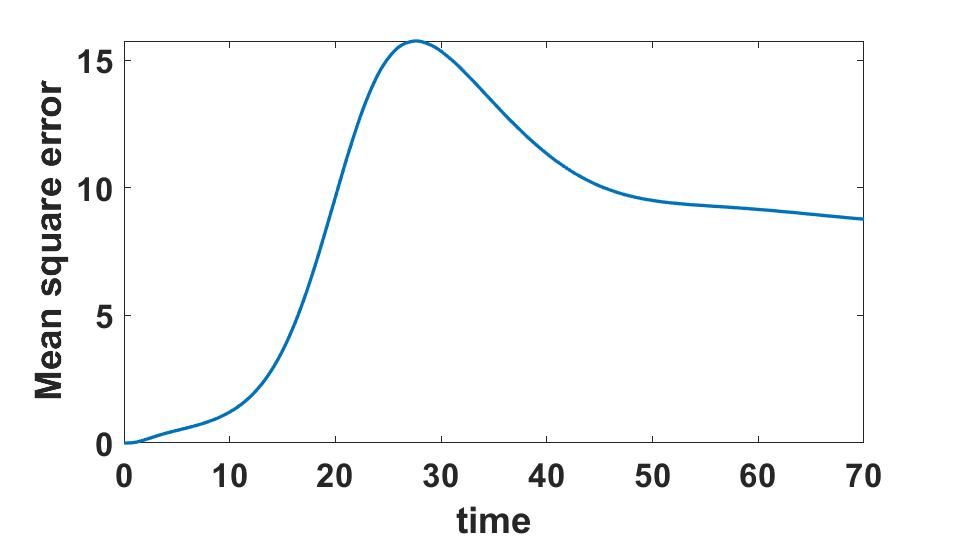}
				\caption{Sensitivity Analysis of $\alpha$ in Interval II.}
				\label{sen_alpha_2}
			\end{center}
		\end{figure}
	\end{enumerate}
	
	\begin{enumerate}
		\setcounter{enumi}{2}
		\item \textbf{Interval III : 3 to 4 :} The results related to sensitivity of $\alpha$, varied from 3 to 4 in step size of 0.05, are shown in the Figure (\ref{sen_alpha_3}). The parameter $\alpha$ is not sensitive to the system in this interval. 
		
		We conclude from these plots that the parameter $\alpha$ is sensitive in interval II and insensitive in I and III. To confirm and validate the same, we have plotted the infected population for each varied value of the parameter $\alpha$ per interval in Figure (\ref{sen_alpha}).
		
		\begin{figure}[hbt!]
			\begin{center}
				\includegraphics[width=3in, height=1.8in, angle=0]{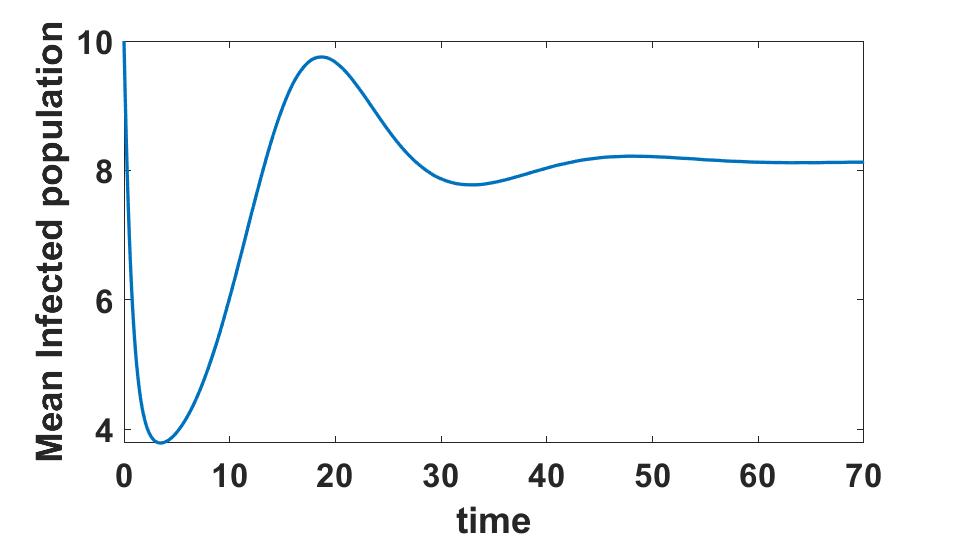}
				\includegraphics[width=3in, height=1.8in, angle=0]{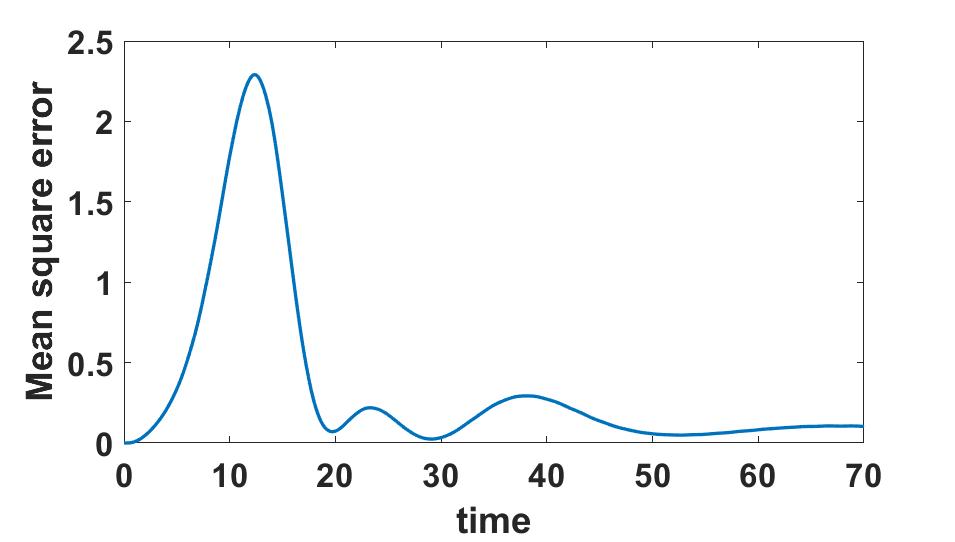}
				\caption{Sensitivity Analysis of $\alpha$ in Interval III.}
				\label{sen_alpha_3}
			\end{center}
		\end{figure}
	\end{enumerate}
	
	\begin{figure}[hbt!]
		\begin{center}
			\subcaptionbox*{(a) Interval I}
			{\includegraphics[width=3in, height=1.8in, angle=0]{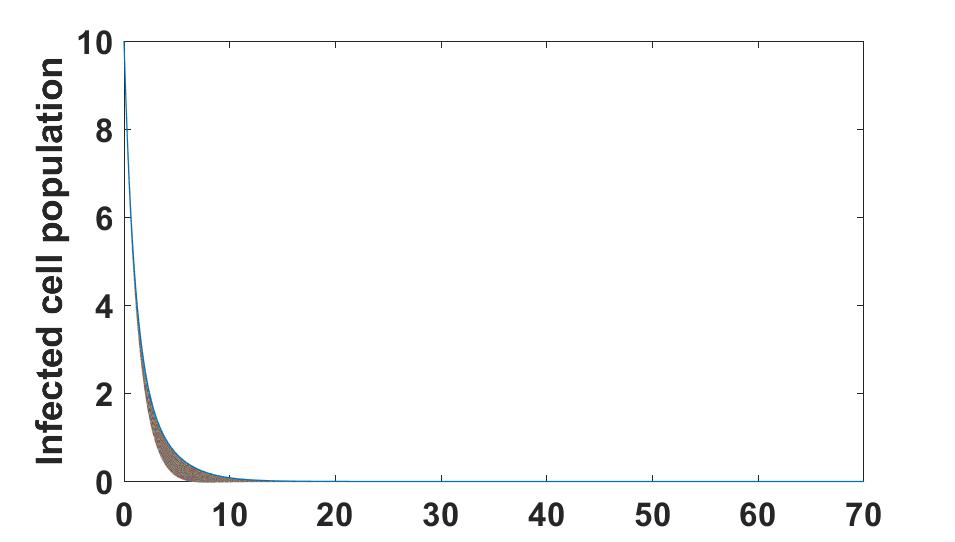}}
		\end{center}
	\end{figure}
	
	\addtocounter{figure}{-1}
	
	\begin{figure}[hbt!]
		\begin{center}
			\subcaptionbox*{(b) Interval II}
			{\includegraphics[width=3in, height=1.8in, angle=0]{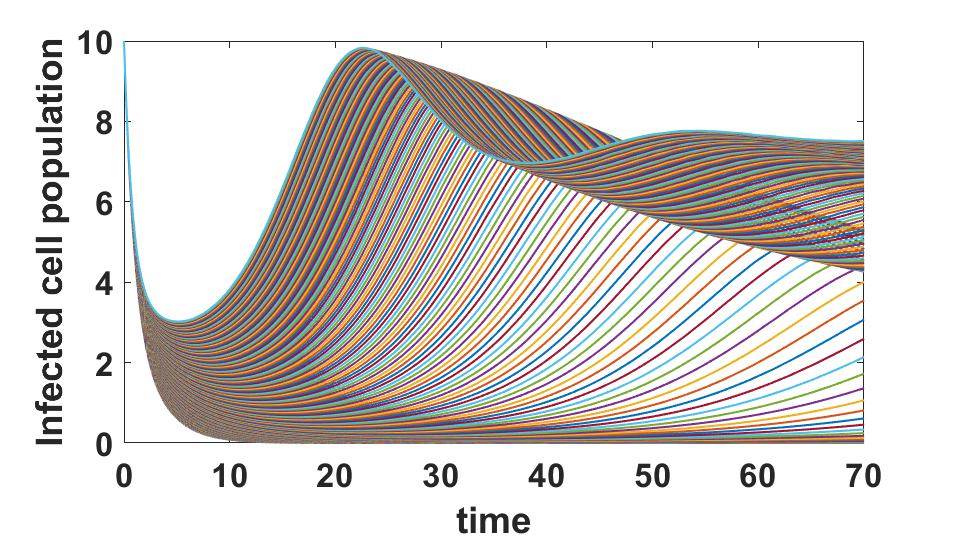}}
			\subcaptionbox*{(c) Interval III}
			{\includegraphics[width=3in, height=1.8in, angle=0]{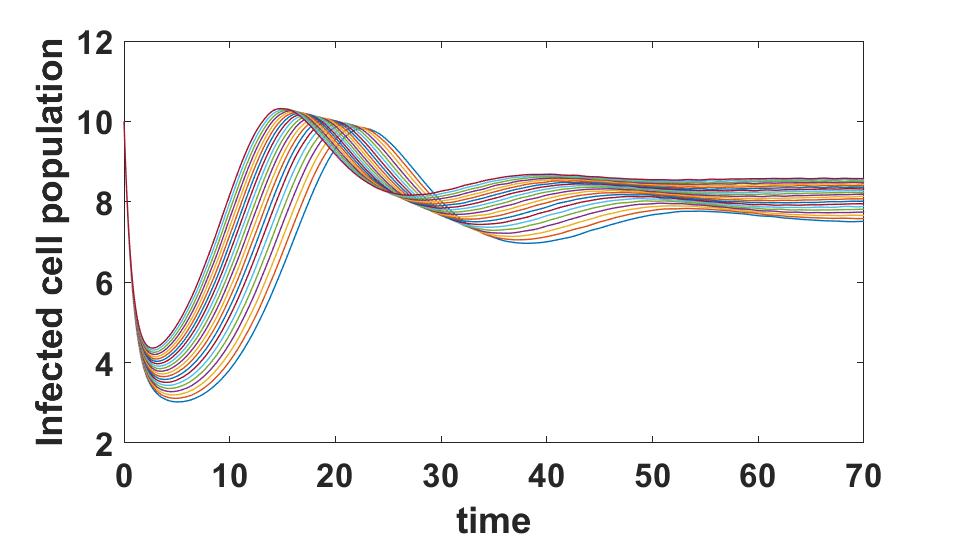}}
            \vspace{1\baselineskip}  
			\caption{Sensitivity Analysis of $\alpha$. Infected cell population in different intervals.}
			\label{sen_alpha}
		\end{center}
	\end{figure}

	\subsubsection{Parameter $\boldsymbol{\mu_{1}}$}
	\begin{enumerate}
		\item \textbf{Interval I : 0.01 to 0.1 :}
		The results related to sensitivity of $\mu_{1}$, varied from 0.01 to 0.1 in step size of 0.001  are shown in the Figure (\ref{sen_mu1_1}). The parameter $\mu_{1}$ is sensitive to the model system in this interval.
		
		\item \textbf{Interval II : 0.1 to 0.4 :}
		The results related to sensitivity of $\mu_{1}$, varied from 0.1 to 0.4 in step size of 0.01  are shown in the Figure (\ref{sen_mu1_2}). The parameter $\mu_{1}$ is sensitive to the model system in this interval.
		
		\item \textbf{Interval III : 0.4 to 1.8 :}
		The results related to sensitivity of $\mu_{1}$, varied from 0.4 to 1.8 in step size of 0.01  are shown in the Figure (\ref{sen_mu1_3}). The parameter $\mu_{1}$ is not sensitive to the model system in this interval.
		
	\end{enumerate}
	
	We conclude from these plots that the parameter $\mu_{1}$ is sensitive in interval I and II and insensitive in interval III. To confirm and validate the same, we have plotted the infected population for each varied value of the parameter $\mu_{1}$ per interval in Figure (\ref{sen_mu1}).
	In similar lines, sensitivity analysis is done for other parameters and the results are summarized in Table (\ref{sen_anl_summary}). The corresponding plots are given in Appendix - A owing to the brevity of the manuscript.

	\begin{figure}[hbt!]
		\begin{center}
			\includegraphics[width=3in, height=1.8in, angle=0]{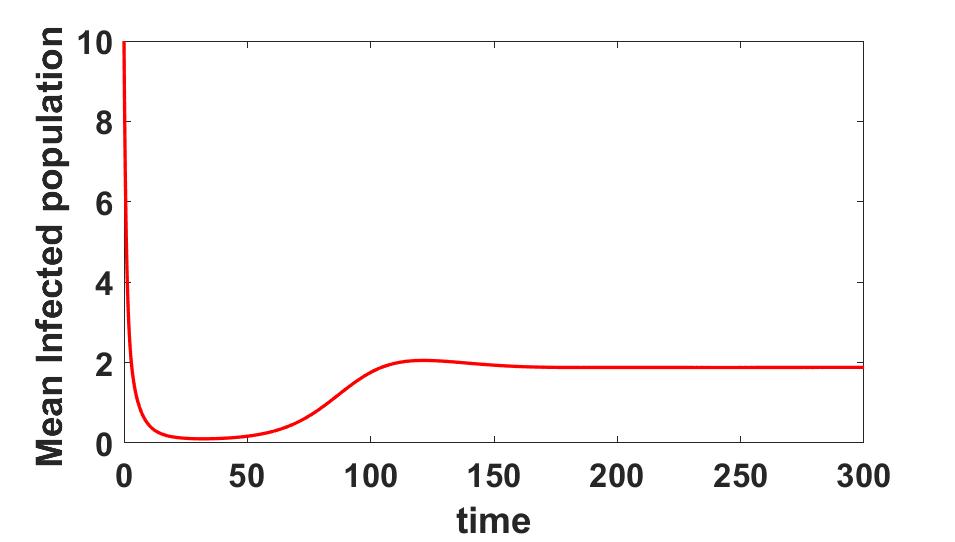}
			\includegraphics[width=3in, height=1.8in, angle=0]{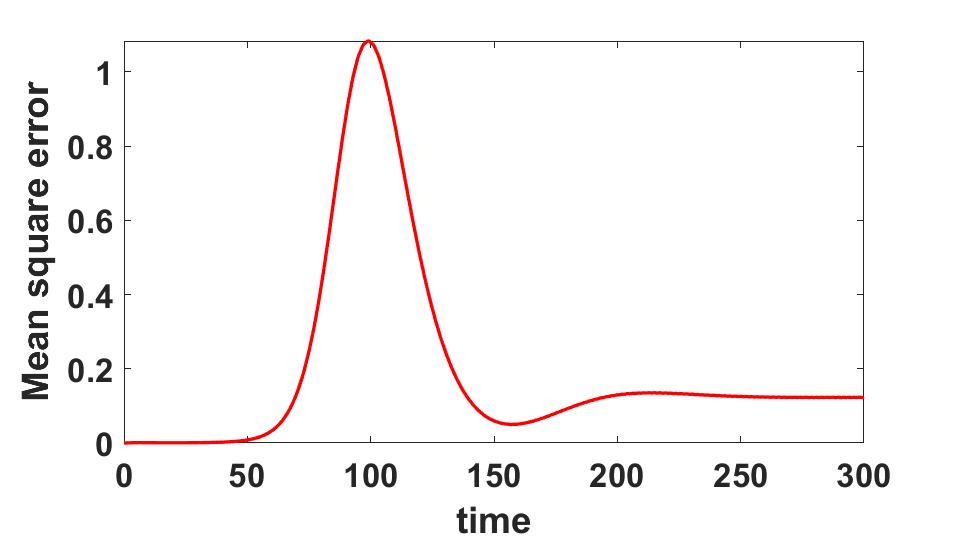}
			\caption{Sensitivity Analysis of $\mu_{1}$ in Interval I.}
			\label{sen_mu1_1}
		\end{center}
	\end{figure}
	
	\begin{figure}[hbt!]
		\begin{center}
			\includegraphics[width=3in, height=1.8in, angle=0]{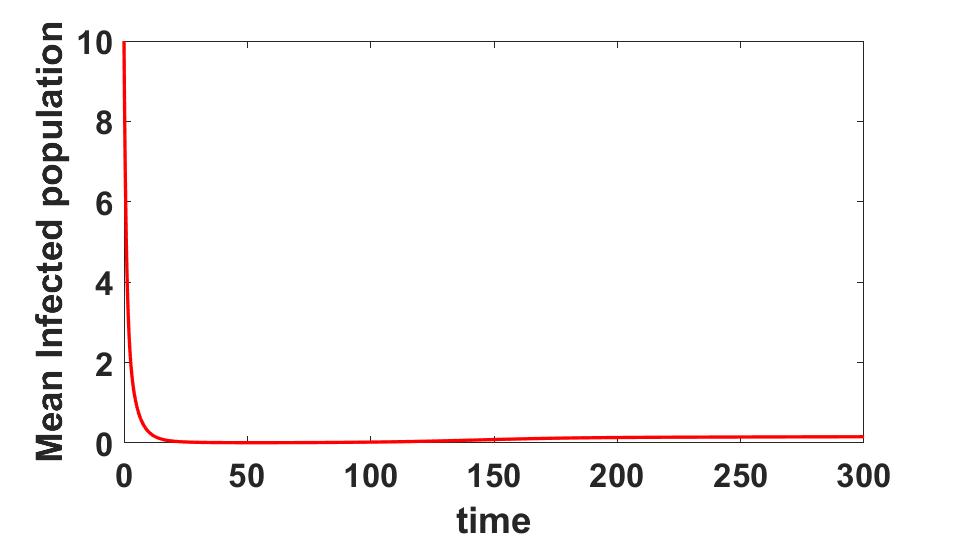}
			\includegraphics[width=3in, height=1.8in, angle=0]{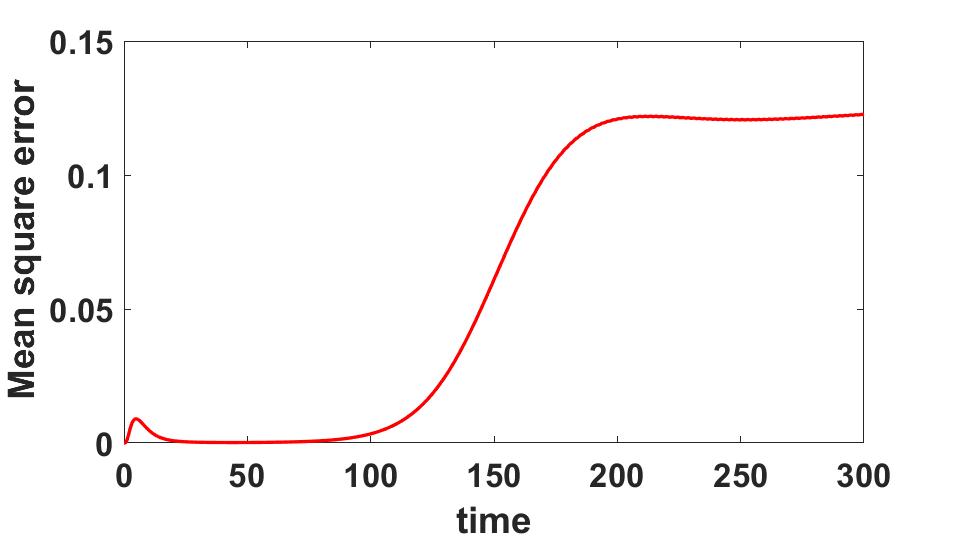}
			\caption{Sensitivity Analysis of $\mu_{1}$ in Interval II.}
			\label{sen_mu1_2}
		\end{center}
	\end{figure}
	
	\begin{figure}[hbt!]
		\begin{center}
			\includegraphics[width=3in, height=1.8in, angle=0]{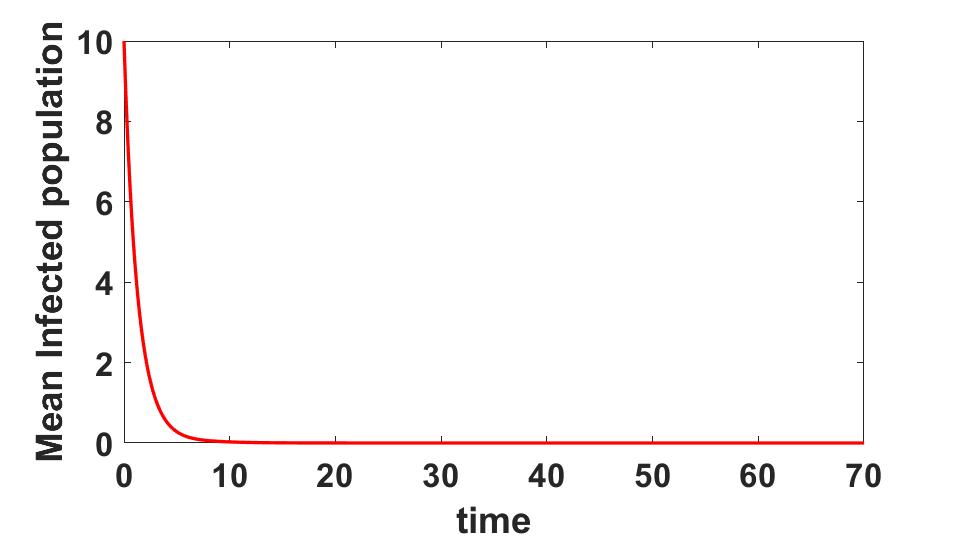}
			\includegraphics[width=3in, height=1.8in, angle=0]{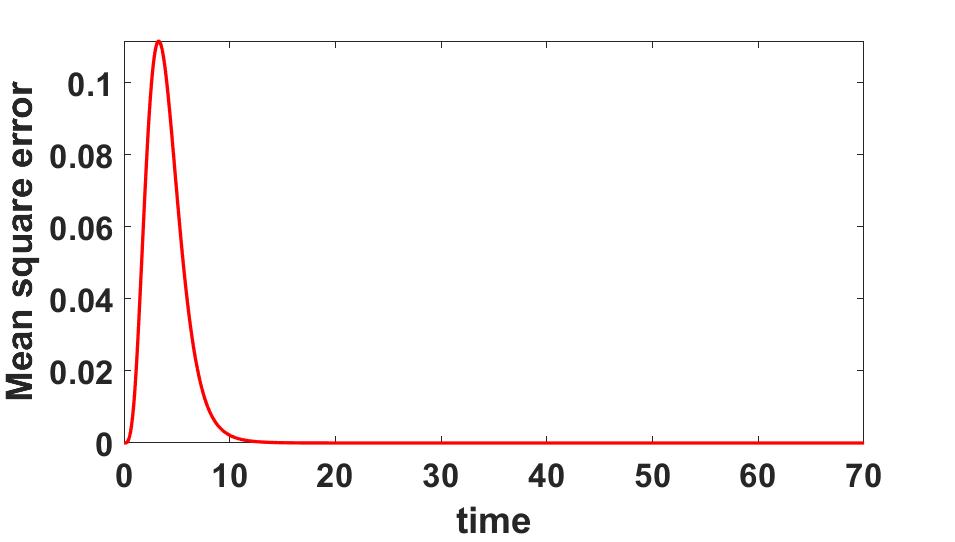}
			\caption{Sensitivity Analysis of $\mu_{1}$ in Interval III.}
			\label{sen_mu1_3}
		\end{center}
	\end{figure}
	
	\begin{figure}[ht]
		\begin{center}
			\subcaptionbox*{(a) Interval I}
			{\includegraphics[width=3in, height=1.8in, angle=0]{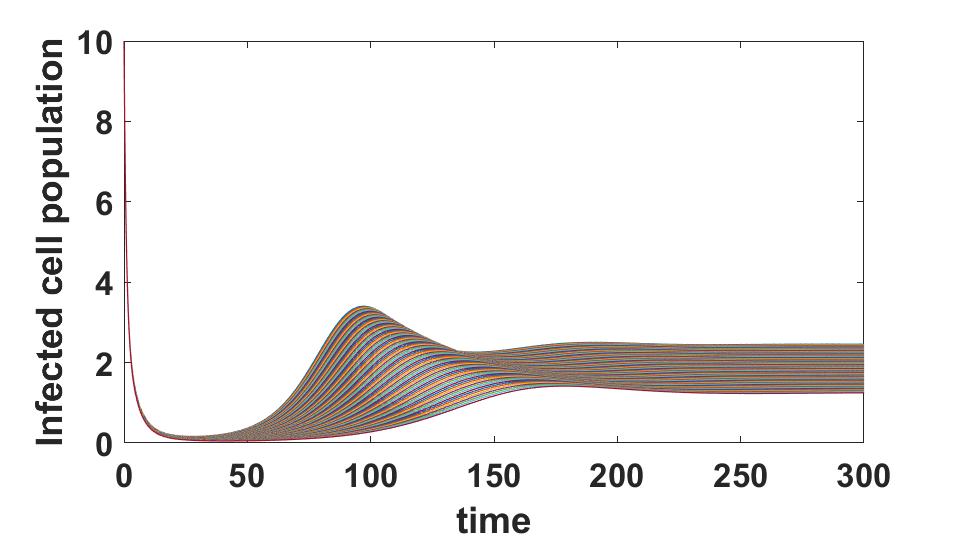}}
		\end{center}
	\end{figure}
	
	\addtocounter{figure}{-1}
	
	\begin{figure}[ht]
		\begin{center}
			\subcaptionbox*{(b) Interval II}
			{\includegraphics[width=3in, height=1.8in, angle=0]{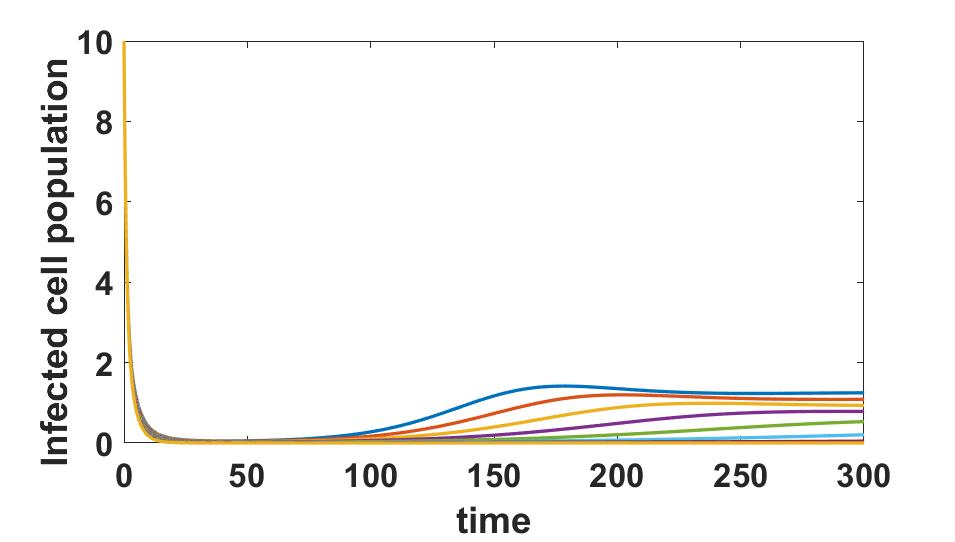}}
			\subcaptionbox*{(c) Interval III}
			{\includegraphics[width=3in, height=1.8in, angle=0]{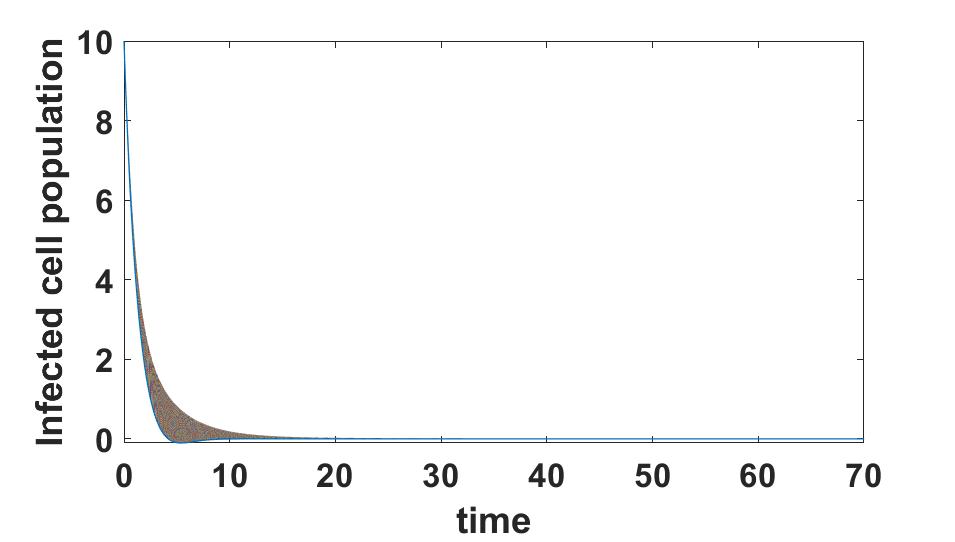}}
			\vspace{1\baselineskip}
			\caption{Sensitivity Analysis of $\mu_{1}$. Infected cell population in different intervals.}
			\label{sen_mu1}
		\end{center}
	\end{figure}
	
\newpage
	
	\vspace{.25cm}
	
	\begin{table}[hbt!]
		\caption{Summary of Sensitivity Analysis}
		\centering
		\label{sen_anl_summary}
		{
			\begin{tabular}{|l|l|l|}
				\hline
				\textbf{Parameter} & \textbf{Interval} & \textbf{Sensitivity} \\
				\hline
				
				$u_{11}$ & 0 to 0.4 & $\times$ \\ \cline{2-3}
				& 0.4 to 0.8 & $\times$ \\ \cline{2-3}
				& 0.8 to 1 & $\times$ \\ 
				\hline
				
				$u_{12}$ & 0 to 1 & $\times$ \\ \cline{2-3}
				& 1 to 2 &  $\times$   \\
				\hline
				
				$\beta$ & 0 to 0.01 & $\times$ \\ \cline{2-3}
				& 0.01 to 0.02 & $\times$ \\ \cline{2-3}
				& 0.02 to 0.03 & $\times$ \\ 
				\hline
				
				$\alpha$ & 0.5 to 1.2 & $\times$ \\ \cline{2-3}
				& 1.2 to 3 & $\checkmark$ \\ \cline{2-3}
				& 3 to 4 & $\times$ \\ 
				\hline
				
				$\omega$ & 5 to 15 & $\times$ \\ \cline{2-3}
				& 20 to 40 & $\times$ \\ \cline{2-3}
				& 40 to 60 & $\times$ \\ 
				\hline
				
				$\mu$ & 0.001 to 0.015 & $\times$ \\ \cline{2-3}
				& 0.01 to 0.03 & $\times$ \\ \cline{2-3}
				& 0.03 to 0.08 & $\times$ \\ 
				\hline
				
				$\mu_{1}$ & 0.01 to 0.1 & $\checkmark$ \\ \cline{2-3}
				& 0.1 to 0.4 & $\checkmark$ \\ \cline{2-3}
				& 0.4 to 1.8 & $\times$ \\ 
				\hline
				
			\end{tabular}
		}
	\end{table}


	\section{Drug Interventions and Optimal Control Studies} \vspace{.25cm}
	{\flushleft{  \textbf{OPTIMAL CONTROL PROBLEM} }} \vspace{.25cm}
	
	In this section, we will formulate an optimal control problem for the model 2 with drug interventions as control. The controls to be considered are:
	
	\begin{enumerate}
		\item \textbf{Drug intervention to boost  immune response:}  The innate immune response that  is inherent in the human body is the first  to counter and  reduce the growth of infected cells and viral load. In this intervention, we provide medication that boosts this innate immune system. This is modelled using two controls based on their role:
		\begin{enumerate}
			\item Control intervention that reduces the growth of the infected cells by releasing the necessary cytokines and chemokines which act on these cells. We denote this intervention by control variable $u_{11}(t)$.
			\item  Control intervention that reduces the  the viral load by releasing the necessary cytokines and chemokines. We denote this intervention by control variable $u_{12}(t)$.
		\end{enumerate}
	    \item \textbf{Drug intervention to prevent viral replication:} We know that one of the main causes of severity of disease caused by viruses is their rapid replication in the human body. We consider the intervention of providing medication that directly acts on the virus cells and prevent its replication which in turn reduces the birth rate of the virus.  We use the control variable $u_{2}(t)$ to denote this intervention.
	\end{enumerate}
	
	The set of all admissible controls is given by \\
	
	$U = \left\{(u_{11}(t),u_{12}(t),u_{2}(t)) : u_{11}(t) \in [0,u_{11} max] , u_{12}(t) \in [0,u_{12} max] , u_{2}(t) \in [0,u_{2} max] ,t \in [0,T] \right\}$
	
	Without medical interventions, $u_{11},\; u_{12},\;$ and $u_2$ are just constant parameters with $u_2 = 0$. $u_{11}$ and $ u_{12}$ have some value based on the inherent release of these cytokines and chemokines by the body. In the control problem, these variables are considered as functions of time. The upper bounds of control variables are based on the resource limitation based on availability and  the limit to which these drugs would be prescribed to the patients.
	
	Since these antiviral drugs administered are not inherently present in the human body and  being foreign particles to the body, there could be  side effects once administered. Also cost of these drugs could be an issue that needs to be addressed.
	
	Based on these we now propose and  define the optimal control problem with the goal to reduce the cost functional 
	
	\begin{equation}
		J(u_{11}(t),u_{12}(t),u_{2}(t)) = \int_{0}^{T} (I(t)+V(t)+A_{1}u_{11}(t)^2+A_{2}u_{12}(t)^2+A_{3}u_{2}(t)^2) dt   \label{obj}
	\end{equation} 
	
	such that $\textbf{u} = (u_{11}(t),u_{12}(t),u_{2}(t)) \in U$ \\
	
	subject to the system 
	
	\begin{eqnarray}
		\frac{dS}{dt}  & = &    \omega  -\beta SV  - \mu S   \nonumber \\
		\frac{dI}{dt} & = &  \beta SV  - \mu I  - (u_{11}(t))I \nonumber \\
		\frac{dV}{dt} & = &  (\alpha - u_{2}(t)) I    - \mu_{1} V - (u_{12}(t))V   \label{sys}
	\end{eqnarray}
	
	The integrand of the cost function (\ref{obj}), denoted by 
	$$L(S,I,V,u_{11},u_{12},u_2) = (I(t)+V(t)+A_{1}u_{11}(t)^2+A_{2}u_{12}(t)^2+A_{3}u_{2}(t)^2)$$
	
	is called the Lagrangian or the running cost.
	
	Here, the cost function represents the number of infected cells and viral load throughout the observation period,	and the side effects of the drug on the body. Effectively, we want to minimize the infected cells and the virus load in the body with the optimal medication that is also least harmful to the body. Since the drugs administered have multiple effects, the non-linearity for the control variables in the objective become justified \cite{Joshi2002}.
	
	The admissible solution set for the Optimal Control Problem (\ref{obj})-(\ref{sys}) is given by
	
	$$\Omega = \left\{ (S, I, V, u_{11}, u_{12}, u_2)\; | \; S,  I \ and \ V  \text{that satisfy (\ref{sys}), } \forall \; \textbf{u} \in U \right\}$$

	
	
	{\textbf{EXISTENCE OF OPTIMAL CONTROL}}\vspace{.25cm}

	We will show the existence of optimal control functions that minimize the cost functions within a finite time span $[0,T]$ showing that we satisfy the conditions stated in Theorem 4.1 of \cite{Wendell}.
	
	\begin{theorem}
		There exists a 3-tuple of optimal controls $(u_{11}^{*}(t) , u_{12}^{*}(t) , u_{2}^{*}(t))$ in the set of admissible controls U such that the cost functional is minimized i.e., 
		
		$$J[u_{11}^{*}(t) , u_{12}^{*}(t) , u_{2}^{*}(t)] = \min_{(u_{11}^{*} , u_{12}^{*} , u_{2}^{*}) \in U} \bigg \{ J[u_{11},u_{12},u_{2}]\bigg\}$$ 
		corresponding to the optimal control problem (\ref{obj})-(\ref{sys}).
	\end{theorem}
	
	
	
	\begin{proof}
		
		 In order to show the existence of optimal control functions, we will show that the following conditions are satisfied : 
		
		\begin{enumerate}
			\item  The solution set for the system (\ref{sys}) along with bounded controls must be non-empty, $i.e.$, $\Omega \neq \phi$.
			
			\item  U is closed and convex and system should be expressed linearly in terms of the control variables with coefficients that are functions of time and state variables.
			
			\item The Lagrangian L should be convex on U and $L(S,I,V,u_{11},u_{12},u_2) \geq g(u_{11},u_{12},u_{2})$, where $g(u_{11},u_{12},u_{2})$ is a continuous function of control variables such that $|(u_{11},u_{12},u_{2})|^{-1} g(u_{11},u_{12},u_{2}) \to \infty$ whenever  $|(u_{11},u_{12},u_{2})| \to \infty$, where $|.|$ is an $l^2(0,T)$ norm.
		\end{enumerate}

		Now we will show that each of the conditions are satisfied : 
		
		1. From Positivity and boundedness of solutions of the system(\ref{sys}), all solutions are bounded for each bounded control variable in U.
		
		Also,the right hand side of the system (\ref{sys}) satisfies Lipschitz condition with respect to state variables. 
		
		Hence, using the positivity and boundedness condition and the existence of solution from Picard-Lindelof Theorem\cite{makarov2013picard}, we have satisfied condition 1.
		
		2. $U$ is closed and convex by definition. Also, the system (\ref{sys}) is clearly linear with respect to controls such that coefficients are only state variables or functions dependent on time. Hence condition 2 is satisfied.
		
		3. Choosing $g(u_{11},u_{12},u_{2}) = c(u_{11}^{2}+u_{12}^{2}+u_{2}^{2})$ such that $c = min\left\{A_{1},A_{2},A_{3}\right\}$, we can satisfy the condition 3.
		
		Hence there exists a control 3-tuple $(u_{11}^{*},u_{12}^{*},u_{2}^{*})\in U$ that minimizes the cost function (\ref{obj}).
	\end{proof}

	\textbf{CHARACTERIZATION OF OPTIMAL CONTROL}\vspace{.25cm}
	
	We will obtain the necessary conditions for optimal control functions using the Pontryagin's Maximum Principle \cite{liberzon2011calculus} and also obtain the characteristics of the optimal controls.
	
	The Hamiltonian for this problem is given by 
	
	$$H(S,I,V,u_{11},u_{12},u_{2},\lambda) := L(S,I,V,u_{11},u_{12},u_{2}) + \lambda_{1} \frac{\mathrm{d} S}{\mathrm{d} t} +\lambda _{2}\frac{\mathrm{d} I}{\mathrm{d} t}+ \lambda _{3} \frac{\mathrm{d} V}{\mathrm{d} t}$$
	
	Here $\lambda$ = ($\lambda_{1}$,$\lambda_{2}$,$\lambda_{3}$) is called co-state vector or adjoint vector.
	
	Now the Canonical equations that relate the state variables to the co-state variables are  given by 
	
	\begin{equation}
	\begin{aligned}
	 \frac{\mathrm{d} \lambda _{1}}{\mathrm{d} t} &= -\frac{\partial H}{\partial S}\\
	 \frac{\mathrm{d} \lambda _{2}}{\mathrm{d} t} &= -\frac{\partial H}{\partial I}\\
	 \frac{\mathrm{d} \lambda _{3}}{\mathrm{d} t} &= -\frac{\partial H}{\partial V}
	\end{aligned}
	\end{equation}

	
	Substituting the Hamiltonian value gives the canonical system 
	
	\begin{equation}
	\begin{aligned}
	\frac{\mathrm{d} \lambda _{1}}{\mathrm{d} t} &= \lambda _{1}(\beta V+\mu)-\lambda _{2} \beta V\\
	\frac{\mathrm{d} \lambda _{2}}{\mathrm{d} t} &= -1+\lambda _{2}(\mu+u_{11})-\lambda _{3} (\alpha -u_{2})\\
	\frac{\mathrm{d} \lambda _{3}}{\mathrm{d} t} &= -1+\lambda _{1}\beta S-\lambda _{2}\beta S+\lambda _{3} (\mu _{1}+b_{7}+u_{12})
	\end{aligned}
	\end{equation}
	
	along with transversality conditions
	$ \lambda _{1} (T) = 0, \  \lambda _{2} (T) = 0, \  \lambda _{3} (T) = 0. $
	
	Now, to obtain the optimal controls, we will use the Hamiltonian minimization condition 
	$ \frac{\partial H}{\partial u_{i}}$ = 0 , at  $u_{i} = u_{i}^{*}$  for i = 11, 12 , 2.
	
	Differentiating the Hamiltonian and solving the equations, we obtain the optimal controls as 
	
	\begin{eqnarray*}
	u_{11}^{*} &=& \min\bigg\{ \max\bigg\{\frac{\lambda _{2}I}{2A_{1}},0 \bigg\}, u_{11}max\bigg\}\\
	u_{12}^{*} &= &\min\bigg\{ \max\bigg\{\frac{\lambda _{3}V}{2A_{2}},0 \bigg\}, u_{12}max\bigg\}\\
	u_{2}^{*}& = &\min\bigg\{ \max\bigg\{\frac{\lambda _{3}I}{2A_{3}},0 \bigg\}, u_{2}max\bigg\}
	\end{eqnarray*}

	{\flushleft{  \textbf{NUMERICAL SIMULATIONS} }} \vspace{.25cm}
	
	In this section, we perform numerical simulations to understand the efficacy of multiple drug interventions. This is done by studying the effect of control  on the  dynamics of the system. These simulations also validate the theoretical results obtained in the previous section.
	
The various combinations of controls considered are: 
	
	1. Providing medication that only boosts the innate immune system by reducing the number of infected cells and viral load. \vspace{.25cm}
	
	2. Providing medication and treatment that only prevents viral replication. \vspace{.25cm}
	
	3. Providing medication and treatment that execute both the above. \vspace{.25cm}
	
	For our simulations, we have taken the total number of days as $T = 30.$ The parameter values considered are as follows :
	$\omega$ = 10, $\mu$= 0.05, $\mu_{1}$ = 1.1, $\beta$ = 0.005, $\alpha$ = 0.5
	
	We first solve the state system numerically using Fourth Order Runge-Kutta method in MATLAB without any interventions. We take the initial values of state variables to be $S(0) = 3.2 \times 10^5, I(0) = 0, V(0) = 5.2$ \cite{ben2015minimal, ea2020host} with the control parameters as constant values $u_{11}$ = 0.8866, $u_{12}$ = 0.56, $u_{2}$ = 0.
	
	Now, to simulate the system with controls, we use the Forward-Backward Sweep method stating with the initial values of controls and solve the state system forward in time. Following this we solve the adjoint state system backward in time due to the transversality conditions, using the optimal state variables and initial values of optimal control.
	
	Now, using the values of adjoint state variables, the values of optimal control are updated and with these updated control variables, we go through this process again. We continue this till the convergence criterion is met \cite{liberzon2011calculus}. The positive weights chosen for objective coefficients are $A_{1}$ = 80, $A_{2}$ = 80, $A_{3}$ = 7480. $A_{3}$ is chosen high compared to $A_1$ and $A_2$ as the effort to stop the viral replication process is higher than the effort to enhance innate immune response as it is already existing in human body.
	\begin{figure}[h!]
		\centering
		\includegraphics[height = 6.5cm, width = 15cm]{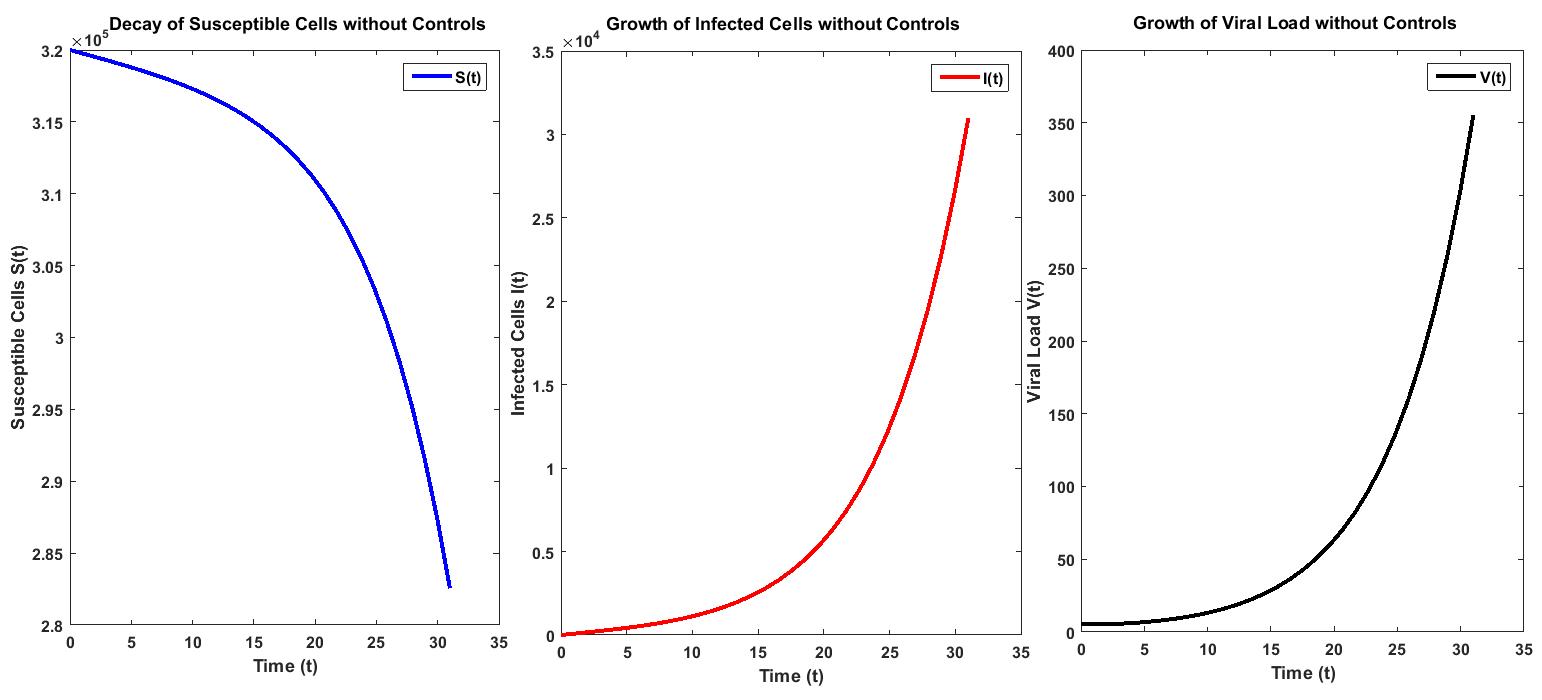} 
		\vspace{3mm}
		\caption{$S, I, V$  without controls}
	\end{figure}
	

	Figure 15 shows the change in the cell population $S(t), I(t), V(t)$ respectively with time. We observed that the susceptible cells reduce and the infected cells increase exponentially due to the increase in viral load over a period of time.
	
	\begin{figure}[hbt!]
		\centering
		\includegraphics[height = 6.5cm, width = 15cm]{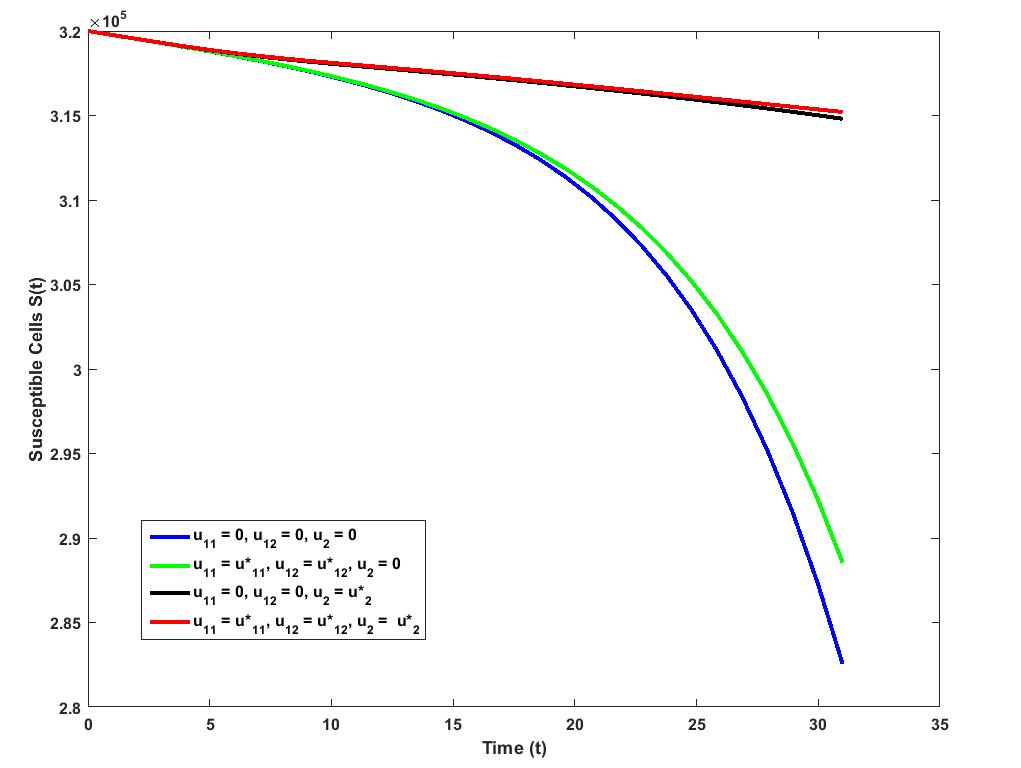} 
		\vspace{3mm}
		\caption{$S$ under optimal controls u$_{1}^{*}$, u$_{2}^{*}$}
	\end{figure}  
	
	Figure 16 shows the behaviour of susceptible cells under all possible combinations of control. We observe that in the presence of immune boosting medication only, there is a slight reduction in the number of susceptible cells getting infected by the virus but when viral replication is prevented the impact is more.  Even lesser damage is seen on susceptible cells when all control interventions are applied together.
	
	\begin{figure}[hbt!]
		\centering
		\includegraphics[height = 6.5cm, width = 15cm]{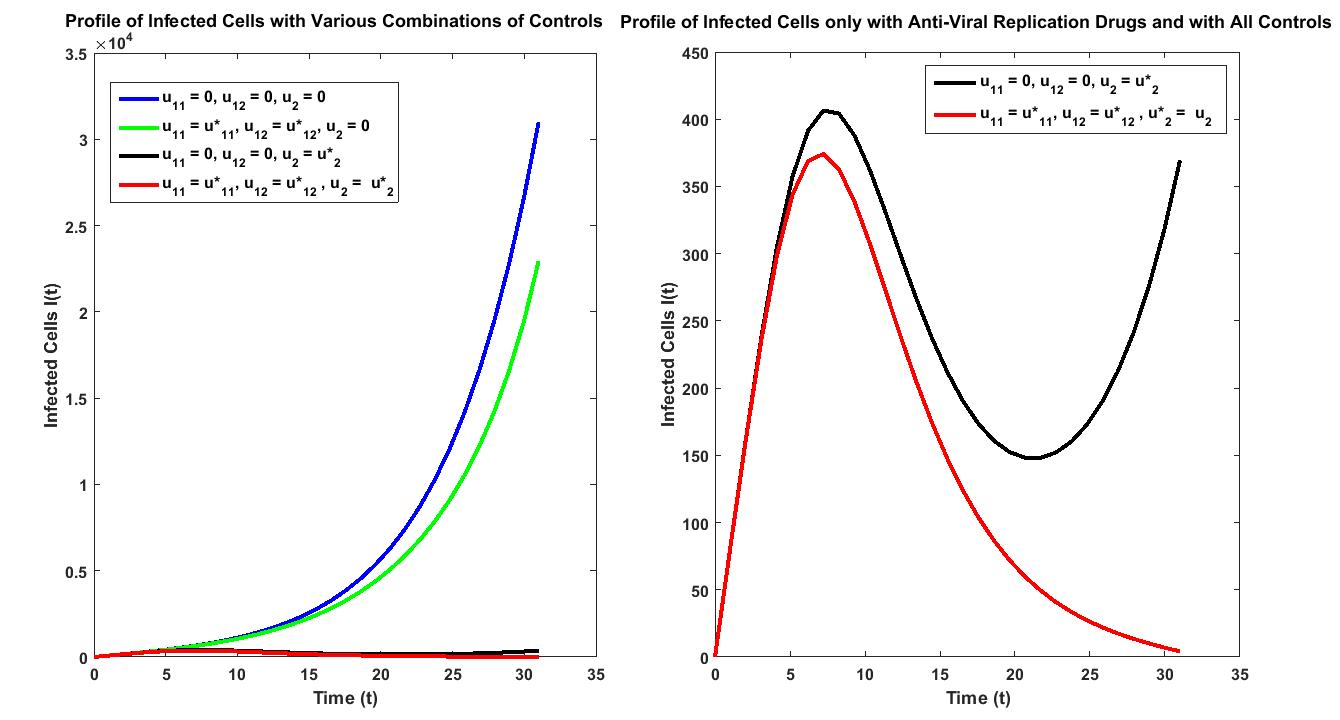} 
		\vspace{3mm}
		\caption{$I$ under optimal controls u$_{1}^{*}$, u$_{2}^{*}$}
	\end{figure}

	In figure 17 the first frame shows how infected cells grow or decay with various combinations of controls. The increase in infected cells is slightly lesser in the presence of immunity boosters. There is a huge difference if viral replication preventing medicines are used. Here too, the best results are obtained only when all the controls are applied. The second frame in figure 17 gives detailed view of the cases when only antiviral replication medicines are used and all 3 controls are used. We see that when only control u$_{2}$ is used, the infected cells increase slowly, and reduce later but increase again after 20 days. This may be harmful to the patient. When all 3 controls are used, the infected cells reach peak around the 8th day and start decaying then on and even become nearly zero after 30 days. 
	
	\begin{figure}[h!]
		\centering
		\includegraphics[height = 6.5cm, width = 15cm]{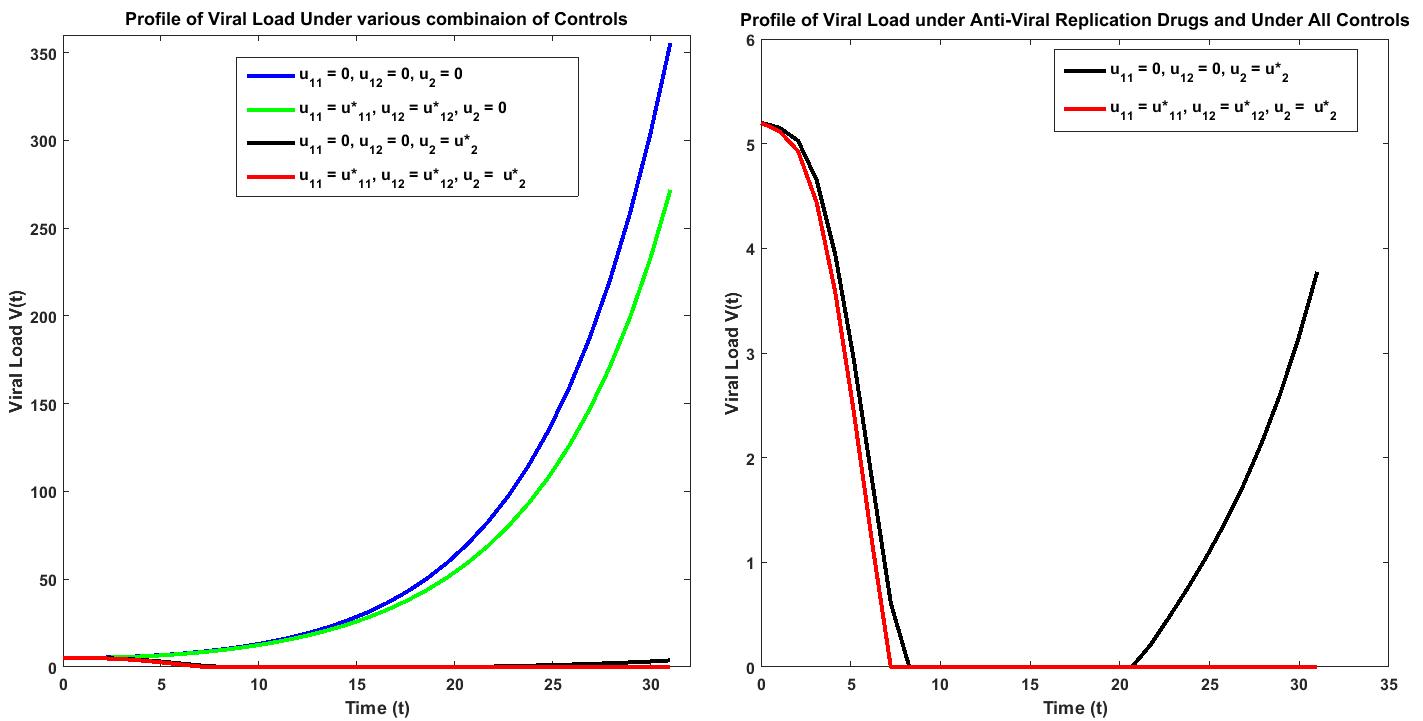} 
		\vspace{3mm}
		\caption{$V$ under optimal controls u$_{1}^{*}$, u$_{2}^{*}$}
	\end{figure}

	In figure 18 the first frame shows the viral load under all combinations of control interventions. When no medication is provided, there is exponential increase in the viral load but when only immunity boosting medication is provided ($u_2 = 0$), there is a little reduction seen in the increase of viral load. The  combinations of both these drug interventions show much better results. From the second frame in figure 18 it can be seen that   when only anti-viral replication medication is provided ($u_{11} = u_{12} = 0$),  the viral load becomes very less by the 5th day but due to the absence of immunity boosters, the viral load tends to increase around 25th day. The best results are shown when all the control interventions are administered together. The detailed view of the figure explains how the viral population tends to become nearly zero around 7th day and remains there throughout. These results are in line with the clinical findings discussed in  \cite{caly2020fda}.
	
	 Thus the optimal control studies and the numerical simulations help us to infer the following :
    
    1. Innate immunity boosters do affect the way infection spreads in the body but they become effective only along with the prevention of viral replication medicines.
    
    2. Viral replication, when prevented, reduces the infected cells and the viral load drastically but the immune system needs to be sufficiently boosted if we want to completely cure the patient.
    
    3. When all controls are used effectively, then the patient can be cured completely, with minimal/optimal dosage of drugs which minimizes the side effects caused to the patient when administered.

\section{discussions and conclusions} \vspace{.25cm}

\qquad The outbreak of novel coronavirus in Wuhan, China marked the introduction of a virulent coronavirus into human society. On Feb. 11, 2020, the World Health Organization named novel corona viral pneumonia induced disease as Coronavirus disease (COVID-19), which is caused by Severe Acute Respiratory Syndrome coronavirus-2 (SARS-CoV-2). Soon this grew into a global pandemic. As on 02 May 2020, 2, 37, 996  people lost their lives and more than 33 lakh people have been affected due to COVID-19 all over the world  \cite{1}. Although a lot of research is being done, effective approaches to treatment and epidemiological control are still lacking.

\vspace{.25cm}

\quad In this context, the invivo mathematical modelling studies can be extremely helpful in understanding the efficacy of the drug interventions. These studies also can be helpful in understanding the role of the crucial inflammatory mediators and the behaviour of immune response towards this novel corona virus. Motivated by these facts, in this paper, we study the invivo dynamics of Covid-19. \vspace{.25cm}

\quad  Based on the pathogenesis of Covid-19, we have proposed two models. The first model deals with natural history and the course of infection while the second model incorporates the drug interventions. The results of these studies show that the disease system admits two steady states: one being the disease free equilibrium and the other being the infected equilibrium. The dynamics of the system show that the disease takes it course to one of these states based on the reproduction number $R_0$. Specifically when $R_0 < 1$ the system tends to stabilize around the disease free equilibrium and when $R_0 > 1$ the system tends to stabilize around the infected equilibrium.  The system also undergoes a trans-critical bifurcation at $R_0 = 1$. This result is inline with the conclusions made at the population level for Covid-19 \cite{khan2020modeling}. From the sensitivity analysis it is seen  that the burst rate of virus particles and the natural death rate of the virus are sensitive parameters of the system. From the sensitivity analysis it is seen that the burst rate of virus particles and the natural death rate of the virus are the sensitive parameters of the system. We also validate the proposed model 1 using two-parameter heat plots that reproduce the characteristics of Covid-19.\vspace{.25cm}

\quad Results from the optimal control studies suggest that the antiviral drugs that target on viral replication and the drugs that enhance the immune system response both reduce the infected cells and viral load when taken individually. In particular, the antiviral drugs that target viral replication seem to yield better results than the drugs that enhance the innate immune response. From figure 18, it can be seen that on administering control intervention $u_2$ there is a substantial decrease in the viral load from day 2. This result validates the clinical findings in  \cite{caly2020fda} which states that "$Ivermectin$, an FDA-approved anti-parasitic previously shown to have broad-spectrum anti-viral activity in vitro, is an inhibitor of the causative virus (SARS-CoV-2), with a single addition to Vero-hSLAM cells 2 h post infection with SARS-CoV-2 able to effect ~5000-fold reduction in viral RNA at 48 h." \vspace{.25cm}

\quad When applied in combination, these drugs yield the best possible results. Hence, the optimal control strategy and the best drug regime would be to use the combination of both these drugs which help in patient's recovery with minimal/optimal dosage that reduce the side effects caused due to these drugs.  \vspace{.25cm}

\quad This invivo modelling study involving the crucial biomarkers of Covid-19 is first of its kind for Covid-19 and the results obtained from this can be helpful to researchers, epidemiologists, clinicians and doctors working in this field.

\vspace{.25cm}

\section{future research} \vspace{.25cm}

This work is an initial attempt to understand the basic dynamics of Covid-19 and its course. The consequences and the outcomes of the two main functions of antiviral drug interventions is modelled and discussed. Further research can be focused on incorporating the side effects of these drugs into the model.  Future studies can also focus on the various other drug interventions along with their effects on Covid-19 dynamics.

{\flushleft{  \textbf{ACKNOWLEDGEMENTS} }}\vspace{.25cm}

The authors from SSSIHL and SSSHSS dedicate this paper to the founder chancellor of SSSIHL, Bhagawan Sri Sathya Sai Baba. The corresponding author also dedicates this paper to his loving elder brother D. A. C. Prakash who still lives in his heart and the first author also dedicates this paper to his loving father  Purna Chhetri . \vspace{.25cm}

\lhead{\emph{Bibliography}}
\bibliographystyle{amsplain}
\bibliography{references}

		\section{appendix - a}
	
	{\flushleft{  \textbf{SENSITIVITY PLOTS FOR OTHER PARAMETERS} }}\vspace{.25cm}	
		
		For all the plots in this section, the time scale is the following: $x-$axis: 10 units = 1 day, $y-$ axis: 1 unit = 1 cell. 
		
		\subsection{Parameter $\boldsymbol{u_{12}}$}
		
		\begin{figure}[hbt!]
			\begin{center}
				\includegraphics[width=3in, height=1.8in, angle=0]{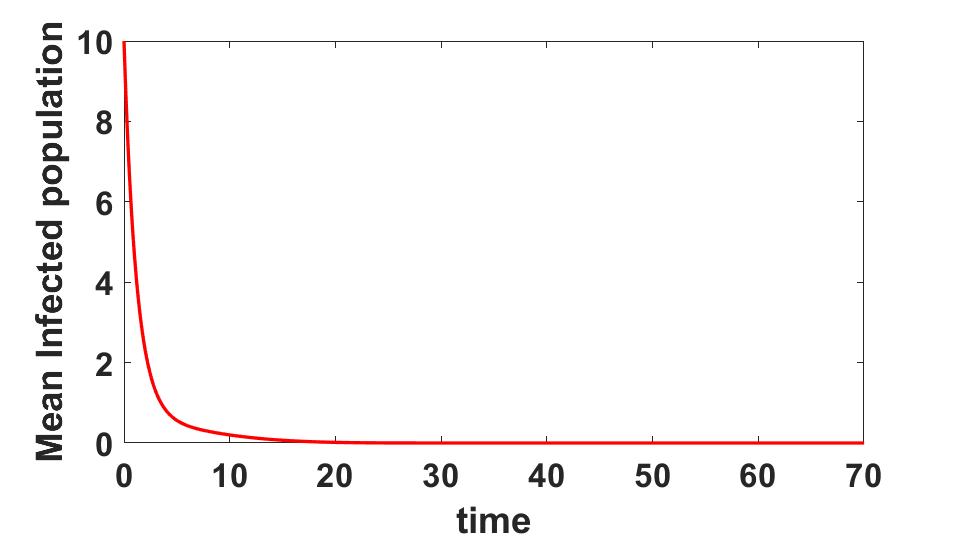}
				\includegraphics[width=3in, height=1.8in, angle=0]{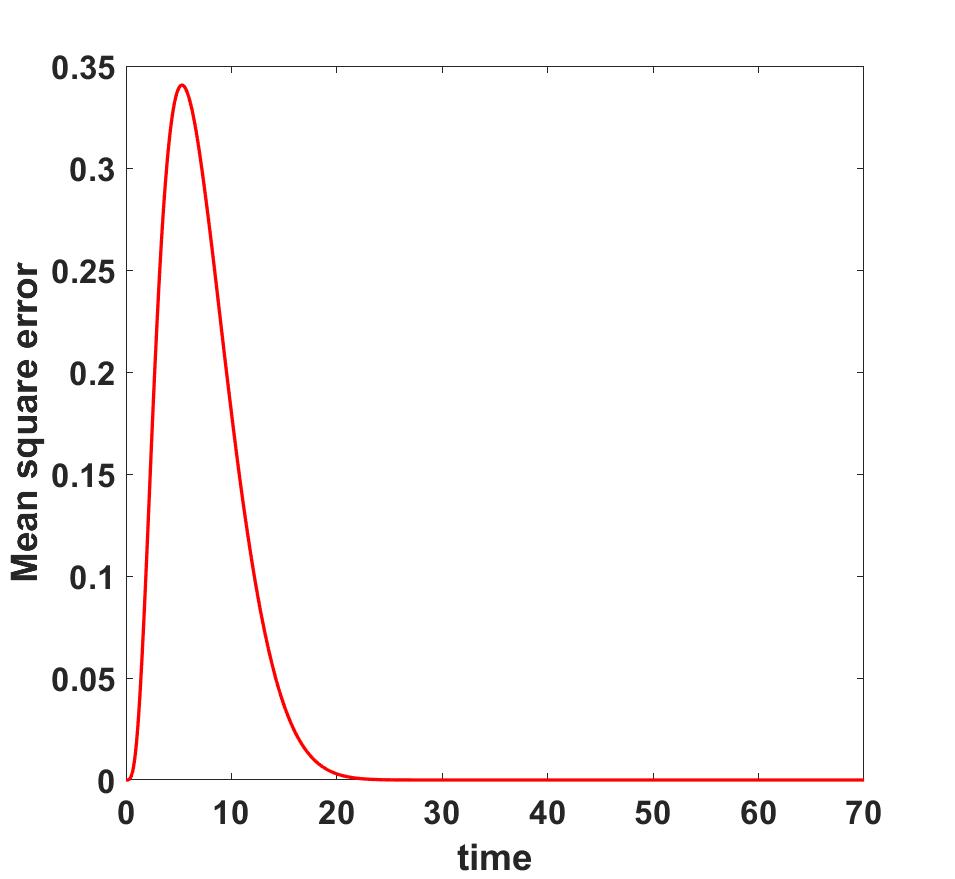}
				\caption{Sensitivity Analysis of $u_{12}$ in Interval I.}
				\label{sen_u12_1}
			\end{center}
		\end{figure}
		
		\begin{figure}[hbt!]
			\begin{center}
				\includegraphics[width=3in, height=1.8in, angle=0]{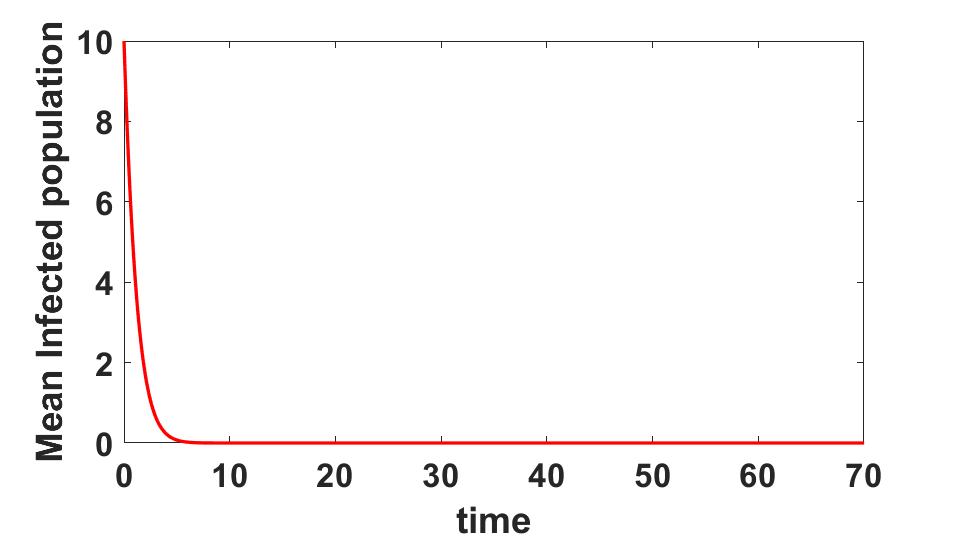}
				\includegraphics[width=3in, height=1.8in, angle=0]{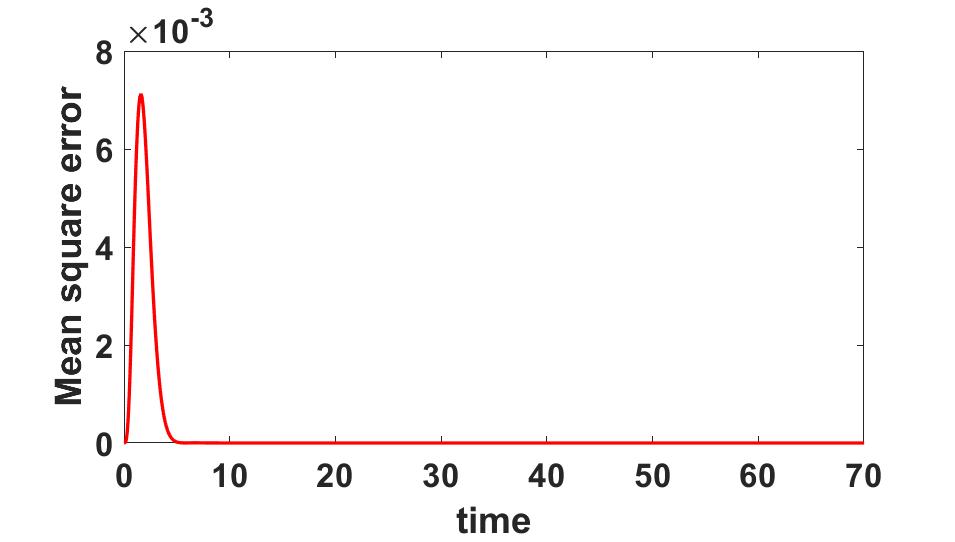}
				\caption{Sensitivity Analysis of $u_{12}$ in Interval II.}
				\label{sen_u12_2}
			\end{center}
		\end{figure}
		
		\begin{figure}[hbt!]
			\begin{center}
				\subcaptionbox*{(a) Interval I}
				{\includegraphics[width=3in, height=1.8in, angle=0]{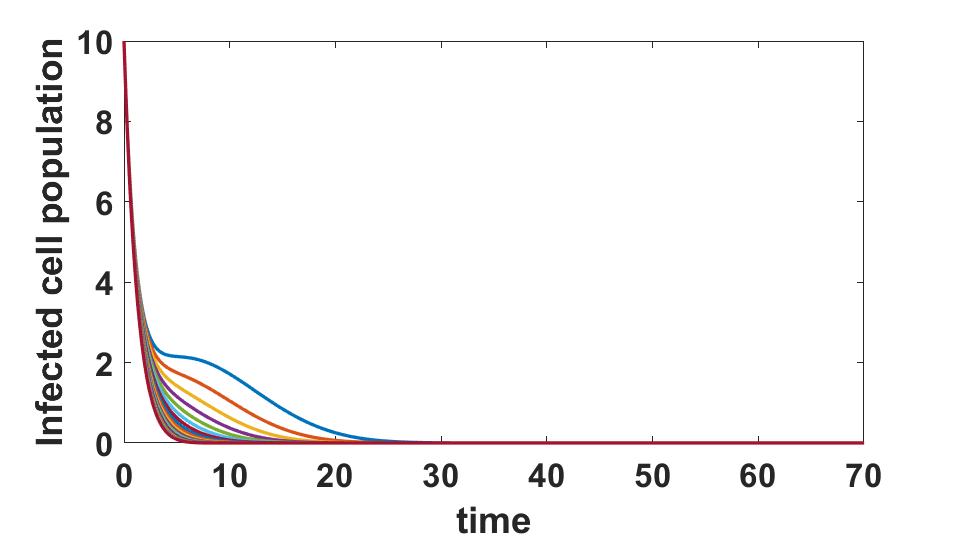}}
				\subcaptionbox*{(b) Interval II}
				{\includegraphics[width=3in, height=1.8in, angle=0]{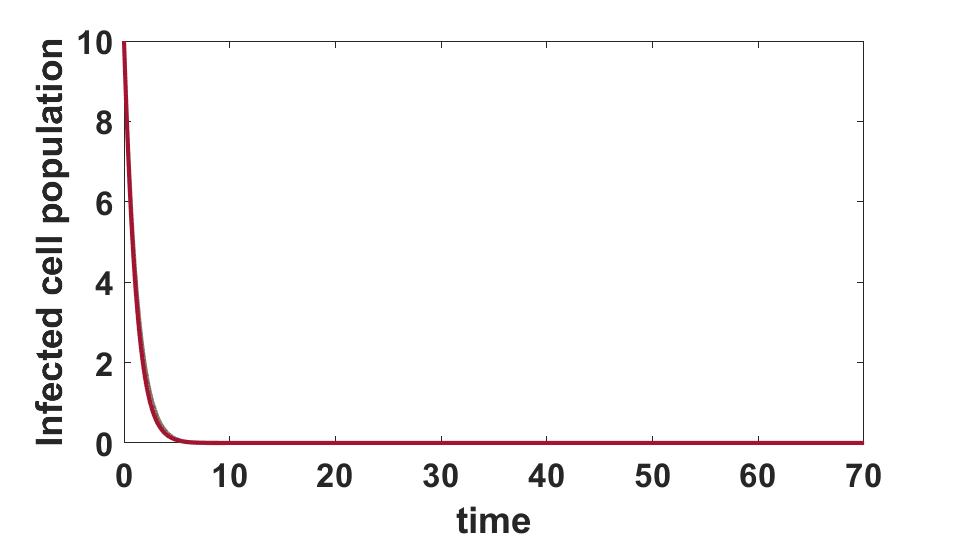}}
				\vspace{1\baselineskip}
				\caption{Sensitivity Analysis of $u_{12}$. Infected cell population in different intervals.}
				\label{sen_u12}
			\end{center}
		\end{figure}
		
		\newpage
		
		\subsection{Parameter $\boldsymbol{u_{11}}$}
		
		\begin{figure}[hbt!]
			\begin{center}
				\includegraphics[width=3in, height=1.8in, angle=0]{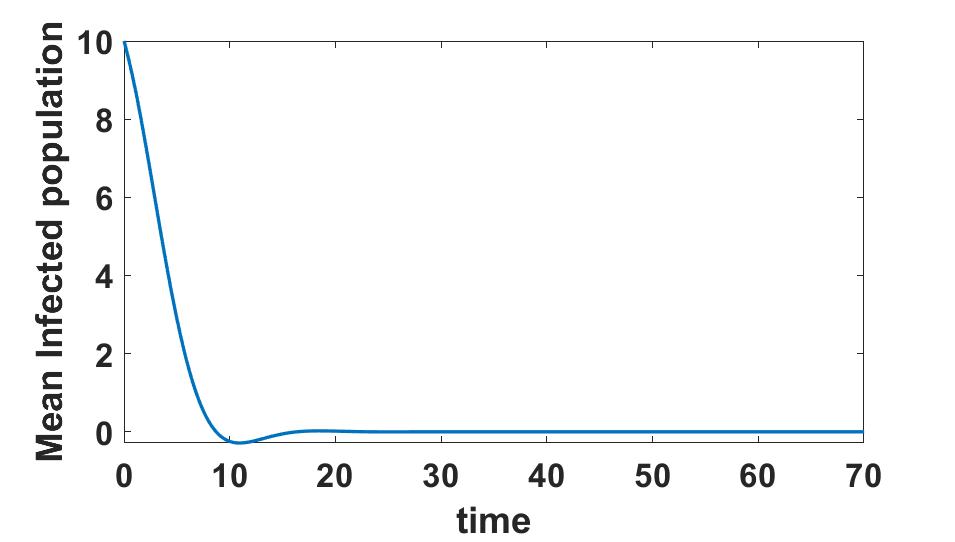}
				\includegraphics[width=3in, height=1.8in, angle=0]{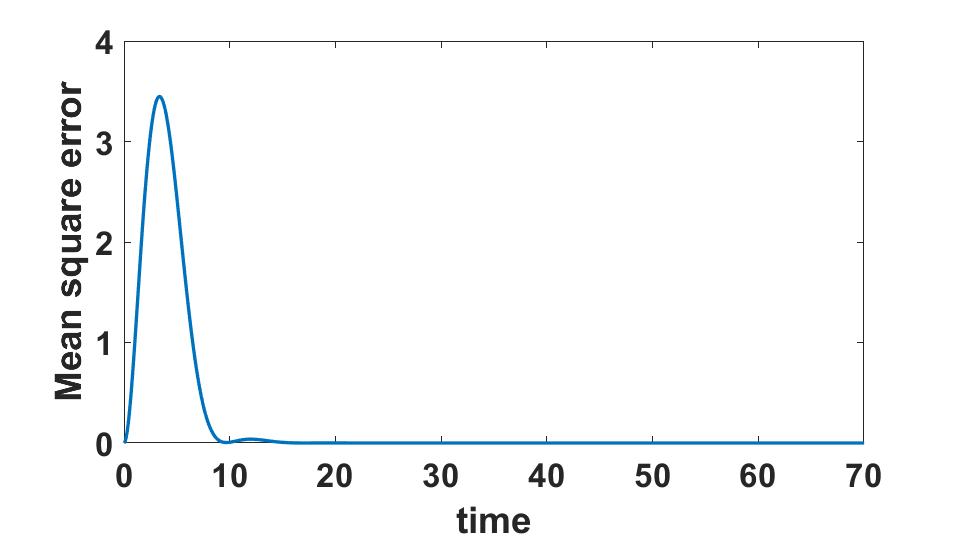}
				\caption{Sensitivity Analysis of $u_{11}$ in Interval I.}
				\label{sen_u11_1}
			\end{center}
		\end{figure}
		
		\begin{figure}[hbt!]
			\begin{center}
				\includegraphics[width=3in, height=1.8in, angle=0]{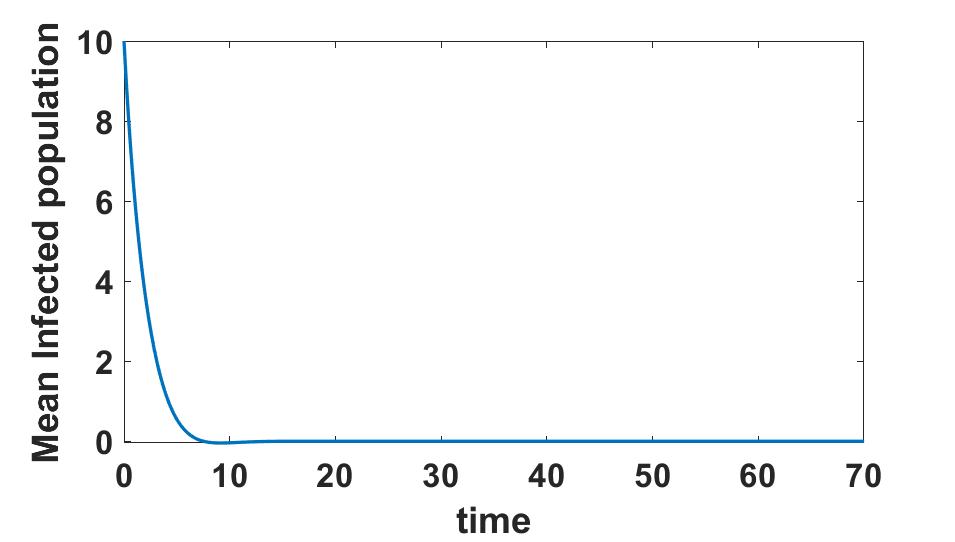}
				\includegraphics[width=3in, height=1.8in, angle=0]{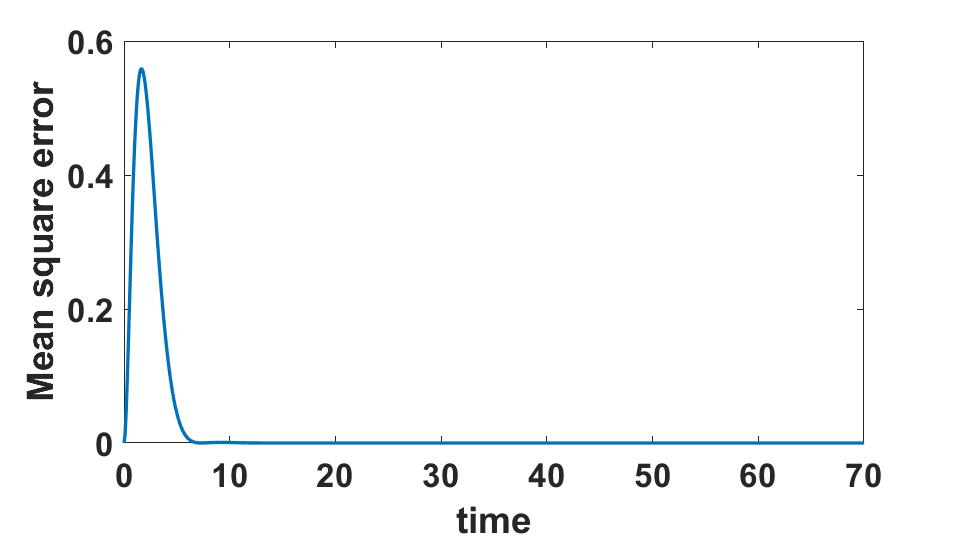}
				\caption{Sensitivity Analysis of $u_{11}$ in Interval II.}
				\label{sen_u11_2}
			\end{center}
		\end{figure}
		
		\begin{figure}[hbt!]
			\begin{center}
				\includegraphics[width=3in, height=1.8in, angle=0]{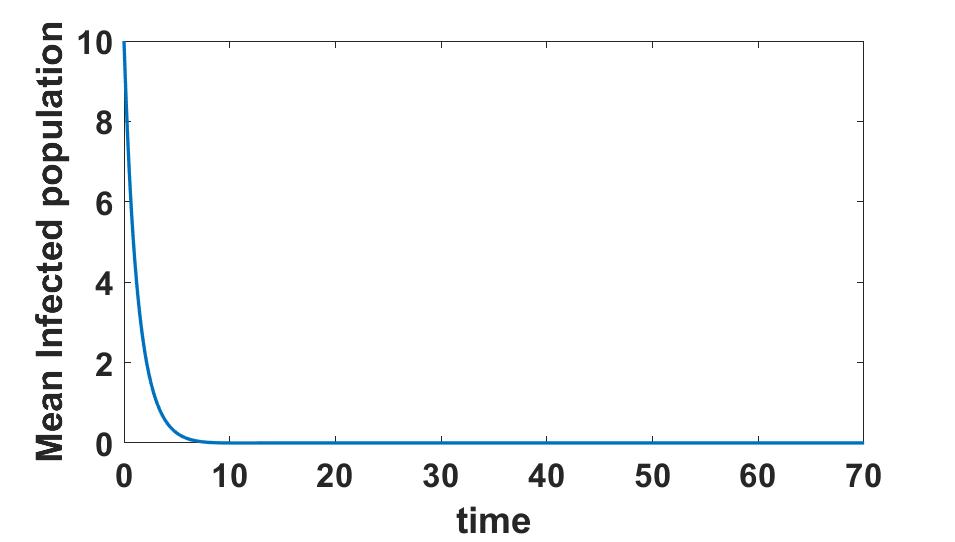}
				\includegraphics[width=3in, height=1.8in, angle=0]{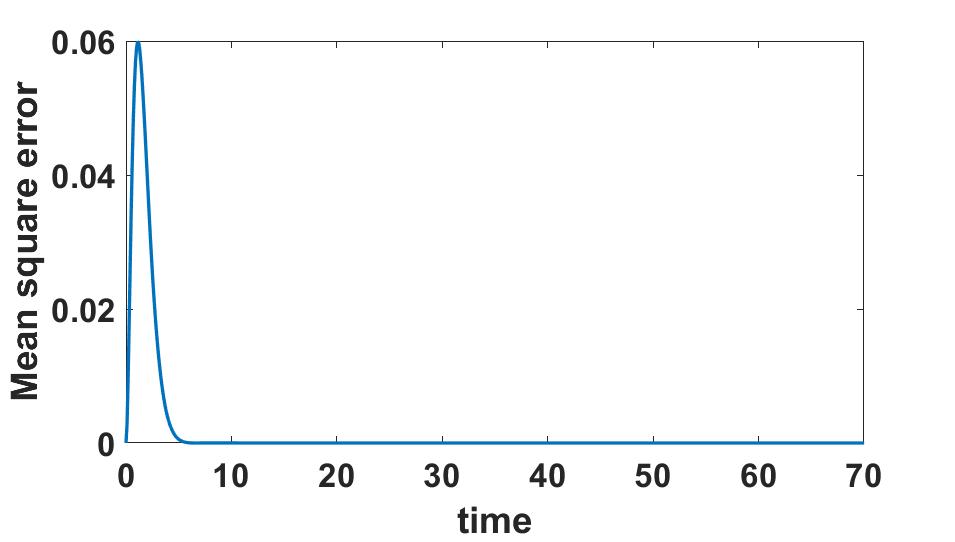}
				\caption{Sensitivity Analysis of $u_{11}$ in Interval III.}
				\label{sen_u11_3}
			\end{center}
		\end{figure}
		
		\begin{figure}[hbt!]
			\begin{center}
				\subcaptionbox*{(a) Interval I}
				{\includegraphics[width=3in, height=1.8in, angle=0]{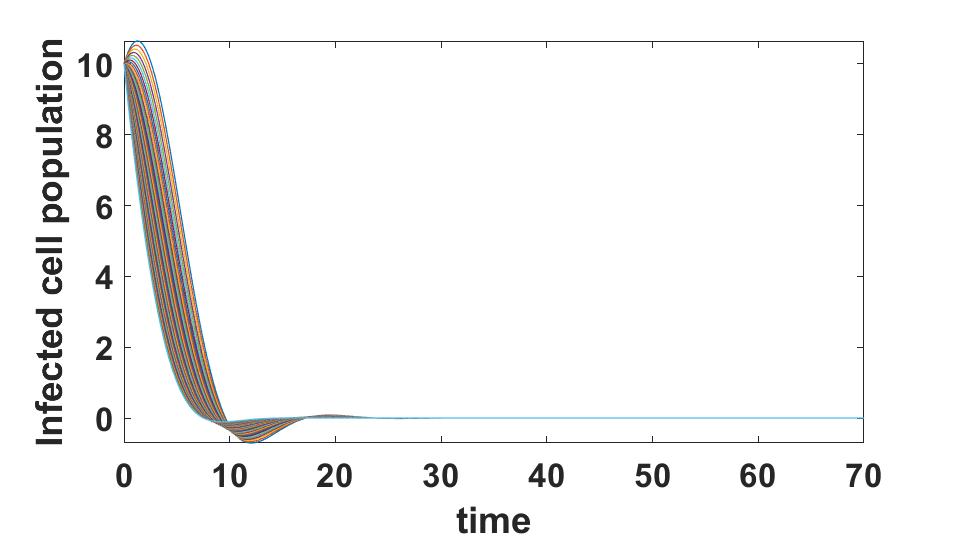}}
			\end{center}
		\end{figure}
		
		\addtocounter{figure}{-1}
		
		\begin{figure}[hbt!]
			\begin{center}
				\subcaptionbox*{(b) Interval II}
				{\includegraphics[width=3in, height=1.8in, angle=0]{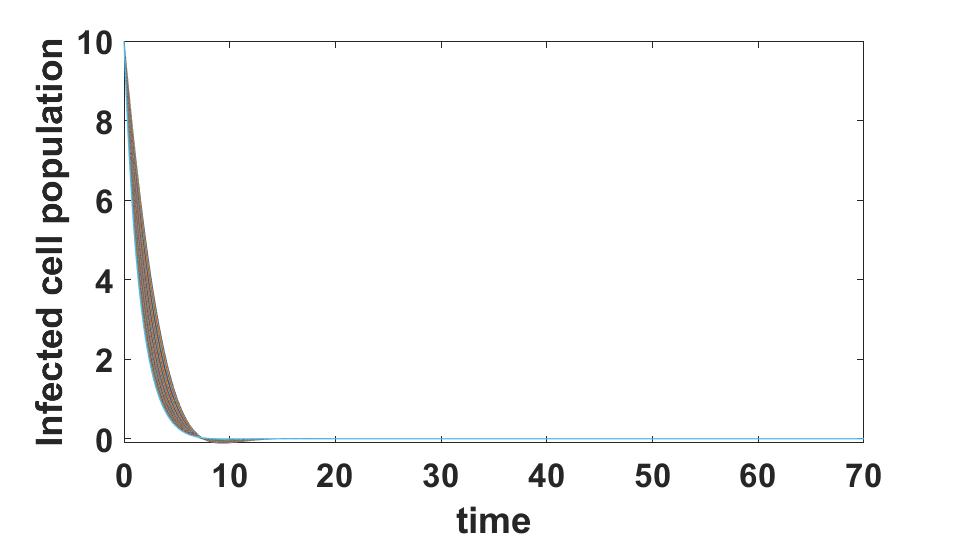}}
				\subcaptionbox*{(c) Interval III}
				{\includegraphics[width=3in, height=1.8in, angle=0]{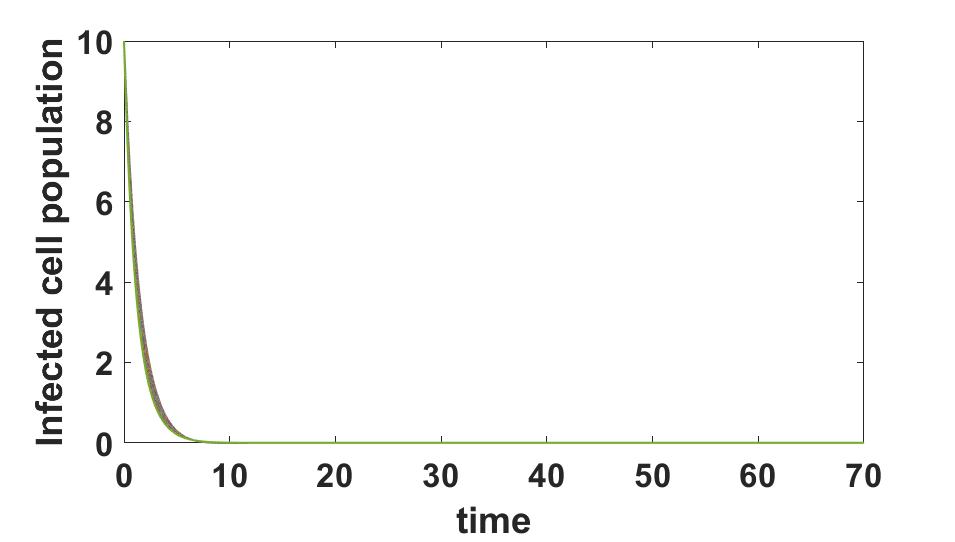}}
				\vspace{0.75\baselineskip}
				\caption{Sensitivity Analysis of $u_{11}$. Infected cell population in different intervals.}
				\label{sen_u11}
			\end{center}
		\end{figure}
		
		\newpage
		\subsection{Parameter $\boldsymbol{\beta}$}
		
		\begin{figure}[hbt!]
			\begin{center}
				\includegraphics[width=3in, height=1.8in, angle=0]{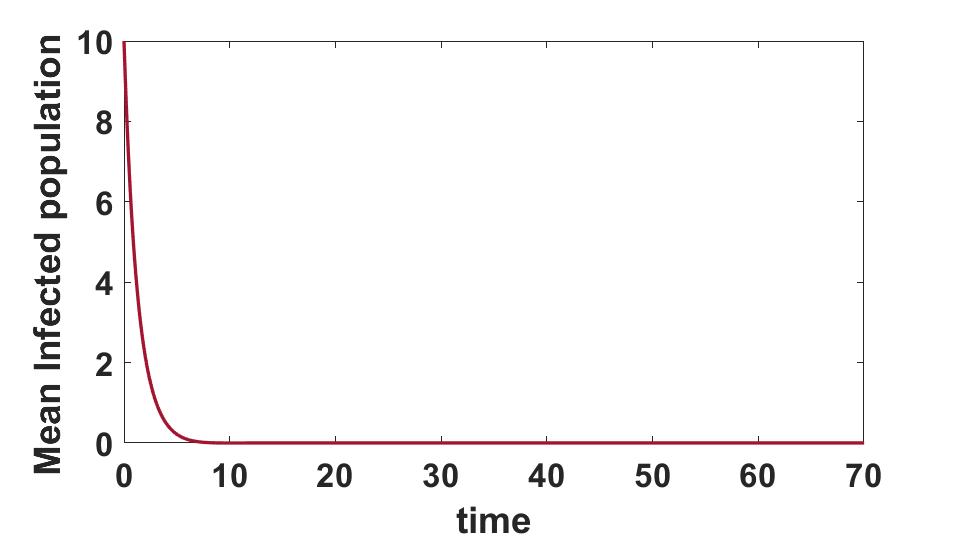}
				\includegraphics[width=3in, height=1.8in, angle=0]{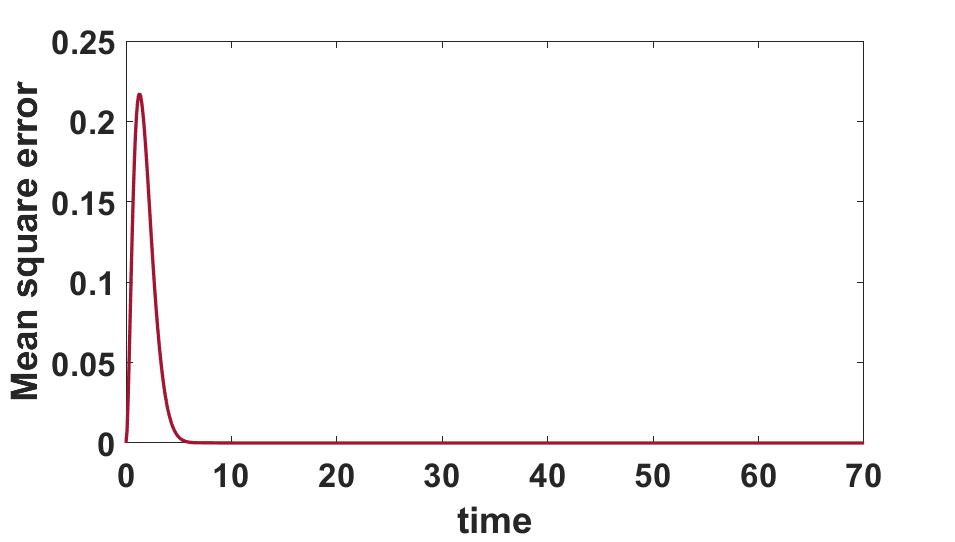}
				\caption{Sensitivity Analysis of $\beta$ in Interval I.}
				\label{sen_beta_1}
			\end{center}
		\end{figure}
		
		\begin{figure}[hbt!]
			
			\begin{center}
				\includegraphics[width=3in, height=1.8in, angle=0]{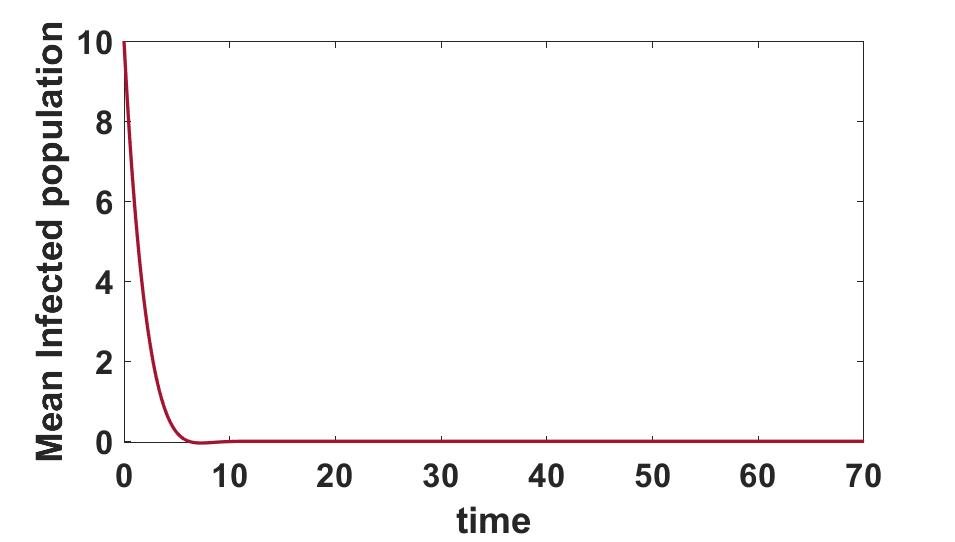}
				\includegraphics[width=3in, height=1.8in, angle=0]{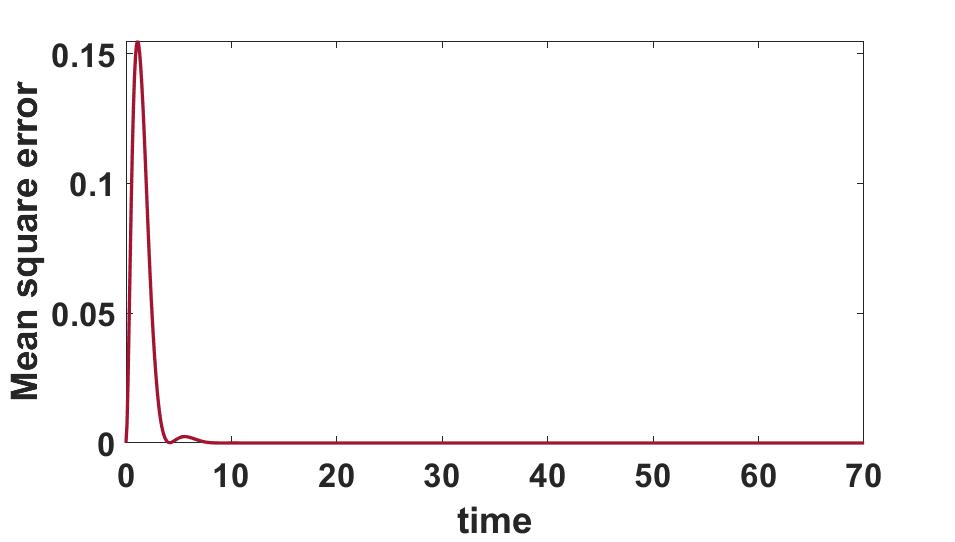}
				\caption{Sensitivity Analysis of $\beta$ in Interval II.}
				\label{sen_beta_2}
			\end{center}
		\end{figure}
		
		\begin{figure}[hbt!]
			\begin{center}
				\includegraphics[width=3in, height=1.8in, angle=0]{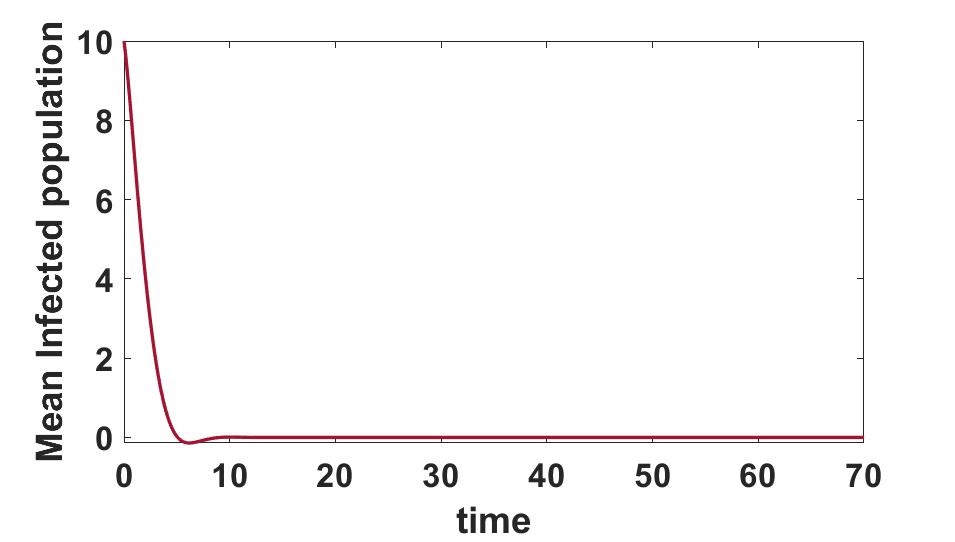}
				\includegraphics[width=3in, height=1.8in, angle=0]{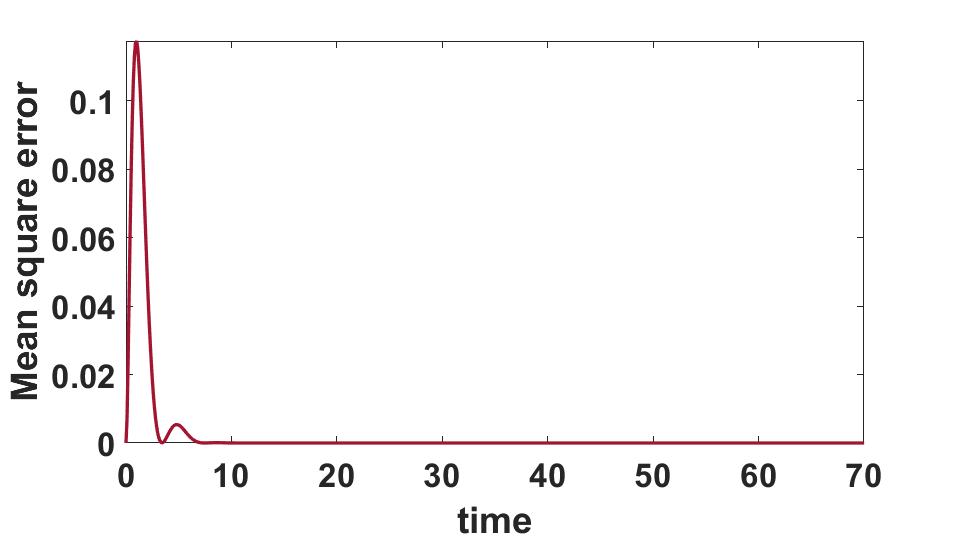}
				\caption{Sensitivity Analysis of $\beta$ in Interval III.}
				\label{sen_beta_3}
			\end{center}
		\end{figure}
		
		\begin{figure}[hbt!]
			\begin{center}
				\subcaptionbox*{(a) Interval I}
				{\includegraphics[width=3in, height=1.8in, angle=0]{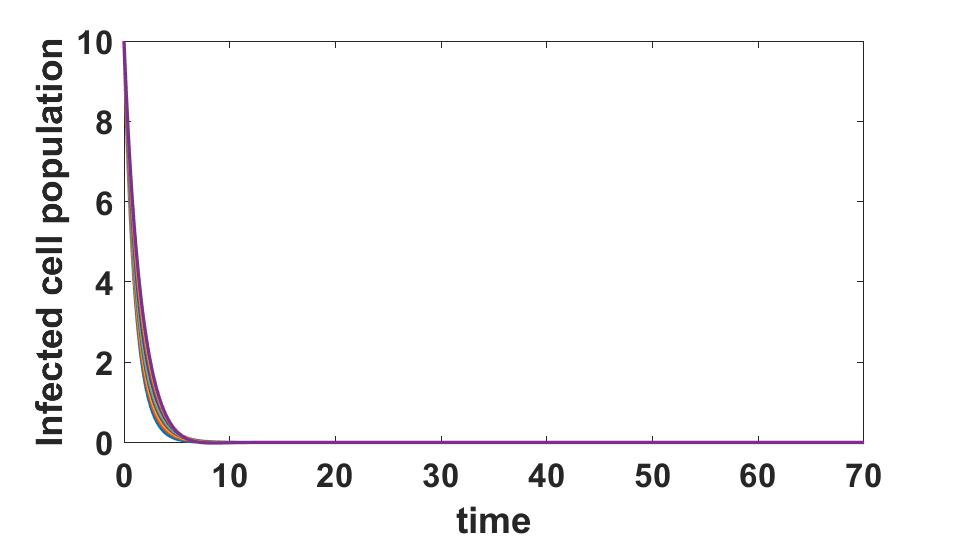}}
			\end{center}
		\end{figure}
		
		\addtocounter{figure}{-1}
		
		\begin{figure}[hbt!]
			\begin{center}
				\subcaptionbox*{(b) Interval II}
				{\includegraphics[width=3in, height=1.8in, angle=0]{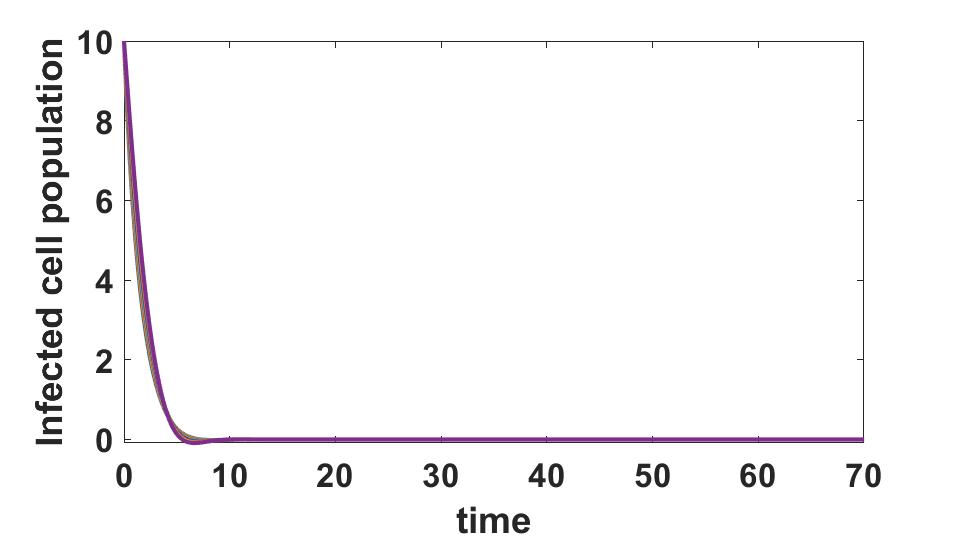}}
				\subcaptionbox*{(c) Interval III}
				{\includegraphics[width=3in, height=1.8in, angle=0]{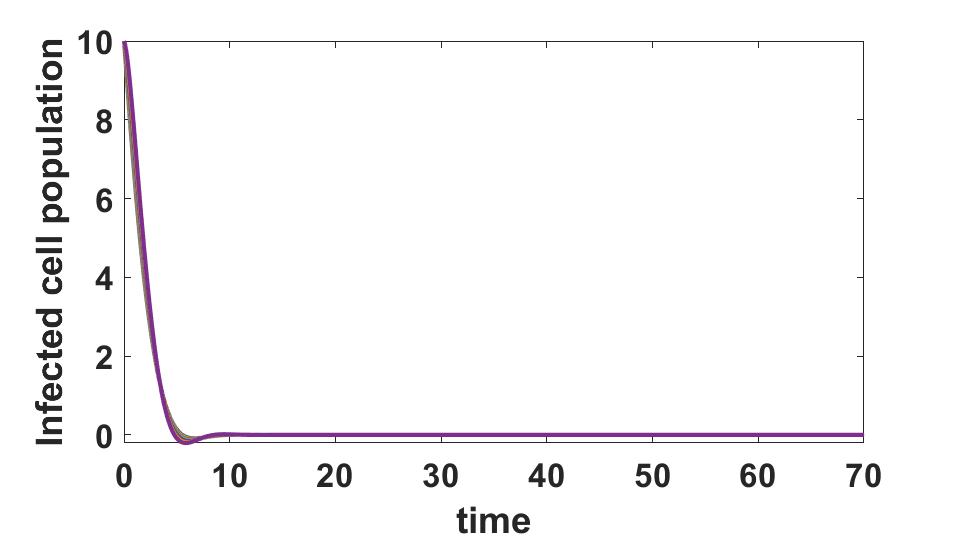}}
				\vspace{0.5\baselineskip}
				\caption{Sensitivity Analysis of $\beta$. Infected cell population in different intervals.}
				\label{sen_beta}
			\end{center}
		\end{figure}
		
		\newpage
		\subsection{Parameter $\boldsymbol{\omega}$}
		
		\begin{figure}[hbt!]
			\begin{center}
				\includegraphics[width=3in, height=1.8in, angle=0]{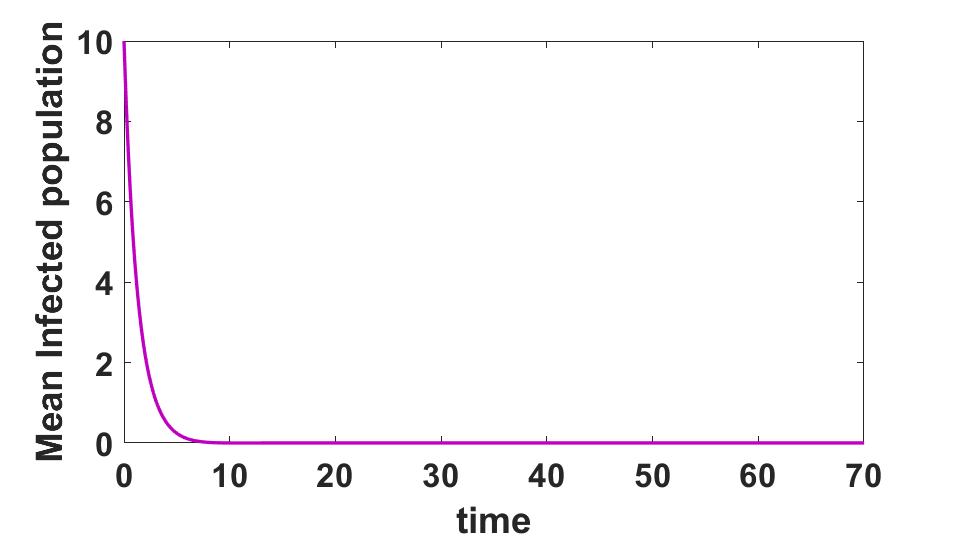}
				\includegraphics[width=3in, height=1.8in, angle=0]{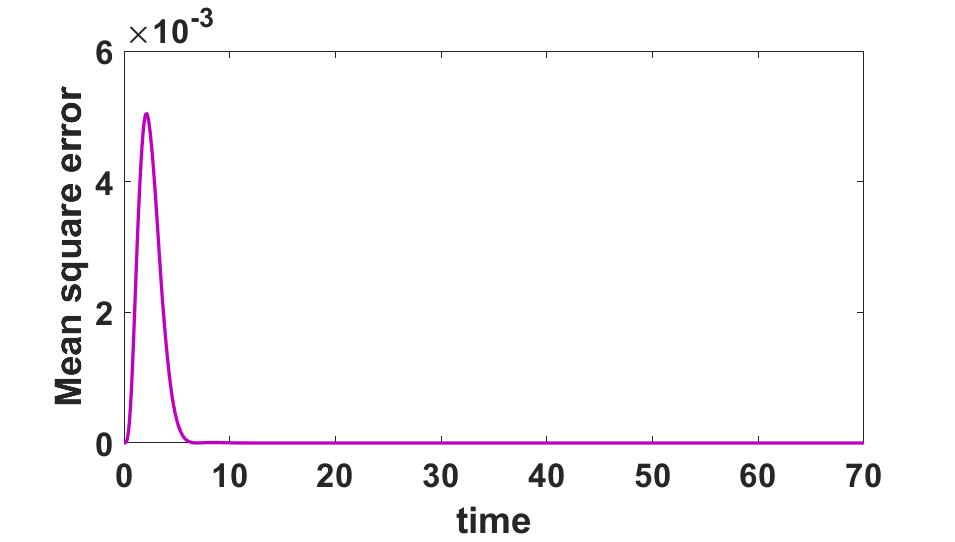}
				\caption{Sensitivity Analysis of $\omega$ in Interval I.}
				\label{sen_omega_1}
			\end{center}
		\end{figure}
		
		\begin{figure}[hbt!]
			
			\begin{center}
				\includegraphics[width=3in, height=1.8in, angle=0]{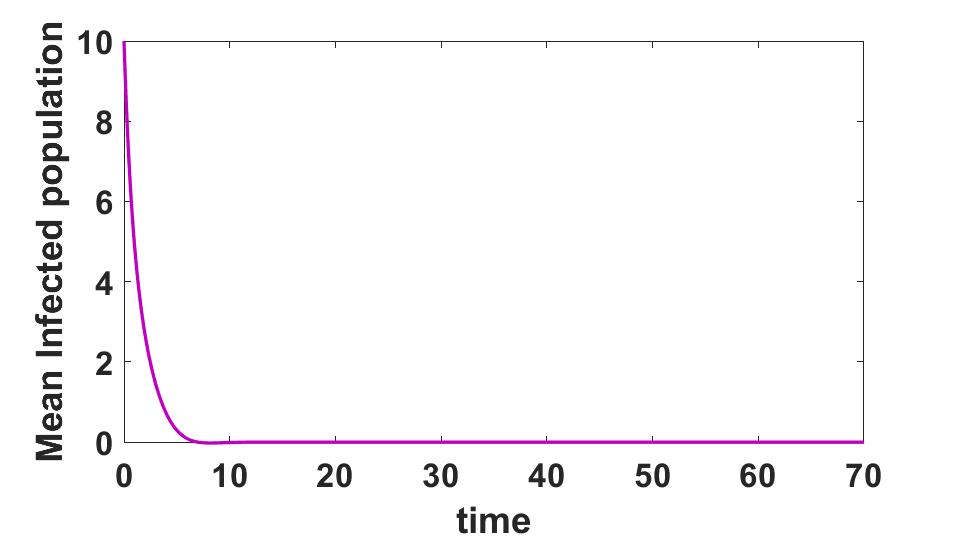}
				\includegraphics[width=3in, height=1.8in, angle=0]{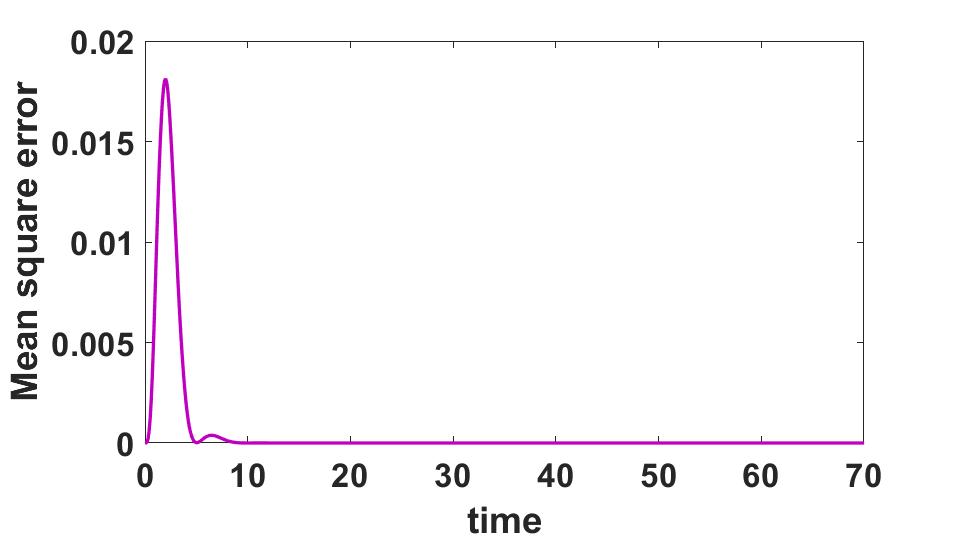}
				\caption{Sensitivity Analysis of $\omega$ in Interval II.}
				\label{sen_omega_2}
			\end{center}
		\end{figure}
		
		\begin{figure}[hbt!]
			\begin{center}
				\includegraphics[width=3in, height=1.8in, angle=0]{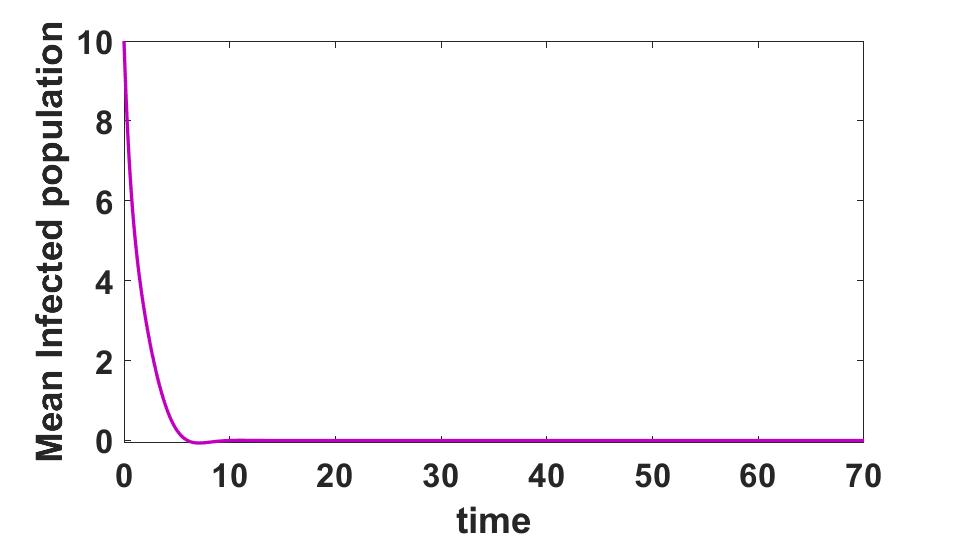}
				\includegraphics[width=3in, height=1.8in, angle=0]{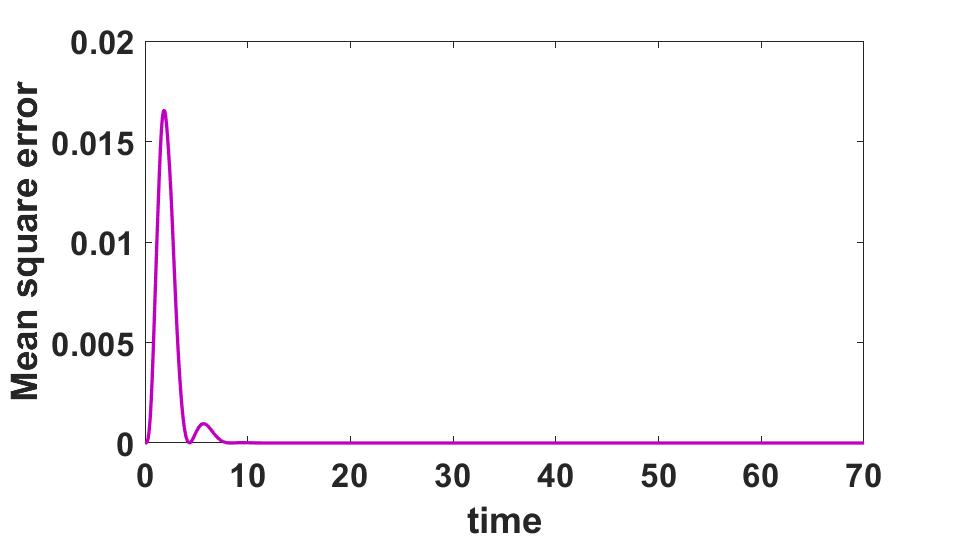}
				\caption{Sensitivity Analysis of $\omega$ in Interval III.}
				\label{sen_omega_3}
			\end{center}
		\end{figure}
		
		\begin{figure}[hbt!]
			\begin{center}
				\subcaptionbox*{(a) Interval I}
				{\includegraphics[width=3in, height=1.8in, angle=0]{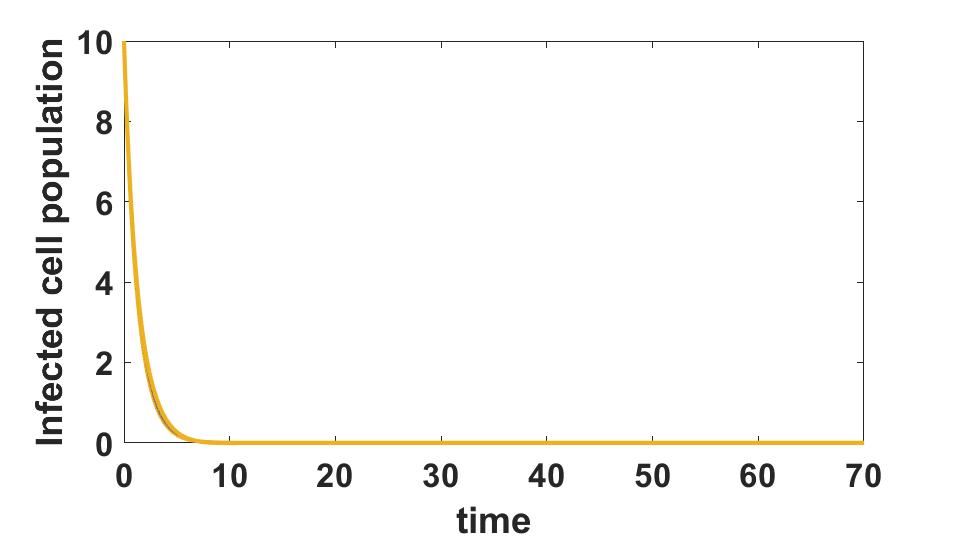}}
			\end{center}
		\end{figure}
		
		\addtocounter{figure}{-1}
		
		\begin{figure}[hbt!]
			\begin{center}
				\subcaptionbox*{(b) Interval II}
				{\includegraphics[width=3in, height=1.8in, angle=0]{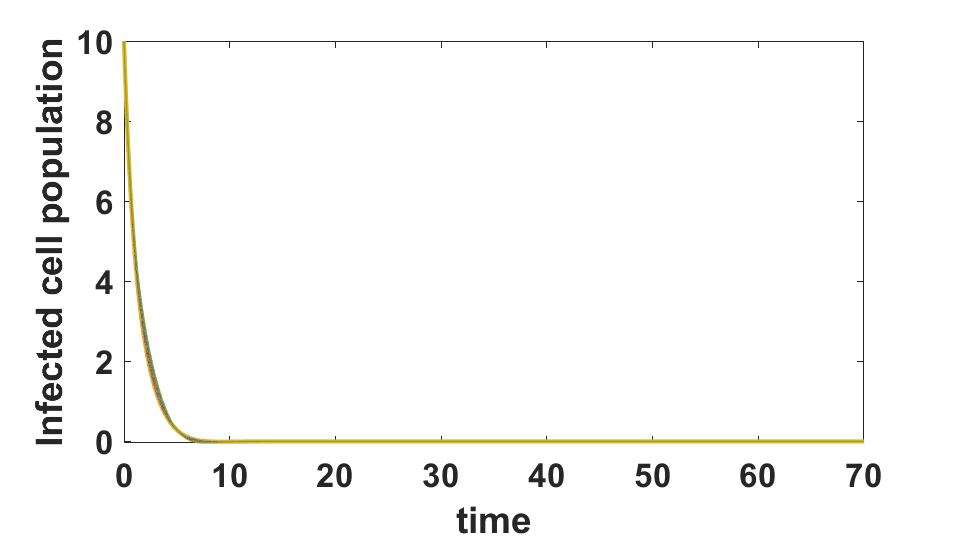}}
				\subcaptionbox*{(c) Interval III}
				{\includegraphics[width=3in, height=1.8in, angle=0]{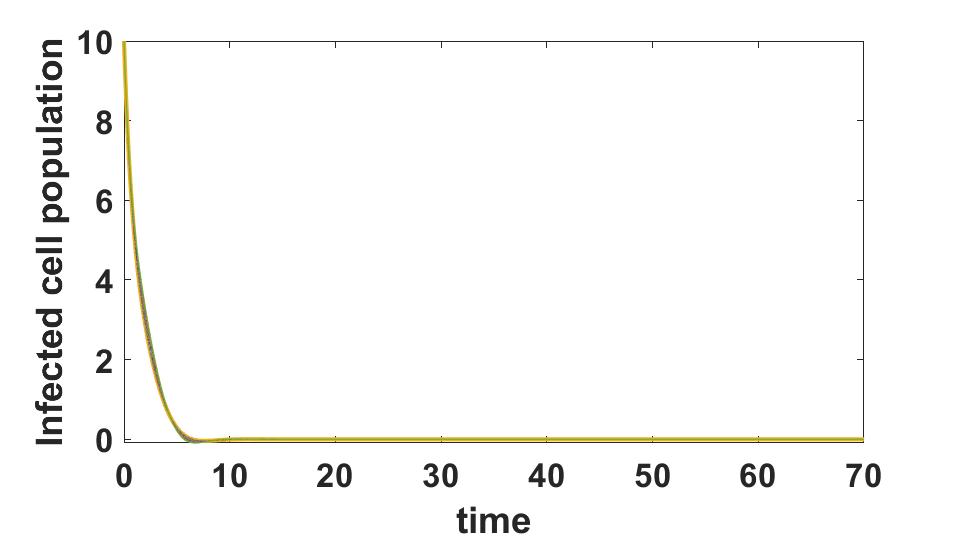}}
				\vspace{0.75\baselineskip}
				\caption{Sensitivity Analysis of $\omega$. Infected cell population in different intervals.}
				\label{sen_omega}
			\end{center}
		\end{figure}
		
		\newpage
		\subsection{Parameter $\boldsymbol{\mu}$}
		
		\begin{figure}[hbt!]
			\begin{center}
				\includegraphics[width=3in, height=1.8in, angle=0]{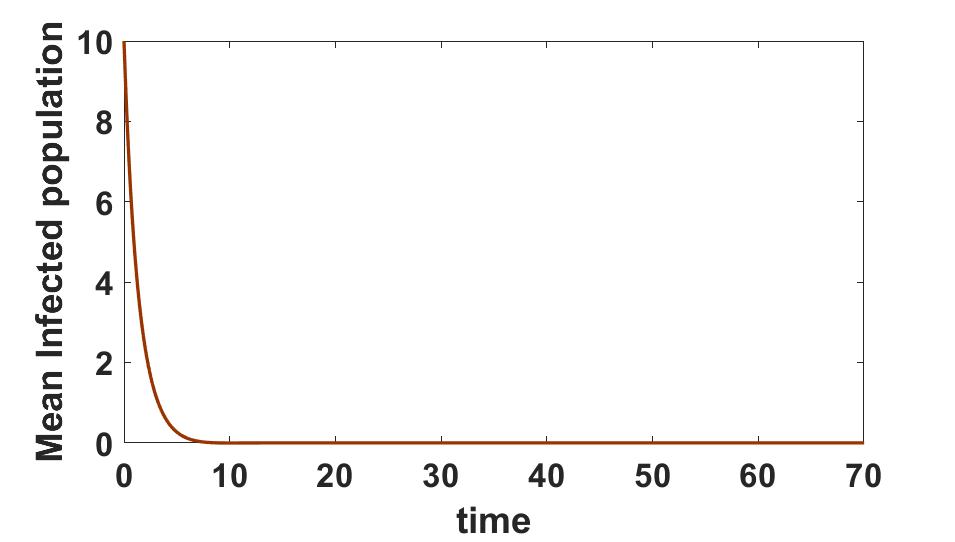}
				\includegraphics[width=3in, height=1.8in, angle=0]{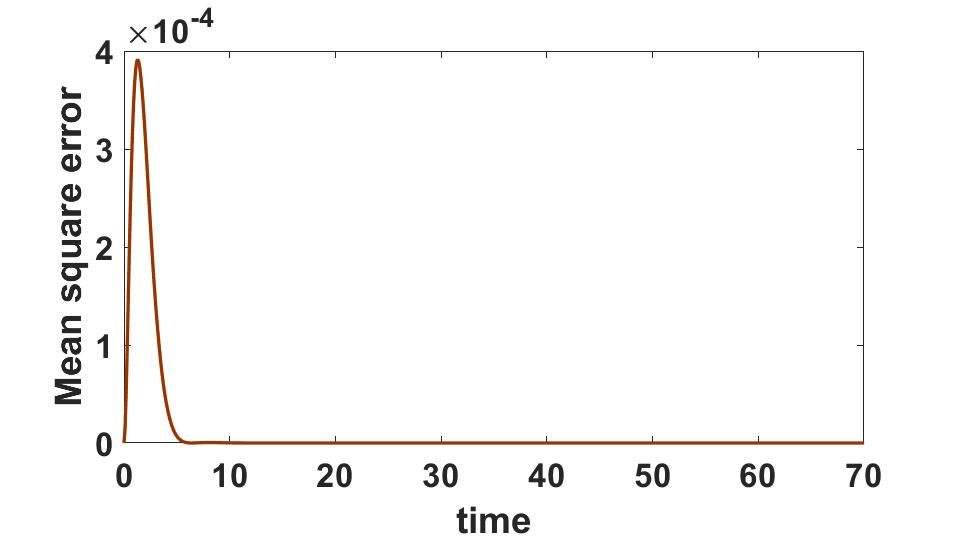}
				\caption{Sensitivity Analysis of $\mu$ in Interval I.}
				\label{sen_mu_1}
			\end{center}
		\end{figure}
		
		\begin{figure}[hbt!]
			
			\begin{center}
				\includegraphics[width=3in, height=1.8in, angle=0]{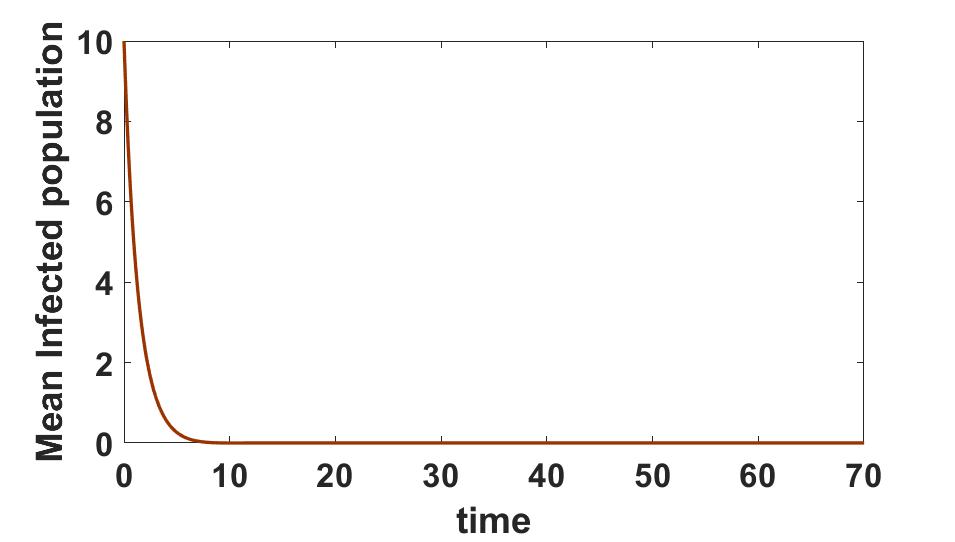}
				\includegraphics[width=3in, height=1.8in, angle=0]{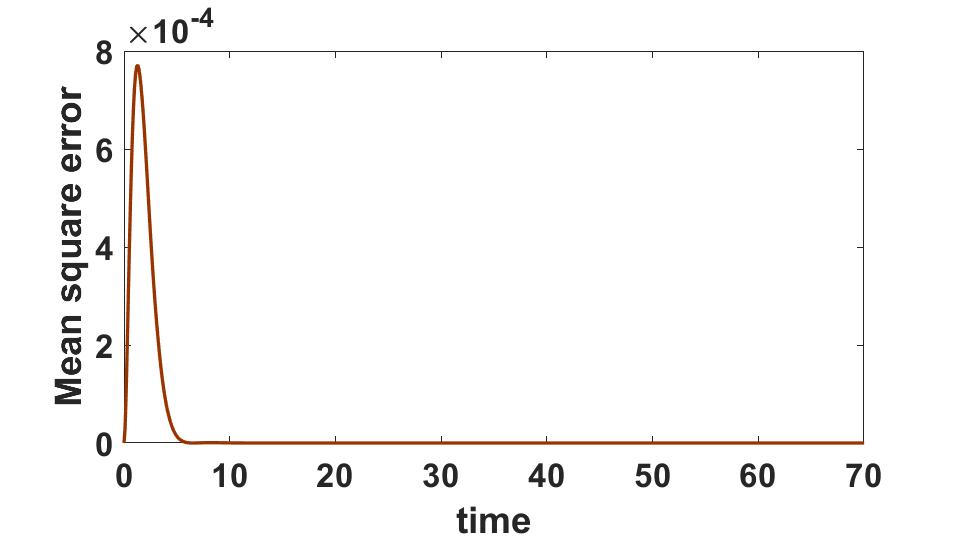}
				\caption{Sensitivity Analysis of $\mu$ in Interval II.}
				\label{sen_mu_2}
			\end{center}
		\end{figure}
		
		\begin{figure}[hbt!]
			\begin{center}
				\includegraphics[width=3in, height=1.8in, angle=0]{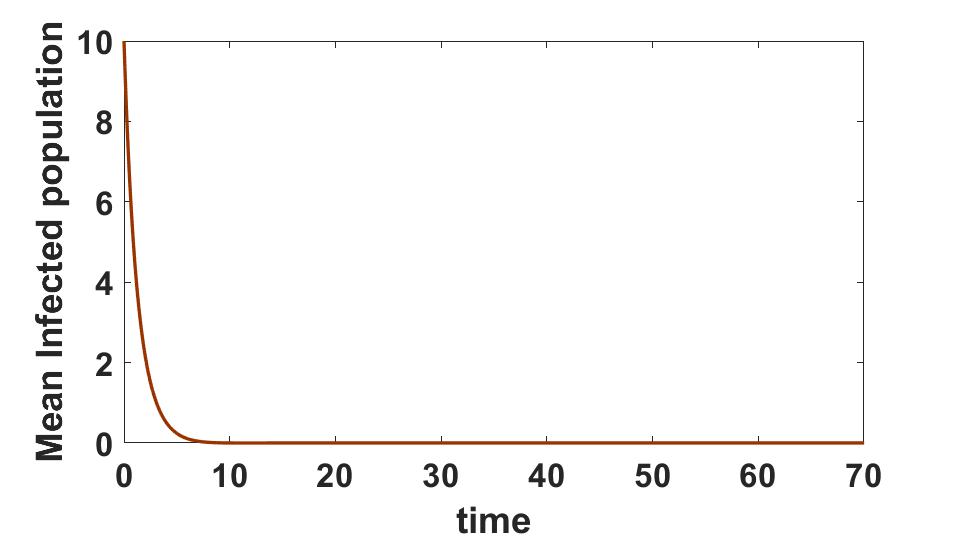}
				\includegraphics[width=3in, height=1.8in, angle=0]{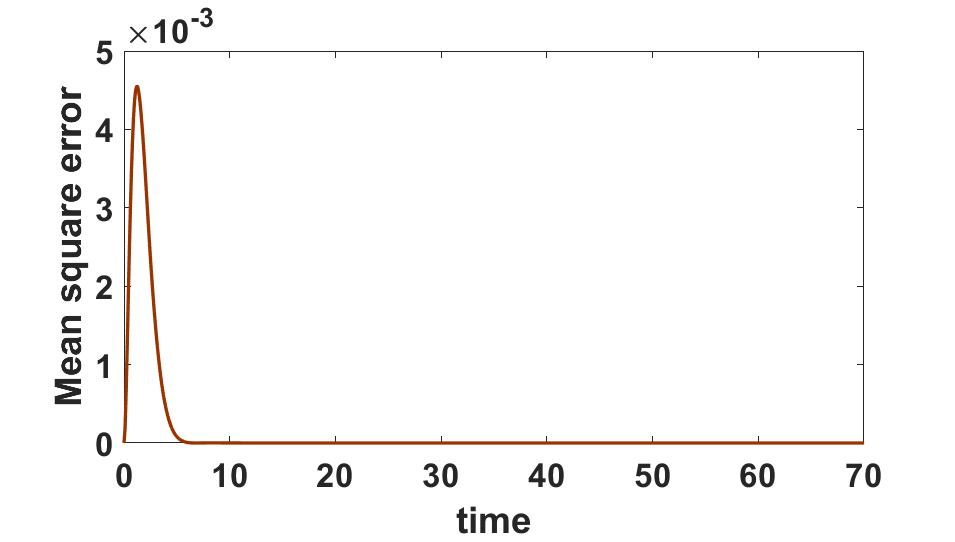}
				\caption{Sensitivity Analysis of $\mu$ in Interval III.}
				\label{sen_mu_3}
			\end{center}
		\end{figure}
		
		\begin{figure}[hbt!]
			\begin{center}
				\subcaptionbox*{(a) Interval I}
				{\includegraphics[width=3in, height=1.8in, angle=0]{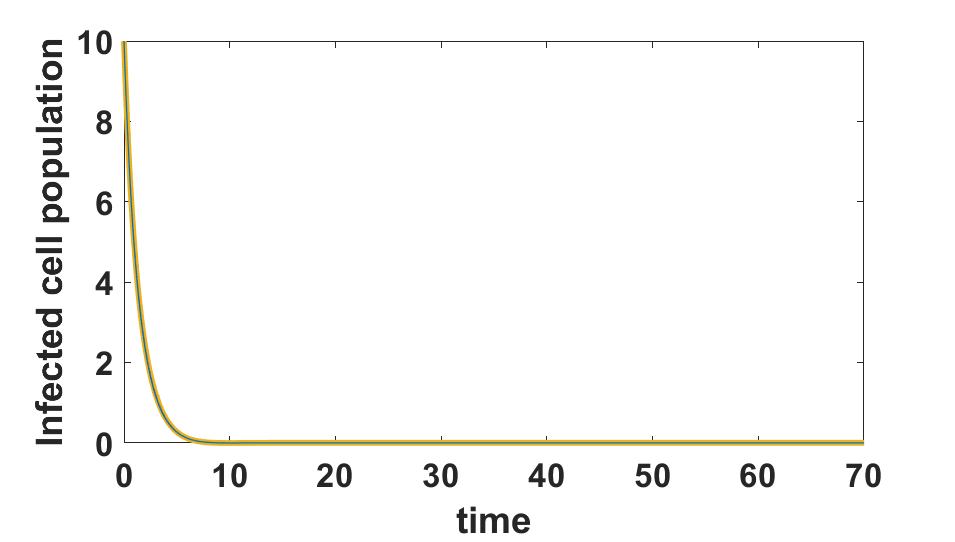}}
			\end{center}
		\end{figure}
		
		\addtocounter{figure}{-1}
		
		\begin{figure}[hbt!]
			\begin{center}
				\subcaptionbox*{(b) Interval II}
				{\includegraphics[width=3in, height=1.8in, angle=0]{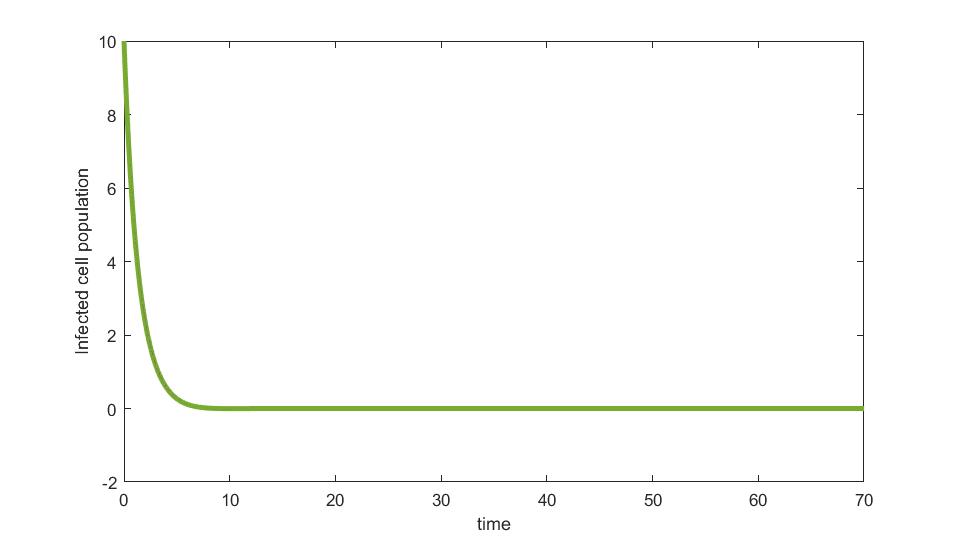}}
				\subcaptionbox*{(c) Interval III}
				{\includegraphics[width=3in, height=1.8in, angle=0]{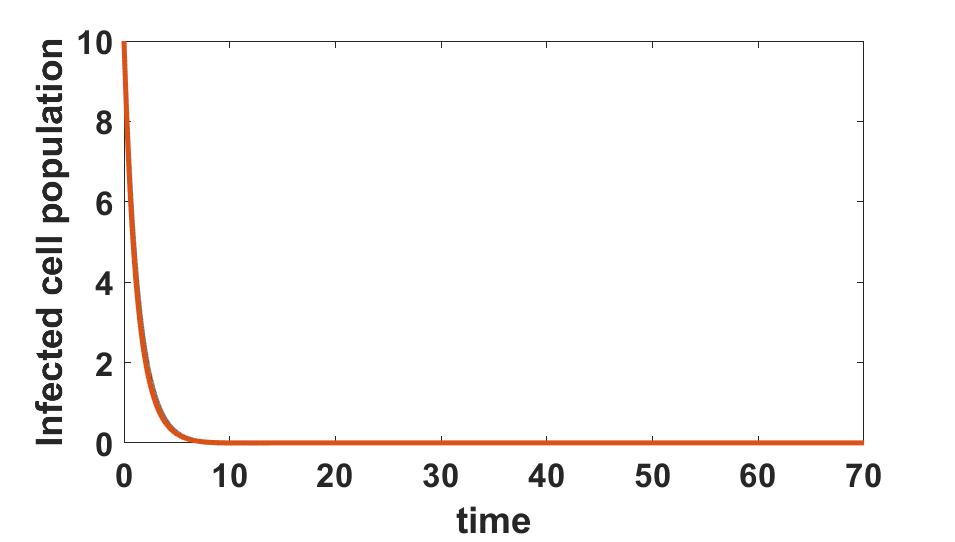}}
				\vspace{0.75\baselineskip}
				\caption{Sensitivity Analysis of $\mu$. Infected cell population in different intervals.}
				\label{sen_mu}
			\end{center}
		\end{figure}
		

\end{document}